\documentclass[12pt]{article}
\usepackage{t1enc}
\usepackage{amsmath}
\usepackage{amsthm}
\usepackage{latexsym}
\usepackage{natbib}
\usepackage{epsfig}
\usepackage{textcomp}
\usepackage{alltt}
\usepackage{graphicx}
\usepackage{rotating}
\usepackage{dcolumn}
\usepackage{enumitem}
\usepackage{amsfonts}
\usepackage{algorithm}
\usepackage{amssymb}
\usepackage{bbm, dsfont}
\usepackage[svgnames]{xcolor}
\usepackage{mathtools}
\usepackage{listings}
\usepackage{prodint}

\usepackage{pdflscape}

\usepackage{tikz}
\tikzstyle{ov}=[shape=rectangle,
                draw=black!50,
                thick,
                minimum width=0.7cm,
                minimum height=0.7cm]

\tikzstyle{av}=[shape=rectangle,
                draw=black!50,
                fill=black!10,
                thick,
                minimum width=0.7cm,
                minimum height=0.7cm]

\tikzstyle{lv}=[shape=circle,draw=black!50,thick]

\usetikzlibrary{shapes,calendar,matrix,backgrounds,folding,snakes}
\allowdisplaybreaks

\DeclareMathOperator*{\argmin}{arg\,min}

\begin{document}
\newcommand{\cip}{\perp\!\!\!\!\perp}
\newcommand{\nothere}[1]{}
\newcommand{\noi}{\noindent}
\newcommand{\mbf}[1]{\mbox{\boldmath $#1$}}
\newcommand{\cond}{\, |\,}
\newcommand{\hO}[2]{{\cal O}_{#1}^{#2}}
\newcommand{\hF}[2]{{\cal F}_{#1}^{#2}}
\newcommand{\tl}[1]{\tilde{\lambda}_{#1}^T}
\newcommand{\la}[2]{\lambda_{#1}^T(Z^{#2})}
\newcommand{\I}[1]{1_{(#1)}}
\newcommand{\cd}{\mbox{$\stackrel{\mbox{\tiny{\cal D}}}{\rightarrow}$}}
\newcommand{\cp}{\mbox{$\stackrel{\mbox{\tiny{p}}}{\rightarrow}$}}
\newcommand{\cas}{\mbox{$\stackrel{\mbox{\tiny{a.s.}}}{\rightarrow}$}}
\newcommand{\ld}{\mbox{$\; \stackrel{\mbox{\tiny{def}}}{=3D} \; $}}
\newcommand{\nk}{\mbox{$n \rightarrow \infty$}}
\newcommand{\con}{\mbox{$\rightarrow $}}
\newcommand{\dprime}{\mbox{$\prime \vspace{-1 mm} \prime$}}
\newcommand{\Borel}{\mbox{${\cal B}$}}
\newcommand{\bevis}{\mbox{$\underline{\em{Proof}}$}}
\newcommand{\Rd}[1]{\mbox{${\Re^{#1}}$}}
\newcommand{\il}[1]{{\int_{0}^{#1}}}
\newcommand{\pl}[1]{\mbox{\bf {\LARGE #1}}}
\newcommand{\expit}{\text{expit}}
\newcommand{\indep}{\rotatebox[origin=c]{90}{$\models$}}
\newcommand{\blind}{1}
\newcommand{\pr}{\text{pr}}
\newcommand{\var}{\text{var}}
\newcommand{\cov}{\text{cov}}
\newcommand{\Bin}{\text{Bin}}
\newcommand{\Exp}{\text{Exp}}
\newcommand{\unif}{\text{unif}}
\newcommand{\logit}{\text{logit}}
\newcommand{\sign}{\text{sign}}
\newcommand{\support}{\text{support}}
\newcommand\norm[1]{\left\lVert#1\right\rVert}
\newcommand{\diag}{\text{diag}}

\newtheorem{theorem}{Theorem}
\newtheorem{lemma}{Lemma}
\newtheorem{prop}{Proposition}
\newtheorem{assumption}{Assumption}
\newtheorem{definition}{Definition}
\newtheorem*{remark}{Remark}
\newtheorem{corollary}{Corollary}
\newtheorem{example}{Example}
\newtheorem{condition}{Condition}

\parindent12pt

\begin{center}{\Large{On Weighted Orthogonal Learners \\ for Heterogeneous Treatment Effects}}
 \end{center}
 
{ 
\begin{center}
Pawe{\l} Morzywo{\l}ek$^{1,2}$, Johan Decruyenaere$^{3}$, Stijn Vansteelandt$^{1}$ \\
\bigskip
\scriptsize{$^1$Department of Applied Mathematics, Computer Science and Statistics, Ghent University} \\
\scriptsize{$^2$Department of Statistics, University of Washington} \\
\scriptsize{$^3$Department of Intensive Care Medicine, Ghent University Hospital} 
\end{center}
}
\smallskip
\begin{center}
\today
\end{center}
\smallskip

\setlength{\parindent}{0.3in} \setlength{\baselineskip}{24pt}
\begin{abstract}
Motivated by applications in personalized medicine and individualized policy-making, there is a growing interest in techniques for quantifying treatment effect heterogeneity in terms of the conditional average treatment effect (CATE). Some of the most prominent methods for CATE estimation developed in recent years are T-Learner, DR-Learner and R-Learner. The latter two were designed to improve on the former by being Neyman-orthogonal. However, the relations between them remain unclear, and likewise the literature remains vague on whether these learners converge to a useful quantity or (functional) estimand when the underlying optimization procedure is restricted to a class of functions that does not include the CATE. In this article, we provide insight into these questions by discussing DR-Learner and R-Learner as special cases of a general class of weighted Neyman-orthogonal learners for the CATE, for which we moreover derive oracle bounds. Our results shed light on how one may construct Neyman-orthogonal learners with desirable properties, on when DR-Learner may be preferred over R-Learner (and vice versa), and on novel learners that may sometimes be preferable to either of these. Theoretical findings are confirmed using results from simulation studies on synthetic data, as well as an application in critical care medicine.
\end{abstract}
\noi
{\it Causal Inference, Heterogeneous Treatment Effects, Orthogonal Statistical Learning, \\ Precision Medicine}

\section{Introduction} \label{Section1}

Data-driven decision support systems hold a promise to revolutionize decision-making in various areas ranging from medicine to policy-making. With the increasing amount of available data, it becomes clear that the old paradigm of "one-size-fits-all" is no longer satisfactory and modern decision-support systems need to recognize heterogeneity in the data to allow personalized decision-making \citep{Athey2017, Kosorok2019}.  Recognition of this emerging field has led causal inference and computer science research communities to develop techniques for heterogeneous treatment effects (or conditional average treatment effect - CATE) estimation, which has been applied in medicine and the social sciences \citep{vanderLaan2014, Imai2013, Athey2016, Luedtke2016a, Kunzel2019, Curth2020, Nie2021, Kennedy2023, Foster2023}. 
\\
\indent
The vast majority of the causal inference literature focuses on the estimation of the average treatment effect, averaged across the whole population.  In contrast, the CATE expresses the treatment effect as a function of a set of observable features. This has a number of advantages. For example, it allows the identification of sub-populations in which the intervention is particularly beneficial or harmful. Furthermore, it is often believed to be better transportable between studies \citep{Dahabreh2020}. However, CATE estimation is more ambitious than the classical problem of estimating the average treatment effect in the whole population. Whereas the latter summarizes the treatment effect into a single number, the former is an infinite-dimensional response curve in function of covariate values. 
\\
\indent
Some of the most prominent methods for CATE estimation developed in recent years are T-Learner \citep{Kunzel2019}, DR-Learner \citep{Kennedy2023} and R-Learner \citep{Nie2021}. These different learners have been developed from very different perspectives, making the literature opaque on what are the precise relations between these learners. While they all intend to deliver estimates of the CATE, it is for instance unclear what they deliver when these distinct learners are optimized over a function class that does not contain the CATE.  
\\
\indent
In this work, we provide insight by developing a class of weighted Neyman-orthogonal CATE learners obtained by minimization of well-chosen loss functions that quantify squared bias in the CATE learner, and which differ in how this squared bias is averaged over the study population. This reproduces popular examples from the literature on heterogeneous treatment effects, e.g. DR-Learner and R-Learner. Viewing different methods through the lens of a single framework provides insight into their differences, but also leads to novel learners based on Neyman-orthogonal loss functions \citep{Chernozhukov2018a, Foster2023}. This Neyman-orthogonality is essential as it insulates the learners to some extent against slow convergence of the estimated nuisance parameters that index the loss function on which they are based. This renders the resulting estimators for the CATE, obtained upon substituting these nuisance parameters by (sufficiently fast converging) consistent plug-in estimators, asymptotically equivalent to the oracle estimator, which uses the known nuisance parameters. This allows mitigation of the regularization bias resulting from using flexible modeling techniques for the estimation of nuisance parameters.
\\
\indent
The paper is structured as follows. In Section \ref{Section2}, we describe the problem setting and introduce key concepts related to heterogeneous treatment effects estimation. In Section \ref{Section3}, we describe a general approach to construct Neyman-orthogonal loss functions for the CATE, building on the literature of semiparametric inference \citep{Pfanzagl1990, Bickel1998, vanderLaan2003, Tsiatis2006, Kosorok2008, vanderLaan2011, Chernozhukov2018a, Hines2022, Foster2023}. This gives rise to a general class of weighted orthogonal learners for heterogeneous treatment effects estimation. In Section \ref{Section4}, we present an empirical example from critical care medicine. In Section \ref{Section5}, we present a simulation study comparing different methods. We conclude with a discussion of the results in Section \ref{Section6}.

\section{Problem Setting and Notation} \label{Section2}

We consider $n$ independent and identically distributed observations $Z_i \coloneqq \left( Y_i, A_i, X_i \right)$, for $i = 1, \dots,n$, where $X_i \in \mathcal{X} \subseteq \mathbb{R}^{d_X}$ is minimal set of covariates sufficient for confounding adjustment of a binary treatment $A_i \in \left\lbrace 0, 1 \right\rbrace$ on outcome $Y_i \in \mathbb{R}$. Based on this, our aim is to estimate the conditional average treatment effect (CATE) given by 
\begin{equation*} 
\tau \left( v \right) \coloneqq \mathbb{E} \left( Y^1 - Y^0 \cond V = v \right).
\end{equation*}
Here $Y^a$ is the potential outcome under treatment $A = a$, for $a = 0,1$ \citep{Neyman1923, Rubin1974,Robins1986, Hernan2020} and $V \in \mathcal{V} \subseteq \mathbb{R}^{d_V}$ is a subset of the features $X$, i.e. $d_V \leq d_X$. The setting where $V$ is a strict subset of the features $X$, i.e. $d_V < d_X$, is sometimes referred to as "runtime confounding" \citep{Coston2020}.
\\
\indent
The CATE is identifiable under assumptions usually invoked in the causal inference literature:
\begin{assumption} \label{Assump1}
Conditional exchangeability, i.e. $Y_i^a \cip A_i \cond X_i$ for $a = 0,1$.
\end{assumption}
\noindent
Conditional exchangeability, sometimes referred to as the no unmeasured confounding assumption, states that the set of collected features is sufficient for confounding adjustment. 
\begin{assumption} \label{Assump2}
Positivity, i.e. $0 < \pi_0 \left( X_i \right) < 1$, where $\pi_0 \left( X_i \right) \coloneqq \mathbb{P} \left( A_i = 1 \cond X_i\right)$ is a propensity score. 
\end{assumption}
\noindent
The positivity assumption guarantees that for each covariate stratum $X = x$ we have a positive probability of observing individuals receiving each of the treatment options. This is necessary because to infer the effect of a treatment we need to compare the outcome of patients with similar characteristics but different treatment assignments. 
\begin{assumption} \label{Assump3}
Consistency, i.e. $Y_i^a = Y_i$ for $A_i = a$.
\end{assumption}
\noindent
The consistency assumption states that versions of the interventions observed in the data correspond to well-defined interventions \citep{VanderWeele2013}.
\\
\indent
To estimate the conditional average treatment effect we set out to find a function $g \in \mathcal{G}$, where $\mathcal{G}$ is a function space with a norm $\norm{ \cdot }_\mathcal{G}$, that is the best approximation of the CATE in terms of the mean-squared error. In what follows, we will consider $\mathcal{G}$ to be $L_2$-space, i.e. $\mathcal{G} = \left\lbrace g \cond \norm{ g }_p < \infty \right\rbrace$ with $\norm{ g }_p \coloneqq \left[ \mathbb{E} \left\lbrace \lvert g \left( V \right) \rvert^p \right\rbrace \right]^{1/p}$ being the $L_p$-norm, for $p=2$. For instance, we may consider minimizing the following population risk function
\begin{align} \label{eqn1}
\mathcal{L} \left( g \right) &= \mathbb{E} \left[ \left\lbrace \left( Y^1 - Y^0 \right) - g \left( V \right) \right\rbrace^2 \right] \\
\nonumber
&= \mathbb{E} \left[ \left\lbrace \left( Y^1 - Y^0 \right) - \mathbb{E} \left( Y^1 - Y^0 \cond V \right) \right\rbrace^2 \right] \\
\nonumber
&+ \mathbb{E} \left[ \left\lbrace \mathbb{E} \left( Y^1 - Y^0 \cond V \right) - g \left( V \right) \right\rbrace^2 \right]
\end{align}
where the second equality follows using the law of iterated expectations. We denote the minimizer of the population risk function (\ref{eqn1}) by
\begin{align*}
g_0 \coloneqq \argmin_{g \in \mathcal{G}} \mathbb{E} \left[ \left\lbrace \left( Y^1 - Y^0 \right) - g \left( V \right) \right\rbrace^2 \right].
\end{align*}
Note that if $\mathbb{E} \left( Y^1 - Y^0 \cond V \right) \in \mathcal{G}$, then by the second equality in (\ref{eqn1}) we have $g_0 \left( V \right) = \mathbb{E} \left( Y^1 - Y^0 \cond V \right)$.
\\
\indent
The population risk (\ref{eqn1}) cannot be directly used for estimation as it relies on knowledge of both potential outcomes $Y^0$ and $Y^1$, which are not available to the data analyst who observes only one of the potential outcomes for each individual. To circumvent this problem one usually invokes the identifiability assumptions of conditional exchangeability, positivity and consistency to express the population risk (\ref{eqn1}) in terms of the observed data. This leads to a population risk function $\mathcal{L} \left( g, \eta \right)$, which additionally to the parameter of interest $g$ depends on so-called nuisance parameters $\eta$ (e.g., the infinite-dimensional propensity score and/or conditional mean outcome). We illustrate this in the following sections.  

\subsection{Nonparametric outcome regression} \label{subsect4}

When $V$ is sufficient for the confounding adjustment, i.e. $V = X$, under the above-stated identifiability assumptions of conditional exchangeability, positivity and consistency, we can express the CATE as $\tau \left( x \right) = Q_0^{\left( 1 \right)} \left( x \right) - Q_0^{\left( 0 \right)} \left( x \right)$, where $Q_0^{\left( a \right)} \left( x \right) \coloneqq \mathbb{E} \left( Y \cond X = x, A = a \right)$. Note that, if $\mathbb{E} \left( Y^1 - Y^0 \cond X \right) \in \mathcal{G}$, then $g_0 \left( x \right) = Q_0^{\left( 1 \right)} \left( x \right) - Q_0^{\left( 0 \right)} \left( x \right)$. This suggests the following approach to CATE estimation, which treats it as a classical prediction problem. First, construct estimators $\hat{Q}^{\left( a \right)} \left( x \right)$, for $a = 0, 1$, through fitting nonparametric outcome prediction models separately in the treated and untreated populations. Once we have $\hat{Q}^{\left( 0 \right)} \left( x \right)$ and $\hat{Q}^{\left( 1 \right)} \left( x \right)$ we subtract them to obtain a CATE estimator
\begin{align*}
\hat{\tau} \left( x \right) = \hat{Q}^{\left( 1 \right)} \left( x \right) - \hat{Q}^{\left( 0 \right)}\left( x \right).
\end{align*}
\cite{Kunzel2019} refer to this approach as the T-Learner. While simple, it has several limitations. 
\\
\indent
Firstly, direct minimization of the population risk function (\ref{eqn1}) enables an optimal bias-variance trade-off to be made for CATE estimation. In contrast, T-Learner does not target the CATE directly, but instead consists of two separate minimization problems resulting in two, possibly regularized, outcome predictions. Each of them trades bias for variance in a single outcome regression problem, but this does not guarantee that also the difference between those predictions makes an optimal bias-variance trade-off relative to the CATE. Moreover, T-Learner is unable to exploit a priori knowledge that for instance the CATE is an element of a "simple" class $\mathcal{G}$. Even if we had such knowledge (as indeed we may often believe treatment effects to lack heterogeneity w.r.t. certain features), then the use of regularization separately on each outcome regression merely helps to produce "simple" functions $\hat{Q}^{\left( 0 \right)} \left( x \right)$ and $\hat{Q}^{\left( 1 \right)} \left( x \right)$; the difference $\hat{Q}^{\left( 1 \right)} \left( x \right) - \hat{Q}^{\left( 0 \right)} \left( x \right)$ of such "simple" functions is nonetheless often much more "complex".  
\\
\indent
Secondly note that the outcome regression $\hat{Q}^{\left( 0 \right)} \left( x \right)$ is fitted in the untreated population and $\hat{Q}^{\left( 1 \right)} \left( x \right)$ is fitted in the treated population. However, to have a well-performing T-Learner one needs the outcome regressions $\hat{Q}^{\left( 0 \right)} \left( x \right)$ and $\hat{Q}^{\left( 1 \right)} \left( x \right)$ to model the response surfaces well in the whole population, i.e. over the whole covariate space. It is often the case that exposed and unexposed groups of individuals have a different distribution of the covariates $X$, i.e. they are concentrated in different regions of the covariate space, which is referred to as covariate shift.  In this case, the outcome regression $\hat{Q}^{\left( 0 \right)} \left( x \right)$, by relying strongly on the parts of the covariate space with many unexposed individuals, may over-smooth and extrapolate in the parts of covariate space with most treated individuals, and vice versa for $\hat{Q}^{\left( 1 \right)} \left( x \right)$.  This may lead to severe bias, resulting in poor CATE estimators. Moreover, when the treatment groups are unbalanced, T-Learner might over-smooth the outcome regression in the less represented group leading to spurious heterogeneity resulting from the fitted models \citep{Kunzel2019}. 
\\
\indent
Thirdly, T-Learner is not Neyman-othogonal, which is a property that will be defined precisely later. As a result the slow rates of convergence affecting the outcome predictions would propagate into the final CATE estimates, which would be undesirable as the true outcome functions may be complex and therefore suffer from slow rates of convergence.
\\
\indent
Finally, T-Learner does not optimize a useful or well-understood quantity when $\tau \left( X \right) \notin \mathcal{G}$,  i.e. when the CATE does not belong to the class of functions over which we are minimizing the population risk, as is for instance the case when we wish to restrict $\mathcal{G}$ to functions of only the subset of features $V$.  

\subsection{Inverse-probability-weighting} 

In view of the aforementioned problems with the T-Learner, we may instead consider replacing the treatment effect $Y^1 - Y^0$ in (\ref{eqn1}) by the pseudo-outcome
\begin{equation*}
\frac{AY}{\pi_0 \left( X \right)} - \frac{\left( 1-A \right) Y}{1- \pi_0 \left( X \right)}.
\end{equation*}
In this way, we obtain the following population risk function
\begin{align} \label{eqn2}
\mathcal{L} \left( g, \pi_0 \right) = \mathbb{E} \left[ \left\lbrace \frac{AY}{\pi_0 \left( X \right)} - \frac{\left( 1-A \right) Y}{1- \pi_0 \left( X \right)} - g \left( V \right) \right\rbrace^2 \right].
\end{align}
The inverse-probability-weighted population risk function (\ref{eqn2}) is an example of a population risk function with a nuisance parameter given by the propensity score.  
\begin{lemma} \label{lemma1}
Population risk functions (\ref{eqn1}) and (\ref{eqn2}) lead to the same minimizer over the function space $\mathcal{G}$.
\end{lemma}

Proofs of all the results can be found in the Appendix \ref{AppA}. Lemma \ref{lemma1} shows that this IPW-Learner addresses several of the concerns raised for the T-Learner. First, minimization of population risk function (\ref{eqn2}) delivers the CATE when $\mathbb{E} \left( Y^1 - Y^0 \cond V \right) \in \mathcal{G}$. Second, when the CATE does not belong to $\mathcal{G}$, minimization of (\ref{eqn2}) still retrieves the function that is closest to the CATE in terms of mean-squared error, or from (\ref{eqn1}), the function that best predicts the individual treatment effect, i.e. $Y^1 - Y^0$. Third, by directly modeling the CATE via (\ref{eqn2}), we can also directly incorporate any a priori knowledge that we may have (e.g., that the function class $\mathcal{G}$ should contain constant functions to allow for the possibility of no treatment effect heterogeneity).
\\
\indent
In practice, unless known by design, the propensity score $\pi_0 \left( X \right)$ is unknown and needs to be estimated. It is therefore natural to minimize the sample analog of the risk function of the form
\begin{align} \label{eqn3}
&\hat{\mathcal{L}} \left( g, \hat{\pi} \right) = \frac{1}{n} \sum_{i=1}^n \left\lbrace \frac{A_i Y_i}{\hat{\pi} \left( X_i \right)} - \frac{\left( 1-A_i \right) Y_i}{1- \hat{\pi} \left( X_i \right)} - g \left( V_i \right) \right\rbrace^2 + \Lambda \left\lbrace g \left( \cdot \right) \right\rbrace.
\end{align}
Here $\Lambda \left\lbrace g \left( \cdot \right) \right\rbrace$ is a penalization term that usually will be present when employing a machine learning algorithm to avoid that the resulting CATE estimate is an overly "complex" function and $\hat{\pi} \left( X \right)$ is a nonparametric or data-adaptive estimator (e.g. spline estimator or random forest estimator) of the propensity score. This is sometimes referred to as plug-in empirical risk minimization \citep{Foster2023}. One then hopes that the minimizer $\hat{g}$ for the parameter of interest $g$ obtained through minimizing the sample risk function evaluated at the estimated value of the propensity score, i.e. sample risk (\ref{eqn3}), is close to the minimizer $g_0$ obtained through minimizing the population risk evaluated at the true value of the propensity score, i.e. the population risk function (\ref{eqn2}). Unfortunately, in general this is not the case for inverse-probability-weighting. In particular, when the dimension of $V$ is much smaller than that of $X$, or when the function class $\mathcal{G}$ is relatively small, then the estimation errors in $\hat{\pi} \left( X \right)$ may dominate those in the minimizer of (\ref{eqn3}) with known propensity score, thereby causing the minimizer to have slower convergence rates \citep{McGrath2022}.
\\
\indent
Throughout, we will therefore aim to find estimation procedures that are at most minimally affected by the estimation of nuisance parameters, in the sense that small errors in the estimated nuisance parameters have only a small impact on the estimation of the target parameter (as we will formalize later). In line with the recent literature, we will achieve this by utilizing loss functions with a property called Neyman-orthogonality \citep{Neyman1979, Chernozhukov2018a, Foster2023}.

\subsection{Retargeting for heterogeneous treatment effects estimation} \label{subsect8}

Before discussing Neyman-orthogonality in Section \ref{Section3}, note that minimization of the mean-squared error (\ref{eqn1}) may not be relevant in populations where treatment or no-treatment are not feasible options for some. For instance, in the empirical example presented in Section \ref{Section4} that inspired this work, we are interested in estimates of the effect of initiating renal replacement therapy (RRT) in acute kidney injury (AKI) patients at the intensive care unit (ICU), conditional on patient characteristics. Such estimates could then be used by clinicians to decide whether or not to initiate treatment for particular patients. However, for many AKI patients at the ICU clinicians would never consider initiating RRT, based on their available characteristics. It is therefore not of much interest to minimize, as in (\ref{eqn1}), average mean squared error in the individual treatment effect $Y^1-Y^0$, averaged over the whole population of AKI patients at the ICU. Instead, one may rather wish to focus on the patient population with the highest uncertainty about the optimal treatment decision, as this is arguably exactly the patient population for which clinicians are in most need of treatment decision support. 
\\
\indent
One way to address this problem is through the use of weighted estimands \citep{Hirano2003, Crump2006} and retargeting, also used by \cite{Kallus2021} in the context of policy learning. The idea is to change the population to a retargeted population in which all treatment options are plausible, and thus treatment effects are easier to infer; this is achieved by downweighting individuals for whom one treatment option is unlikely (based on their covariates $X$). Retargeting thus prioritizes the population of individuals for whom decisions need to be made. As shown in Lemma \ref{lemma2} below, this reweighting of the loss function does not change the minimizer, provided that $V = X$ and the function space $\mathcal{G}$ contains the CATE.
\begin{lemma} \label{lemma2}
Assume that $V = X$ and $\mathbb{E} \left( Y^1 - Y^0 \cond V \right) \in \mathcal{G}$. Then CATE is the minimizer of the population risk function \begin{align} \label{eqn4}
\mathcal{L} \left( g; \omega \right) = \mathbb{E} \left[ \omega \left( X \right) \left\lbrace \left( Y^1 - Y^0 \right) - g \left( V \right) \right\rbrace^2 \right]
\end{align}
over the function space $\mathcal{G}$, for each choice of the weight function $\omega \left( \cdot \right)$. 
\end{lemma}
It follows from Lemma \ref{lemma2} that the population risk function (\ref{eqn4}) forms the basis for an entire class of CATE learners. These all deliver the CATE when $V = X$ and $\mathbb{E} \left( Y^1 - Y^0 \cond V \right) \in \mathcal{G}$, but may have drastically different behavior otherwise. For instance, in the case $V$ is a strict subset of the features $X$, (\ref{eqn4}) has a different minimizer 
\begin{equation} \label{eqn9}
\frac{ \mathbb{E} \left\lbrace \omega \left( X \right) \left( Y^1 - Y^0 \right) \cond V \right\rbrace }{ \mathbb{E} \left\lbrace \omega \left( X \right) \cond V \right\rbrace }.
\end{equation}
depending on the choice of weight function. We will study this in more detail in the subsequent sections, where we will also address the fact that the population risk function (\ref{eqn4}) cannot be used directly in a data analysis as it involves both potential outcomes $Y^0$ and $Y^1$. In particular, in the next section, we will construct weighted Neyman-orthogonal learners whose population risk functions approximate (\ref{eqn4}) under the earlier identifiability assumptions. We will moreover find these learners to be differently affected by the estimation of nuisance parameters, even when $V = X$ and $\mathbb{E} \left( Y^1 - Y^0 \cond V \right) \in \mathcal{G}$. 

\section{Orthogonal statistical learning for heterogeneous treatment effects} \label{Section3}

Following the literature on orthogonal statistical learning \citep{Foster2023} we consider population risk functions of the form $\mathcal{L} \left( g, \eta \right)$, where $g \in \mathcal{G}$ is called the target parameter (or the parameter of interest) belonging to a function space $\mathcal{G}$ and $\eta \in \mathcal{H}$ is a nuisance parameter belonging to an infinite-dimensional nuisance space $\mathcal{H}$. The aim is to estimate $g_0$, which is a minimizer of the population risk function $\mathcal{L} \left( g, \eta \right)$, i.e.
\begin{equation*}
g_0 \coloneqq \argmin_{g \in \mathcal{G}} \mathcal{L} \left( g, \eta_0 \right),
\end{equation*}
evaluated at the true value of the nuisance parameter $\eta_0$. As $\eta_0$ is unknown, we consider minimizing the population risk function $\mathcal{L} \left( g, \hat{\eta} \right)$ with plugged-in estimated value $\hat{\eta}$ instead.  This will lead to a different minimizer $\hat{g}$ than $g_0$. The quality of the nuisance parameter estimation may then have a significant impact on the quality of $\hat{g}$. 
\\
\indent
Ideally, we would like the estimator for the parameter of interest $g$ obtained through minimizing the risk function evaluated at the estimated value of the nuisance parameter to be close to the minimizer obtained through minimizing the risk function evaluated at the true value of the nuisance parameter. Risk functions with such a property are called Neyman-orthogonal \citep{Neyman1979, Chernozhukov2018a, Foster2023}. Before we can precisely define Neyman-orthogonality we define the notion of a directional derivative, which allows us to consider differentiability in function spaces. 
\begin{definition}[Directional derivative \citep{Foster2023}] Let $\mathcal{F}$ be a vector space of functions. For a functional $F \colon \mathcal{F} \rightarrow \mathbb{R}$, we define the derivative operator $D_f F \left( f \right) \left[ h \right] = \frac{d}{dt} F \left( f + t h \right) \Bigr|_{t=0} $ for a pair of functions $f,h \in \mathcal{F}$. Likewise, we define $D_f^k F \left( f \right) \left[ h_1, \dots, h_k \right] = \frac{\partial^k}{\partial t_1 \dots \partial t_k} F \left( f + t_1 h_1 + \dots + t_k h_k \right)  \Bigr|_{t_1 = \dots = t_k = 0}$.
\end{definition}
\begin{definition}[Neyman-orthogonal loss function \citep{Foster2023}] The population risk $\mathcal{L}$ is Neyman-orthogonal, if
\begin{align*}
D_\eta D_g \mathcal{L} \left( g_0, \eta_0 \right) \left[ g - g_0, \eta - \eta_0 \right] = 0 \text{ } \forall g \in \mathcal{G}, \forall \eta \in \mathcal{H}.
\end{align*}
\end{definition}
Note that the directional derivative with respect to the target parameter of the population risk $\mathcal{L}$ can be viewed as an estimating function, analogous to the score function in maximum likelihood estimation. Therefore, the Neyman-orthogonality of a population risk function states that the estimating function obtained as a directional derivative with respect to the target parameter is locally insensitive to small perturbations of the nuisance parameter around its true value.

\subsection{Orthogonal loss function based on the efficient influence function}  \label{subsect1}

To construct an observed data loss function that approximates  (\ref{eqn4}) and is Neyman-orthogonal, we leverage results on Neyman-orthogonality for scalar estimands from the literature on semiparametric inference \citep{Pfanzagl1990, Bickel1998, vanderLaan2003, Tsiatis2006, Kosorok2008, vanderLaan2011, Chernozhukov2018a}. As a first step we choose a finite-dimensional estimand, defined as the minimizer to (\ref{eqn4}) with $g(V)$ set to a constant $g$. This delivers the weighted average treatment effect (WATE) \citep{Hirano2003, Crump2006}
\begin{equation} \label{eqn10}
g \coloneqq \frac{ \mathbb{E} \left\lbrace \omega \left( X \right) \left( Y^1 - Y^0 \right) \right\rbrace }{ \mathbb{E} \left\lbrace \omega \left( X \right)  \right\rbrace },
\end{equation}
where $\omega \left( X \right) \coloneqq \lambda \left\lbrace \pi \left( X \right) \right\rbrace$ with $\lambda \colon \left[ 0, 1 \right] \rightarrow \mathbb{R}$ known continuously differentiable function and the propensity score unknown. Note that the average treatment effect is a special case with $\lambda \left( \cdot \right) \equiv 1$. We discuss different choices of the weight function in Sections \ref{subsect6}-\ref{subsect7}. More generally, we thus obtain a class of candidate estimands indexed by weight function $\omega \left( \cdot \right) = \lambda \left\lbrace \pi \left( \cdot \right) \right\rbrace$. 
\\
\indent
The key building block in the construction of Neyman-orthogonal loss functions for CATE estimation is the efficient influence function (EIF) of the chosen finite-dimensional estimand under the nonparametric model. For the WATE \citep{Hirano2003, Crump2006} this is given by:
\begin{align} \label{eqn11}
\phi \left( Z; \eta, \lambda \left\lbrace \pi \left( X \right) \right\rbrace \right) &= \frac{ \lambda \left\lbrace \pi \left( X \right) \right\rbrace }{ \mathbb{E} \left[ \lambda \left\lbrace \pi \left( X \right) \right\rbrace \right] } \left( \frac{A - \pi \left( X \right)}{\pi \left( X \right) \left\lbrace 1 - \pi \left( X \right) \right\rbrace} \left\lbrace Y - Q^{\left( A \right)} \left( X \right) \right\rbrace \right.\\
\nonumber
&\left.+ \frac{ \rho \left( A, \pi \left( X \right) \right) }{\lambda \left\lbrace \pi \left( X \right) \right\rbrace} \left\lbrace Q^{\left( 1 \right)} \left( X \right) - Q^{\left( 0 \right)} \left( X \right) - g \right\rbrace \right)
\end{align}
where $\eta \coloneqq \left( \pi \left( X \right), Q^{\left( 0 \right)} \left( X \right), Q^{\left( 1 \right)} \left( X \right) \right)$ is the nuisance parameter and $\rho \left( A, \pi \left( X \right) \right) \coloneqq \left\lbrace A - \pi \left( X \right) \right\rbrace \lambda^\prime \left\lbrace \pi \left( X \right) \right\rbrace + \lambda \left\lbrace \pi \left( X \right) \right\rbrace$. If $\rho \left( A, \pi \left( X \right) \right) \neq 0$, then we can rewrite (\ref{eqn11}) as: 
\begin{align*}
&\phi \left( Z; \eta, \lambda \left\lbrace \pi \left( X \right) \right\rbrace \right) = \frac{ \rho \left( A, \pi \left( X \right) \right)}{ \mathbb{E} \left[ \lambda \left\lbrace \pi \left( X \right) \right\rbrace \right] } \left\lbrace \varphi \left( Z; \eta,  \lambda \left\lbrace \pi \left( X \right) \right\rbrace \right) - g \right\rbrace
\end{align*}
where
\begin{align} \label{eqn24}
&\varphi \left( Z; \eta, \lambda \left\lbrace \pi \left( X \right) \right\rbrace \right) \coloneqq \frac{\lambda \left\lbrace \pi \left( X \right) \right\rbrace}{ \rho \left( A, \pi \left( X \right) \right)} \frac{A - \pi \left( X \right)}{\pi \left( X \right) \left\lbrace 1 - \pi \left( X \right) \right\rbrace} \left\lbrace Y - Q^{\left( A \right)} \left( X \right) \right\rbrace + Q^{\left( 1 \right)} \left( X \right) - Q^{\left( 0 \right)} \left( X \right).
\end{align} 
We consider an example of a weight function $\lambda \left\lbrace \pi \left( \cdot \right) \right\rbrace$, such that $\rho \left( A, \pi \left( X \right) \right) = 0$ for some individuals in the study population in Section \ref{subsect2}.
\\
\indent
As known from semiparametric theory \citep{Pfanzagl1990, Bickel1998, vanderLaan2003, Tsiatis2006, Kosorok2008, vanderLaan2011, Chernozhukov2018a}, the EIF is Neyman-orthogonal in the sense that its directional derivatives along paths that perturb the nuisance parameters have mean zero. This makes it an attractive candidate estimating function for $g$. In particular, a population risk function whose derivative with respect to $g$ equals the EIF - and which is therefore itself Neyman-orthogonal - is given by
\begin{align*}
 \frac{1}{\mathbb{E} \left[ \lambda \left\lbrace \pi \left( X \right) \right\rbrace \right]} \mathbb{E} \left[ \rho \left( A, \pi \left( X \right) \right) \left\lbrace \varphi \left( Z; \eta, \lambda \left( \pi \right) \right) - g \right\rbrace^2 \right].
\end{align*}
We can now generalize this to a population risk function for CATE estimation by allowing $g$ to depend on $V$:
\begin{align} \label{eqn5}
&\mathcal{L} \left( g, \eta, \lambda \left\lbrace \pi \left( X \right) \right\rbrace \right) = \frac{1}{\mathbb{E} \left[ \lambda \left\lbrace \pi \left( X \right) \right\rbrace \right]} \mathbb{E} \left[ \rho \left( A, \pi \left( X \right) \right) \left\lbrace \varphi \left( Z; \eta, \lambda \left( \pi \right) \right) - g \left( V \right) \right\rbrace^2 \right].
\end{align}
This leads to the following result:
\begin{lemma}  \label{lemma4} 
Population risk function $\mathcal{L} \left( g, \eta_0, \lambda \left\lbrace \pi_0 \left( X \right) \right\rbrace \right)$ given by (\ref{eqn5}) is Neyman-orthogonal for each choice of weight function $\lambda \left( \cdot \right)$.
\end{lemma}
\noindent
Importantly, the above construction of a Neyman-orthogonal risk function does not change the minimization problem, in the following sense:
\begin{lemma} \label{lemma3}
Population risk function $\mathcal{L} \left( g, \eta_0, \lambda \left\lbrace \pi_0 \left( X \right) \right\rbrace \right)$ given by (\ref{eqn5}) has the same minimizer $g_0$ over the function space $\mathcal{G}$ as the population risk function
\begin{align} \label{eqn12}
&\mathcal{L} \left( g, \lambda \left\lbrace \pi_0 \left( X \right) \right\rbrace \right) = \mathbb{E} \left[ \lambda \left\lbrace \pi_0 \left( X \right) \right\rbrace \left\lbrace \left( Y^1 - Y^0 \right) - g \left( V \right) \right\rbrace^2 \right].
\end{align}
\end{lemma}
\noindent
Combining Lemma \ref{lemma2} and Lemma \ref{lemma3} we obtain the following result:
\begin{corollary} \label{corollary1}
Assume that $V = X$ and $\mathbb{E} \left( Y^1 - Y^0 \cond X \right) \in \mathcal{G}$. Then CATE is the minimizer of population risk function $\mathcal{L} \left( g, \eta_0, \lambda \left\lbrace \pi_0 \left( X \right) \right\rbrace \right)$ given by (\ref{eqn5}) for each choice of the weight function $\lambda \left( \cdot \right)$. 
\end{corollary} 
One may alternatively construct orthogonal learners by minimizing a debiased estimator of the population risk function (or a more easily estimable function with the same minimizer), based on that function's EIF (considering $g(V)$ to be fixed). However, this need not result in debiased loss functions that are easy to optimize (e.g., \cite{vansteelandt2023orthogonal,van2024combining}).  

\subsection{Examples of orthogonal CATE learners}

Below we present some of the recently developed methods for CATE estimation, e.g. DR-Learner \citep{Kennedy2023} and R-Learner \citep{Nie2021}, in light of the framework described above, and we introduce a novel propensity-score-weighted DR-Learner.

\subsubsection{DR-Learner}\label{subsect6}
The constant weight function $\lambda \left\lbrace \pi \left( X \right) \right\rbrace \equiv 1$ gives rise to the following population risk function:
\begin{align} \label{eqn6}
&\mathcal{L} \left( g, \eta, \lambda \left\lbrace \pi \left( X \right) \right\rbrace \right) \\
\nonumber
&= \mathbb{E} \left( \left[ \frac{ A-\pi \left( X \right) }{\pi \left( X \right) \left\lbrace 1-\pi \left( X \right) \right\rbrace} \left\lbrace Y - Q^{\left( A \right)} \left( X \right) \right\rbrace + Q^{\left( 1 \right)} \left( X \right) - Q^{\left( 0 \right)} \left( X \right) - g \left( V \right) \right]^2 \right).
\end{align}
The CATE estimation approach aiming to minimize the population risk function (\ref{eqn6}) is known as DR-Learner \citep{Kennedy2023}. From Lemma \ref{lemma3} we can see that minimizing the population risk function (\ref{eqn6}) leads to the same minimizer as obtained through minimization of the population risk function
\begin{equation*}
\mathcal{L} \left( g, \lambda \left\lbrace \pi \left( X \right) \right\rbrace \right) = \mathbb{E} \left[ \left\lbrace \left( Y^1 - Y^0 \right) - g \left( V \right) \right\rbrace^2 \right].
\end{equation*}
This means that DR-Learner finds the function $g \in \mathcal{G}$ that is the closest to the treatment effect $Y^1 - Y^0$ in the whole population in terms of mean-squared error. 

\subsubsection{Propensity-score-weighted DR-Learner} \label{subsect2}

The weight function $\lambda \left\lbrace \pi \left( X \right) \right\rbrace = \pi \left( X \right)$ yields $\rho \left( A, \pi \left( X \right) \right) = A$, which is non-zero in the treated only (i.e. individuals with $A = 1$). In this case, population risk function (\ref{eqn5}) equals
\begin{align} \label{eqn13}
\mathcal{L} \left( g, \eta, \lambda \left\lbrace \pi \left( X \right) \right\rbrace \right) &= \frac{1}{\mathbb{E} \left\lbrace \pi \left( X \right) \right\rbrace} \mathbb{E} \left( A \left[ \frac{A - \pi \left( X \right) }{ \left\lbrace 1 - \pi \left( X \right) \right\rbrace A} \left\lbrace Y - Q^{\left( 0 \right)} \left( X \right) \right\rbrace - g \left( V \right) \right]^2 \right) \\
\nonumber
&= \frac{1}{\mathbb{E} \left\lbrace \pi \left( X \right) \right\rbrace} \mathbb{E} \left( \left[ \left\lbrace Y - Q^{\left( 0 \right)} \left( X \right) \right\rbrace - g \left( V \right) \right]^2 \right)
\end{align}
where the second equality follows as one restricts minimization to the population of treated. We refer to the approach aiming to minimize the population risk function as propensity-score-weighted DR-Learner. When minimization is over an unrestricted function class, then it targets the conditional average treatment effect in the treated:
\begin{align*}
\frac{\mathbb{E} \left\lbrace \pi_0 \left( X \right) \left( Y^1 - Y^0 \right) \cond V \right\rbrace}{\mathbb{E} \left\lbrace \pi_0 \left( X \right) \cond V \right\rbrace} &= \mathbb{E} \left( Y^1 - Y^0 \cond V, A = 1 \right).
\end{align*} 
From Lemma \ref{lemma3}, we moreover see that minimizing the population risk function (\ref{eqn13}) leads to the same minimizer as obtained through minimization of population risk function
\begin{align*}
\mathcal{L} \left( g, \lambda \left\lbrace \pi_0 \left( X \right) \right\rbrace \right) &= \mathbb{E} \left[ \pi_0 \left( X \right) \left\lbrace \left( Y^1 - Y^0 \right) - g \left( V \right) \right\rbrace^2 \right] \\
&= \mathbb{E} \left[ A \left\lbrace \left( Y^1 - Y^0 \right) - g \left( V \right) \right\rbrace^2 \right].
\end{align*}
This means that the propensity-score-weighted DR-Learner finds the function $g \in \mathcal{G}$ closest to the treatment effect $Y^1 - Y^0$ in the treated in terms of mean-squared error. 
\\
\indent
By aiming to minimize prediction error in the treated, the propensity-score-weighted DR-Learner is less relevant for providing treatment decision support in prospective settings where treatment decisions must be made for new patients. It is primarily indicated in retrospective settings, where one wishes to quantify treatment effect heterogeneity in settings where some patients are ineligible for treatment, as in the previously discussed critical care example, or in exposure prevention studies. For instance, studies on the impact of hospital-acquired infections are primarily interested in the impact of preventing infection in those who acquired it \citep{Vansteelandt2009}. 
\\
\indent
One of the advantages of the propensity-score-weighted DR-Learner is that one may relax the positivity assumption (Assumption \ref{Assump2}) to be $0 \leq \pi \left( X \right) < 1$, since one restricts the attention to the population of treated.

\subsubsection{R-Learner} \label{subsect7}

For $\lambda \left\lbrace \pi \left( X \right) \right\rbrace \coloneqq \pi \left( X \right) \left\lbrace 1 - \pi \left( X \right) \right\rbrace$, population risk function (\ref{eqn5}) specializes to 
\begin{align} \label{eqn14} 
&\mathcal{L} \left( g, \eta, \lambda \left\lbrace \pi \left( X \right) \right\rbrace \right) = \frac{1}{\mathbb{E} \left[ \pi \left( X \right) \left\lbrace 1 - \pi \left( X \right) \right\rbrace \right]} \mathbb{E} \left( \left[ \left\lbrace Y - Q \left( X \right) \right\rbrace -  \left\lbrace A - \pi \left( X \right) \right\rbrace g \left( V \right) \right]^2 \right),
\end{align}
where $Q \left( X \right) \coloneqq \mathbb{E} \left( Y \cond X \right)$ is the outcome prediction and we have used the identity $Q \left( X \right) = \pi \left( X \right) Q^{\left( 1 \right)} \left( X \right) + \left\lbrace 1 - \pi \left( X \right) \right\rbrace Q^{\left( 0 \right)} \left( X \right)$. This approach is known as R-Learner \citep{Nie2021}. It follows from the above that R-Learner can be viewed as a propensity-overlap-weighted DR-Learner, and thus that DR-Learner and R-Learner are special cases of the retargeting framework presented above. Similar remarks were published by different authors while our paper was under review \citep{Fisher2023, Chernozhukov2024}.
\\
\indent
From Lemma \ref{lemma3} we can see that minimizing the population risk function (\ref{eqn14}) leads to the same minimizer as obtained through minimization of the population risk function
\begin{align*}
&\mathcal{L} \left( g,  \lambda \left\lbrace \pi \left( X \right) \right\rbrace \right) = \mathbb{E} \left[ \pi \left( X \right) \left\lbrace 1 - \pi \left( X \right) \right\rbrace \left\lbrace \left( Y^1 - Y^0 \right) - g \left( V \right) \right\rbrace^2 \right].
\end{align*}
It follows that in cases where $\tau \left( V \right) \notin \mathcal{G}$, i.e. when CATE does not belong to the class of functions over which we are minimizing the risk function, R-Learner does not target the CATE, but the propensity-overlap-weighted CATE instead. It thus targets the optimization to the patient population with most uncertainty about the treatment decision. Arguably, this is also a patient population for which clinicians would be the most interested in the treatment effect estimates to guide their decisions on whether or not to initiate the treatment. As argued in \cite{Li2018} the propensity-overlap-weighted population represents a population of patients who could receive either of the treatment options with substantial probability and as such remain in clinical equipoise. For such patients, the clinical consensus regarding the treatment option had not been met yet, and therefore arguably such a population is of particular interest from the research point of view. 

\subsection{Estimation of the CATE} \label{subsect3}

Based on the obtained population risk function (\ref{eqn5}) we estimate the CATE using the following two-step procedure relying on sample-splitting into two parts: A and B \citep{Kennedy2023, Nie2021, Foster2023}. Divide the data evenly into $K$ folds. Commonly, $K$ is equal $2$, $5$ or $10$. Combine $K-1$ folds into "part A" of the data and the remaining fold will be "part B" of the data. In the first step, we estimate the nuisance parameters $\eta = \left( \pi \left( X \right), Q^{\left( 0 \right)} \left( X \right), Q^{\left( 1 \right)} \left( X \right) \right)$ on part A of the data using flexible data-adaptive (e.g. machine learning) methods. In the second step we regress the pseudo-outcome (\ref{eqn24}) with plug-in estimates for the nuisance parameters obtained in the first step on the covariates $V$ again using data-adaptive methods with weights $\rho \left( A_i, \hat{\pi} \left( X_i \right) \right)$ on part B of the data. This corresponds to the following plug-in empirical risk minimization
\begin{align} \label{eqn26}
&\frac{ \sum_{i=1}^n \rho \left( A_i, \hat{\pi} \left( X_i \right) \right) \left\lbrace \varphi \left( Z_i; \hat{\eta}, \lambda \left( \hat{\pi} \right) \right) - g \left( V_i \right) \right\rbrace^2}{\sum_{i=1}^n \lambda \left\lbrace \hat{\pi} \left( X_i \right) \right\rbrace} + \Lambda \left\lbrace g \left( \cdot \right) \right\rbrace.
\end{align}
Here, similarly as in (\ref{eqn3}), $\Lambda \left\lbrace g \left( \cdot \right) \right\rbrace$ is a penalization term that usually will be present when employing a machine learning algorithm to avoid that the resulting CATE estimate is an overly "complex" function. 
\\
\indent
Sample-splitting may lead to a loss of efficiency as now we are using only one part of the data (part A) to estimate nuisance parameters and part of the data (part B) to estimate the parameter of interest, i.e. CATE. For the purpose of making the most efficient use of the available data one may apply the idea of cross-fitting \citep{Chernozhukov2018a, Kennedy2023}. This means that one performs the procedure described above $K$ times, permuting the roles of the folds, i.e. each time different fold will be treated as "part B" of the data. The final estimator is then obtained as the average of the estimators obtained from different folds.

\begin{remark}
Sample-splitting is a crucial step in our procedure of estimating CATE and has been previously commonly used in the CATE estimation literature \citep{Kennedy2023, Nie2021, Foster2023}. Through sample-splitting we obtain the nuisance parameter estimates from a sample that is independent from the sample used for the plug-in empirical risk minimization (\ref{eqn26}), and therefore we can treat the nuisance parameters in the empirical risk minimization problem as fixed. This allows us to invoke the performance guarantees for the machine learning algorithms used for the empirical risk minimization (\ref{eqn26}), since it requires the data to be independent and identically distributed (i.i.d.). Fitting the nuisance parameters on the same sample as used for minimizing (\ref{eqn26}) would result in minimization over a correlated sample (through the plug-in estimates of the nuisance parameters) and hence would lead to violation of the i.i.d. assumption.
\end{remark}

\subsection{Error bounds} 

In this section, we will derive error bounds for the previously discussed learners. By Neyman-orthogonality we have the following favourable error bounds for the CATE estimator $\hat{g}$ obtained via minimization of the weighted population risk function $\mathcal{L} \left( g, \hat{\eta}, \lambda \left\lbrace \hat{\pi} \left( X \right) \right\rbrace \right)$ relative to $g_0$, which is obtained via minimization of the weighted population risk function $\mathcal{L} \left( g, \eta_0, \lambda \left\lbrace \pi_0 \left( X \right) \right\rbrace \right)$ evaluated at the true values of the nuisance parameters.
\begin{theorem} \label{theorem1}
Let $R_g \coloneqq \mathcal{L} \left( \hat{g}, \hat{\eta}, \lambda \left\lbrace \hat{\pi} \left( X \right) \right\rbrace \right) - \mathcal{L} \left( g_0,  \hat{\eta}, \lambda \left\lbrace \hat{\pi} \left( X \right) \right\rbrace \right)$ for $\mathcal{L} \left( g, \eta, \lambda \left\lbrace \pi \left( X \right) \right\rbrace \right)$ given by (\ref{eqn5}). Suppose that 
\begin{equation} \label{eqn8}
\inf_{g \in \mathcal{G} \setminus \left\lbrace g_0 \right\rbrace} \frac{\mathbb{E} \left[ \rho \left( A, \pi \left( X \right) \right) \left\lbrace g \left( V \right) - g_0 \left( V \right) \right\rbrace^2 \right]}{ \norm{ g \left( V \right) - g_0 \left( V \right) }_2^2} > 0
\end{equation}
for all $\eta = \left( \pi \left( X \right), Q^{\left( 0 \right)} \left( X \right), Q^{\left( 1 \right)} \left( X \right) \right) \in \mathcal{H}$ and $\lambda \in C_b^3 \left( \left[ 0,1 \right] \right)$. Then we have the following error bound
\begin{align} \label{eqn7}
\norm{\hat{g} - g_0}_2^2 \lesssim R_g &+ \norm{ C_1 \left( X \right) \left\lbrace \hat{\pi} \left( X \right) - \pi_0 \left( X \right) \right\rbrace^2 }_2^2 \\
\nonumber
&+ \norm{ C_2 \left( X \right) \left\lbrace \hat{\pi} \left( X \right) - \pi_0 \left( X \right) \right\rbrace \left\lbrace \hat{Q}^{\left( 1 \right)} \left( X \right) - Q^{\left( 1 \right)}_0 \left( X \right) \right\rbrace }_2^2 \\
\nonumber
&+ \norm{ C_3 \left( X \right) \left\lbrace \hat{\pi} \left( X \right) - \pi_0 \left( X \right) \right\rbrace \left\lbrace \hat{Q}^{\left( 0 \right)} \left( X \right) - Q^{\left( 0 \right)}_0 \left( X \right) \right\rbrace }_2^2
\end{align}
and $C_1$, $C_2$ and $C_3$ are given by
\begin{align*} 
C_1 \left( X \right) &= \lambda \left\lbrace \overline{\pi} \left( X \right) \right\rbrace \left[ \frac{\pi_0 \left( X \right)}{\overline{\pi}^3 \left( X \right)} \left\lbrace Q_0^{\left( 1 \right)} \left( X \right) - \overline{Q}^{\left( 1 \right)} \left( X \right) \right\rbrace - \frac{ 1-\pi_0 \left( X \right) }{\left\lbrace 1-\overline{\pi} \left( X \right) \right\rbrace^3} \left\lbrace Q_0^{\left( 0 \right)} \left( X \right) - \overline{Q}^{\left( 0 \right)} \left( X \right) \right\rbrace \right] \\
\nonumber
&- \lambda^\prime \left\lbrace \overline{\pi} \left( X \right) \right\rbrace \left[ \frac{\pi_0 \left( X \right)}{\overline{\pi}^2 \left( X \right)} \left\lbrace Q_0^{\left( 1 \right)} \left( X \right) - \overline{Q}^{\left( 1 \right)} \left( X \right) \right\rbrace + \frac{1-\pi_0 \left( X \right)}{\left\lbrace 1-\overline{\pi} \left( X \right) \right\rbrace^2} \left\lbrace Q_0^{\left( 0 \right)} \left( X \right) - \overline{Q}^{\left( 0 \right)} \left( X \right) \right\rbrace \right] \\
\nonumber
&+ \frac{1}{2} \lambda^{\prime\prime} \left\lbrace \overline{\pi} \left( X \right) \right\rbrace \left[ \frac{\pi_0 \left( X \right)}{\overline{\pi} \left( X \right)} \left\lbrace Q_0^{\left( 1 \right)} \left( X \right) - \overline{Q}^{\left( 1 \right)} \left( X \right) \right\rbrace - \frac{1-\pi_0 \left( X \right)}{1-\overline{\pi} \left( X \right)} \left\lbrace Q_0^{\left( 0 \right)} \left( X \right) - \overline{Q}^{\left( 0 \right)} \left( X \right) \right\rbrace \right] \\
\nonumber
&+ \frac{1}{2} \left[ \lambda^{\prime\prime} \left\lbrace \overline{\pi} \left( X \right) \right\rbrace + \lambda^{\prime\prime\prime} \left\lbrace \overline{\pi} \left( X \right) \right\rbrace  \left\lbrace \pi_0 \left( X \right) - \overline{\pi} \left( X \right) \right\rbrace \right] \left\lbrace \overline{Q}^{\left( 1 \right)} \left( X \right) - \overline{Q}^{\left( 0 \right)} \left( X \right) - g_0 \left( V \right) \right\rbrace 
\end{align*}
\begin{align*}
C_2 \left( X \right) &= \lambda \left\lbrace \overline{\pi} \left( X \right) \right\rbrace \frac{\pi_0 \left( X \right)}{\overline{\pi}^2 \left( X \right)} - \lambda^\prime \left\lbrace \overline{\pi} \left( X \right) \right\rbrace \frac{\pi_0 \left( X \right)}{\overline{\pi} \left( X \right)} + \left\lbrace \pi_0 \left( X \right) - \overline{\pi} \left( X \right) \right\rbrace \lambda^{\prime\prime} \left\lbrace \overline{\pi} \left( X \right) \right\rbrace 
\end{align*}
\begin{align*}
C_3 \left( X \right) &= \lambda \left\lbrace \overline{\pi} \left( X \right) \right\rbrace \frac{ 1-\pi_0 \left( X \right) }{ \left\lbrace 1-\overline{\pi} \left( X \right) \right\rbrace^2 } + \lambda^\prime \left\lbrace \overline{\pi} \left( X \right) \right\rbrace \frac{1-\pi_0 \left( X \right)}{1-\overline{\pi} \left( X \right)} - \left\lbrace \pi_0 \left( X \right) - \overline{\pi} \left( X \right) \right\rbrace \lambda^{\prime\prime} \left\lbrace \overline{\pi} \left( X \right) \right\rbrace,
\end{align*}
for some $\overline{\eta} \in \left\lbrace t \eta_0 + \left( 1 - t\right) \hat{\eta} \cond t \in \left[ 0, 1 \right] \right\rbrace$.
\end{theorem}
The notation $x \lesssim y$ means that there exists a constant $M$, such that $x \leq My$, and $C_b^3 \left( \left[ 0,1 \right] \right)$ is a set of three-times continuously differentiable functions $f \colon \left[ 0,1 \right] \rightarrow \mathbb{R}$ with bounded derivatives. We discuss in the following sections the implications of Theorem \ref{theorem1} for different choices of weight functions $\lambda \left( \cdot \right)$. The proof follows the ideas outlined in \cite{Foster2023} and uses the fact that population risk function (\ref{eqn5}) is Neyman-orthogonal, see the Appendix \ref{AppA} for details. Condition (\ref{eqn8}) is required for strong convexity with respect to the target parameter of the loss function (\ref{eqn5}) \citep{Foster2023}.
\\
\indent
Let $n^{-2\varphi_g}$ be the rate of convergence of $\hat{g}$ to $g_0$, i.e. $R_g = O_p \left( n^{-2\varphi_g} \right)$. Furthermore, define the following rates of convergence $n^{-\varphi_i}$, for $i \in \left\lbrace 1,2,3 \right\rbrace$, for the terms in the error bound
\begin{align*}
    &\norm{ C_1 \left( X \right) \left\lbrace \hat{\pi} \left( X \right) - \pi_0 \left( X \right) \right\rbrace^2 }_2^2 = O_p \left( n^{-2\varphi_1} \right) \\
    &\norm{ C_2 \left( X \right) \left\lbrace \hat{\pi} \left( X \right) - \pi_0 \left( X \right) \right\rbrace \left\lbrace \hat{Q}^{\left( 1 \right)} \left( X \right) - Q^{\left( 1 \right)}_0 \left( X \right) \right\rbrace }_2^2 = O_p \left( n^{-2\varphi_2} \right) \\
    &\norm{ C_3 \left( X \right) \left\lbrace \hat{\pi} \left( X \right) - \pi_0 \left( X \right) \right\rbrace \left\lbrace \hat{Q}^{\left( 0 \right)} \left( X \right) - Q^{\left( 0 \right)}_0 \left( X \right) \right\rbrace }_2^2 = O_p \left( n^{-2\varphi_3} \right). 
\end{align*}
In order to have oracle behavior for the estimator $\hat{g}$, the estimation errors for the nuisance parameters must be of smaller order than the estimation error for the target parameter, i.e.
\begin{align} \label{eqn31}
    n^{-\varphi_1} + n^{-\varphi_2} + n^{-\varphi_3} = o_p \left( n^{-\varphi_g} \right).
\end{align}
In the following sections we will express the convergence rates $n^{-\varphi_i}$, for $i \in \left\lbrace 1,2,3 \right\rbrace$, for different learners in terms of the convergence rates for the nuisance parameters to understand when the oracle behavior condition (\ref{eqn31}) is satisfied for different learners.

\subsubsection{DR-Learner}

In the case of DR-Learner condition (\ref{eqn8}) is always satisfied and the terms $C_1$, $C_2$ and $C_3$ in the error bound (\ref{eqn7}) result to
\begin{align*}
C_1 \left( X \right) &= \frac{\pi_0 \left( X \right)}{\overline{\pi}^3 \left( X \right)} \left\lbrace Q_0^{\left( 1 \right)} \left( X \right) - \overline{Q}^{\left( 1 \right)} \left( X \right) \right\rbrace - \frac{ 1-\pi_0 \left( X \right) }{\left\lbrace 1-\overline{\pi} \left( X \right) \right\rbrace^3} \left\lbrace Q_0^{\left( 0 \right)} \left( X \right) - \overline{Q}^{\left( 0 \right)} \left( X \right) \right\rbrace \\
C_2 \left( X \right) &= \frac{\pi_0 \left( X \right)}{\overline{\pi}^2 \left( X \right)} \\
C_3 \left( X \right) &= \frac{ \left\lbrace 1 - \pi_0 \left( X \right) \right\rbrace}{ \left\lbrace 1 - \overline{\pi} \left( X \right) \right\rbrace^2}.
\end{align*}
Furthermore, assuming that $\overline{\pi} \left( X \right)$ is bounded away from $0$ and $1$, i.e. $\overline{\pi} \left( X \right) \in \left( \delta, 1-\delta \right)$, for some $\delta \in \left( 0,1 \right)$, we obtain using the H{\"o}lder inequality the following upper bounds
\begin{align*}
    &\norm{ C_1 \left( X \right) \left\lbrace \hat{\pi} \left( X \right) - \pi_0 \left( X \right) \right\rbrace^2 }_2 \lesssim \norm{ \left\lbrace \hat{\pi} \left( X \right) - \pi_0 \left( X \right) \right\rbrace^2}_4 \norm{ \hat{Q} \left( X \right) - Q_0 \left( X \right) }_4 \\
    &\norm{ C_2 \left( X \right) \left\lbrace \hat{\pi} \left( X \right) - \pi_0 \left( X \right) \right\rbrace \left\lbrace \hat{Q}^{\left( 1 \right)} \left( X \right) - Q^{\left( 1 \right)}_0 \left( X \right) \right\rbrace }_2 \lesssim \norm{ \hat{\pi} \left( X \right) - \pi_0 \left( X \right)}_4 \norm{ \hat{Q} \left( X \right) - Q_0 \left( X \right) }_4 \\
    &\norm{ C_3 \left( X \right) \left\lbrace \hat{\pi} \left( X \right) - \pi_0 \left( X \right) \right\rbrace \left\lbrace \hat{Q}^{\left( 0 \right)} \left( X \right) - Q^{\left( 0 \right)}_0 \left( X \right) \right\rbrace }_2 \lesssim \norm{ \hat{\pi} \left( X \right) - \pi_0 \left( X \right)}_4 \norm{ \hat{Q} \left( X \right) - Q_0 \left( X \right) }_4,
\end{align*}
where 
\begin{align*} 
\norm{ \hat{Q} \left( X \right) - Q_0 \left( X \right) }_4 \coloneqq \max_{a\in\left\lbrace 0,1 \right\rbrace} \norm{ \hat{Q}^{\left( a \right)} \left( X \right) - Q^{\left( a \right)}_0 \left( X \right) }_4.
\end{align*}
The inequality for the first term follows from the fact that
\begin{align*}
    \norm{ \overline{Q}^{\left( a \right)} \left( X \right) - Q_0^{\left( a \right)} \left( X \right) }_4 \leq \norm{ \hat{Q}^{\left( a \right)} \left( X \right) - Q_0^{\left( a \right)} \left( X \right) }_4 ,    
\end{align*}
for $a = 0,1$, because $\overline{Q}^{\left( a \right)} \left( X \right)$ lies between $Q_0^{\left( a \right)} \left( X \right)$ and $\hat{Q}^{\left( a \right)} \left( X \right)$. Furthermore, we can note that $n^{-\varphi_1}$ is of smaller order than $n^{-\varphi_2}$ and $n^{-\varphi_3}$. Therefore, the oracle behavior conditions given by (\ref{eqn31}) for the DR-Learner translates to 
\begin{equation} \label{eqn25}
  n^{-\varphi_\pi} n^{-\varphi_{Q}} = o_p \left( n^{-\varphi_g} \right),
\end{equation}
where $n^{-\varphi_\pi}$ and $n^{-\varphi_{Q}}$ are the rates of convergence for the nuisance parameters $\pi$, $Q^{\left( 0 \right)}$ and $Q^{\left( 1 \right)}$, respectively, defined as follows
\begin{align*}
    \norm{ \hat{\pi} \left( X \right) - \pi_0 \left( X \right) }_4 &= O_p \left( n^{-\varphi_\pi} \right) \\
    \norm{ \hat{Q} \left( X \right) - Q_0 \left( X \right) }_4 &= O_p \left( n^{-\varphi_{Q}} \right).
\end{align*}
This product structure suggests that the resulting estimator $\hat{g}$ is so-called rate doubly robust \citep{Rotnitzky2021}, meaning that slow convergence in the estimated propensity scores can be compensated by fast convergence in the outcome predictions, and vice versa.
Furthermore, when restricting the class of functions $\mathcal{G}$ to parametric functions, we may expect a parametric rate for the target parameter $g$, i.e. $n^{-\varphi_g} = n^{-1/2}$, when the nuisance parameters are known. From (\ref{eqn25}) we see that to attain the oracle behavior, we then need the product of the rates of convergence of the propensity score and the conditional mean outcomes to be $o_p\left( n^{-1/2} \right)$. When $g$ is instead obtained using non-parametric or data-adaptive methods, and hence its rate of convergence is slower than parametric, then also the propensity score and the conditional mean outcomes can have slower rates of convergence as long as (\ref{eqn25}) is satisfied.
\\
\indent
Importantly, the error bound for DR-Learner depends on the propensity score via the denominator of the terms $C_i,i=1,2,3$. When the propensity score is equal $0$ or $1$ using the DR-Learner becomes infeasible. In the empirical example from Section \ref{Section4} we consider RRT initiation in the AKI population at ICU. However, for some patients, namely those with propensity score close to $0$, clinicians would not even consider RRT initiation. In this case, the error bound for the DR-Learner is large or even infinite. This can be remedied through the use of different weight functions (see below). 
\\
\indent
The obtained upper bounds are expressed in terms of $L_4$-norms, however these could be upper bounded by $L_2$-norms under suitable smoothness conditions \citep{Bibaut2021, Mendelson2010, VandeGeer2014}.

\subsubsection{Propensity-score-weighted DR-Learner}

For the propensity-score-weighted DR-Learner, condition (\ref{eqn8}) specializes to 
\begin{equation*}
\inf_{g \in \mathcal{G} \setminus \left\lbrace g_0 \right\rbrace} \frac{\mathbb{E} \left[ \pi_0 \left( X \right) \left\lbrace g \left( V \right) - g_0 \left( V \right) \right\rbrace^2 \right]}{ \norm{ g \left( V \right) - g_0 \left( V \right) }_2^2 } > 0,
\end{equation*}
which holds if $\mathbb{P} \left\lbrace \pi_0 \left( X \right)>0 \right\rbrace > 0$, i.e. if the propensity score is non-zero with non-zero probability. Furthermore, the terms $C_1$, $C_2$ and $C_3$ in the error bound (\ref{eqn7}) from Theorem \ref{theorem1} simplify to 
\begin{align*}
    C_1 \left( X \right) &= \frac{ \left(1-\pi_0 \right) \left\lbrace 2-\overline{\pi} \left( X \right) \right\rbrace}{\left\lbrace 1-\overline{\pi} \left( X \right) \right\rbrace^3} \left\lbrace Q_0^{\left( 0 \right)} \left( X \right) - \overline{Q}^{\left( 0 \right)} \left( X \right) \right\rbrace \\
    C_2 \left( X \right) &= 0 \\
    C_3 \left( X \right) &= \frac{ \left\lbrace 1-\pi_0 \left( X \right) \right\rbrace}{\left\lbrace 1-\overline{\pi} \left( X \right) \right\rbrace^2}.
\end{align*}
In this case, the error bound does not depend on the estimation of $Q^{\left( 1 \right)} \left( X \right)$.  This is because the population risk function (\ref{eqn5}) for $\omega \left( X \right) = \pi \left( X \right)$ does not depend on $Q^{\left( 1 \right)} \left( X \right)$. Furthermore, assuming that $\overline{\pi} \left( X \right)$ is bounded away from $1$, i.e. $\overline{\pi} \left( X \right) \leq 1-\delta$, for some $\delta \in \left( 0,1 \right)$, we obtain the following upper bounds
\begin{align*}
    &\norm{ C_1 \left( X \right) \left\lbrace \hat{\pi} \left( X \right) - \pi_0 \left( X \right) \right\rbrace^2 }_2 \lesssim \norm{ \left\lbrace \hat{\pi} \left( X \right) - \pi_0 \left( X \right) \right\rbrace^2}_4 \norm{ \hat{Q} \left( X \right) - Q_0 \left( X \right) }_4 \\
    &\norm{ C_3 \left( X \right) \left\lbrace \hat{\pi} \left( X \right) - \pi_0 \left( X \right) \right\rbrace \left\lbrace \hat{Q}^{\left( 0 \right)} \left( X \right) - Q^{\left( 0 \right)}_0 \left( X \right) \right\rbrace }_2 \lesssim \norm{ \hat{\pi} \left( X \right) - \pi_0 \left( X \right)}_4 \norm{ \hat{Q} \left( X \right) - Q_0 \left( X \right) }_4.
\end{align*}
Furthermore, we can note that $n^{-\varphi_1}$ is of smaller order than $n^{-\varphi_3}$. Therefore,  for the propensity-score-weighted DR-Learner  to exhibit oracle behaviour, we need
\begin{equation*} 
  n^{-\varphi_\pi} n^{-\varphi_{Q}} = o_p \left( n^{-\varphi_g} \right),
\end{equation*}
which is a product of the rates of convergence for the propensity score and the conditional mean outcome in untreated, hence the resulting estimator $\hat{g}$ is also rate doubly robust. Furthermore, the error bound for the propensity-score-weighted DR-Learner is not affected by the presence of individuals with propensity score equal to $0$, though the presence of individuals with propensity score equal to $1$ leads to an infinite error bound.

\subsubsection{R-Learner}

For the R-Learner condition (\ref{eqn8}) specializes to 
\begin{equation*}
\inf_{g \in \mathcal{G} \setminus \left\lbrace g_0 \right\rbrace} \frac{\mathbb{E} \left[ \pi_0 \left( X \right) \left\lbrace 1 - \pi_0 \left( X \right) \right\rbrace \left\lbrace g \left( V \right) - g_0 \left( V \right) \right\rbrace^2 \right]}{ 
\norm{ g \left( V \right) - g_0 \left( V \right) }_2^2 } > 0,
\end{equation*}
which is non-zero if the propensity score is between $0$ and $1$ with non-zero probability, i.e. $\mathbb{P} \left\lbrace \pi_0 \left( X \right) \in \left( 0,1 \right)\right\rbrace > 0$. Furthermore, the terms $C_1$, $C_2$ and $C_3$ in the error bound (\ref{eqn7}) from Theorem \ref{theorem1} result to 
\begin{align*}
    C_1 \left( X \right) &= \overline{Q}^{\left( 1 \right)} \left( X \right) - \overline{Q}^{\left( 0 \right)} \left( X \right) - g_0 \left( V \right) \\
    C_2 \left( X \right) &= \pi_0 \left( X \right) - 2\overline{\pi} \left( X \right) \\
    C_3 \left( X \right) &= 1 - 2 \overline{\pi} \left( X \right) + \pi_0 \left( X \right)
\end{align*}
and therefore we have the following upper bounds
\begin{align*}
    &\norm{ C_1 \left( X \right) \left\lbrace \hat{\pi} \left( X \right) - \pi_0 \left( X \right) \right\rbrace^2 }_2 \lesssim \norm{ \hat{\pi} \left( X \right) - \pi_0 \left( X \right) }_4^2 \\
    &\norm{ C_2 \left( X \right) \left\lbrace \hat{\pi} \left( X \right) - \pi_0 \left( X \right) \right\rbrace \left\lbrace \hat{Q}^{\left( 1 \right)} \left( X \right) - Q^{\left( 1 \right)}_0 \left( X \right) \right\rbrace }_2 \lesssim \norm{ \hat{\pi} \left( X \right) - \pi_0 \left( X \right)}_4 \norm{ \hat{Q} \left( X \right) - Q_0 \left( X \right) }_4 \\
    &\norm{ C_3 \left( X \right) \left\lbrace \hat{\pi} \left( X \right) - \pi_0 \left( X \right) \right\rbrace \left\lbrace \hat{Q}^{\left( 0 \right)} \left( X \right) - Q^{\left( 0 \right)}_0 \left( X \right) \right\rbrace }_2 \lesssim \norm{ \hat{\pi} \left( X \right) - \pi_0 \left( X \right)}_4 \norm{ \hat{Q} \left( X \right) - Q_0 \left( X \right) }_4,
\end{align*}
where the inequality for the first term follows under the assumption that $C_1 \left( X\right)$ is bounded. In this case, the oracle behavior condition for the R-Learner equals
\begin{equation} \label{eqn32}
  n^{-2 \varphi_\pi} + n^{-\varphi_\pi} n^{-\varphi_{Q}} = o_p \left( n^{-\varphi_g} \right).
\end{equation}
In particular, to attain oracle behavior in the case of the parametric rate for the target parameter, we need the rate of convergence of the propensity score to be $n^{- \varphi_\pi} = o_p \left( n^{- 1/4} \right)$.  Additionally, the rates of convergence for the outcome prediction models can be slower than $n^{-1/4}$ as long as $n^{-\varphi_\pi} n^{-\varphi_{Q}} = o_p \left( n^{-1/2} \right)$. When the rate of convergence for the target parameter is slower than parametric, then also the propensity score and the conditional mean outcomes can have slower rates of convergence as long as (\ref{eqn32}) is satisfied. Importantly, values of the propensity score equal to $0$ and $1$ do not inflate the error bound for the R-Learner.

\section{Empirical Example} \label{Section4}

In our analysis we are interested in estimating the effect on $7$-day ICU mortality of initiating RRT within $24$h from the stage $2$ AKI diagnosis in the stage $\geq2$ AKI patient population.

\subsection{Data}

The ICIS (Intensive Care Information System) of the Ghent University Hospital ICUs is a database containing routinely collected medical information from all adult patients admitted to the intensive care unit since $2013$. We have focused our analysis on $3728$ adult stage $2$ and $3$ AKI patients (based on the KDIGO -- Kidney Disease: Improving Global Outcomes -- criteria \citep{KDIGO2012}) admitted to the ICU between $1/1/2013$ and $31/12/2017$, who had no recorded RRT history and no RRT restrictions by the time of the inclusion at stage $2$ AKI diagnosis. 
\\
\indent
We have information on several patient characteristics, i.e. ICU admission time, ICU discharge time, vital status at discharge, timestamps of all dialysis sessions during each ICU episode, baseline covariates (e.g. age, weight, gender, admission category \{"No surgery", "Planned surgery", "Emergency surgery"\}, an indicator whether the patient has received dialysis prior to current ICU admission, an indicator whether the patient had chronic kidney disease diagnosis prior to current ICU admission) and longitudinal measurements (e.g. SOFA scores, an indicator whether the KDIGO AKI (stage $1/2/3$) creatinine condition has been reached during ICU episode, an indicator whether the KDIGO AKI (stage $1/2/3$) oliguric condition has been reached during ICU episode, an indicator whether the patient has received diuretics during the ICU episode, cumulative total fluid intake, cumulative total fluid output, arterial pH values, serum potassium values (in mmol/L), serum ureum values (in mg/dL), serum magnesium values (in mmol/L), fraction of inspired oxygen (FiO$_2$), peripheral oxygen saturation (SpO$_2$), arterial oxygen concentration (PaO$_2$), ratio of arterial oxygen concentration to the fraction of inspired oxygen (P/F ratio), DNR ("Do Not Resuscitate") code) and their timestamps. Table \ref{table1} presents an overview of the overall observed patient population, observed treated patient population and observed patient population weighted by the propensity-overlap weights. 

\begin{table}[ht]
\begin{center}
\begin{tabular}{c|c|c|c} 
 & \shortstack{\textbf{Observed} \\ \textbf{population}} & \shortstack{\textbf{Observed} \\ \textbf{treated} \\ \textbf{population}} & \shortstack{ \textbf{POW} \\ \textbf{population}} \\
 \hline
 \hline
\textbf{Age on admission} & $64.5 \pm 14.9$ & $63.4 \pm 13.5$ & $63.4 \pm 15.9$  \\
\hline
\textbf{Weight} & $80.6 \pm 17.1$ & $77.1 \pm 13.9$ & $77.7 \pm 16.5$ \\
\hline
\textbf{Gender} & \shortstack{$31.9 \%$ \\ of females} & \shortstack{$26.9 \%$ \\ of females} & \shortstack{$30.4 \%$ \\ of females} \\
\hline
\textbf{SOFA score} & $10.9 \pm 4.0$ & $13.7 \pm 4.7$ & $12.9 \pm 4.3$ \\
\hline
\shortstack{\textbf{Time from ICU} \\ \textbf{admission to AKI} \\ \textbf{stage $2$ diagnosis (in h)}} & $33.4 \pm 68.3$ & $18.4 \pm 44.3$ & $33.3 \pm 87.1$ \\
\hline
\hline
\end{tabular}
\end{center}
\caption{Overview of the patient population in the train set (mean and standard deviation). "POW population" is the observed population weighted with the propensity-overlap weights.}  \label{table1}
\end{table}

\subsection{Results}

To estimate the effect on $7$-day ICU mortality of initiating RRT within $24$h from the stage $2$ AKI diagnosis in the stage $\geq2$ AKI patient population, we apply several learners described in the previous sections, in particular the IPW-Learner,  the DR-Learner and the R-Learner. Furthermore, when focusing on the treatment effect in the population of the AKI patients that had received the RRT, we additionally compute the propensity-score-weighted DR-Learner. As outlined in the Section \ref{subsect3}, implementation of each of the learners has been performed in two steps. In the first step, one estimates the nuisance parameters, i.e. the propensity score model and the conditional mean outcomes models. In the second step, one computes the pseudo-outcome given by (\ref{eqn24}) and regresses the pseudo-outcome on the subset of covariates $V$ using weighted regression, which corresponds to minimizing the empirical risk function (\ref{eqn26}) with plug-in estimates for nuisance parameters. In our analysis, we are interested in estimating the effect of initiating RRT conditional on  relevant patient characteristics. We have chosen to condition on values of serum potassium, arterial pH, cumulative total fluid intake and cumulative total fluid output. The rationale behind this choice was the fact that these patient characteristics are often used to define absolute indications to initiate RRT \citep{KDIGO2012, Morzywolek2022a} and as such are candidate effect modifiers.
\\
\indent
To perform our analysis, we split the data into three equally sized parts: training data set A, training data set B and test data set. The test data set is put aside and used for computation of the final output, i.e. effect on risk difference scale of initiating RRT within 24h from the stage $2$ AKI diagnosis in the stage $\geq2$ AKI patient population on $7$-day ICU mortality conditional on patient's values of serum potassium, arterial pH, cumulative total fluid intake and cumulative total fluid output. To make the most efficient use of the available data we apply the idea of cross-fitting as described in the Section \ref{subsect3}. In the final step, given the model for the CATE estimation obtained via cross-fitting on the training data sets A and B, we use the test data set to obtain the final output, i.e. effect on risk difference scale of initiating RRT within 24h from the stage $2$ AKI diagnosis in the stage $\geq2$ AKI patient population on $7$-day ICU mortality. We evaluate the performance on a separate test set to avoid possible overoptimism in the reported performance, which could arise once evaluating the performance of different methods on the same data that has been used for training the models. Note, that for computing propensity-score-weighted DR-Learner we perform the above procedure restricting the data set to the population of treated patients to guarantee that the corresponding loss function is well-defined. All the models, i.e. propensity score model, outcome prediction model, and weighted regressions for all pseudo-outcomes have been computed using \texttt{SuperLearner} \citep{VanDerLaan2007} with the following list of wrappers: \texttt{glm}, \texttt{glmnet}, random forest (\texttt{ranger}) and \texttt{xgboost}.
\\
\indent
Figure \ref{figure1} presents the histogram of the effect on the risk difference scale of initiating RRT within 24h from the stage $2$ AKI diagnosis in the stage $\geq2$ AKI patient population on $7$-day ICU mortality. Panel A presents the histogram of the treatment effect conditional on the values of serum potassium, arterial pH, fluid intake and fluid output computed in the whole patient population using IPW-Learner ("IPW"), DR-Learner ("DR") and R-Learner ("R"). Dashed vertical lines indicate the average treatment effect. The mean average treatment effect together with the standard deviation were equal $0.16 \pm 0.01$ for the IPW-Learner, $0.21 \pm 0.13$ for the DR-Learner and $0.23 \pm 0.1$ for the R-Learner. Note, that negative values represent a beneficial effect of RRT initiation, i.e. lowering $7-$day ICU mortality, for a particular patient. The IPW-Learner estimates a harmful effect of RRT initiation at considered timepoint for $100\%$ of patients in the test set, DR-Learner for $95.3\%$ and R-Learner for $98.2\%$. This is comparable with the observed proportion of patients receiving RRT at that timepoint, which was $2.9\%$. Panel B presents a histogram of the effect on the risk difference scale of initiating RRT within $24$h from the stage $2$ AKI diagnosis in the stage $\geq2$ AKI patient population that received the RRT on $7$-day ICU mortality. Similarly as in Panel A, the presented treatment effects are conditional on the values of serum potassium, arterial pH, fluid intake and fluid output. For the treatment effect in the treated, additionally to IPW-Learner, DR-Learner and R-Learner, we also applied the propensity-score-weighted DR-Learner ("psDR"). These learners gave predictions with mean plus or minus standard deviation equal $0.16 \pm 0.01$ for the IPW-Learner, $0.16 \pm 0.16$ for the DR-Learner, $0.15 \pm 0.1$ for the R-Learner and $0.15 \pm 0.18$ for the propensity-score-weighted DR-Learner. Note, that since in the empirical example we are interested in estimation of CATE conditional on a subset of available covariates we do not present results of T-Learner as it targets a different estimand, i.e. CATE conditional on all available covariates. 
\\
\indent
Figure \ref{figure2} presents the relationship between the treatment effects as presented in the Figure \ref{figure1} and the propensity scores for each patient from the test set. As in Figure \ref{figure1}, we present the results for the whole patient population in Panel A and for the treated patient population in Panel B. Dashed horizontal lines indicate the average treatment effect for each of the methods.
\\
\indent
Note that the unrealistically small standard deviation for the IPW-Learner suggests poor behavior of this learner. Furthermore, IPW-Learner and DR-Learner are more prone to the occurrence of extreme weights than other methods, which leads to instability of the results. In particular, due to propensity scores being close to $0$ or $1$ in our example the results obtained using IPW-Learner and DR-Learner were highly unstable between different runs of the analysis. See Figure \ref{figure5} and Figure \ref{figure6} in the Appendix \ref{subsect5} for comparison of the inverse-probability and propensity-overlap weights for the empirical example.
\begin{landscape}
\begin{figure}[ht]
\begin{center}
\includegraphics[scale=0.97]{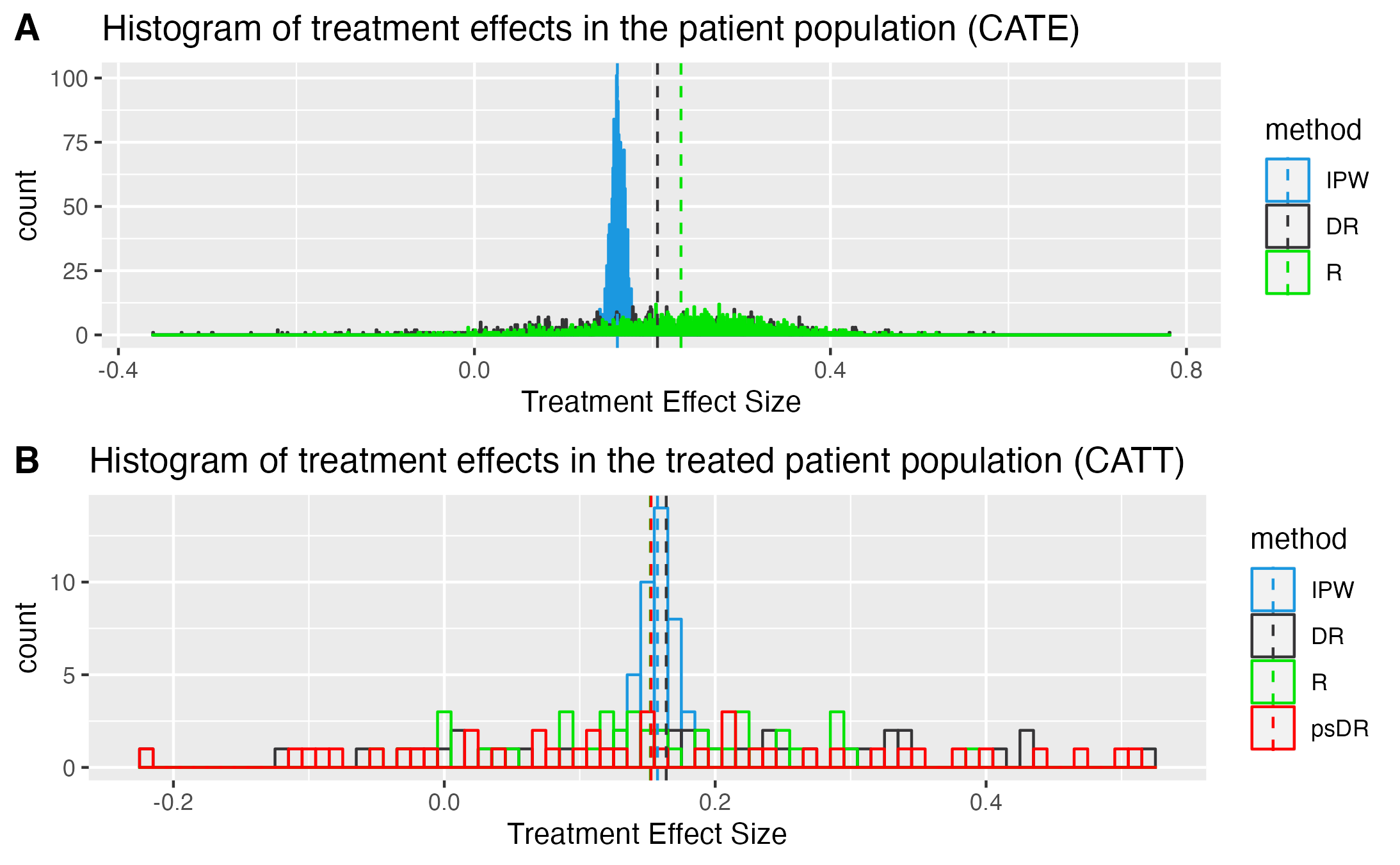}
\end{center}
\caption{Panel A presents the histogram of the effect on the risk difference scale of initiating RRT within $24$h from the stage $2$ AKI diagnosis in the stage $\geq2$ AKI patient population on $7$-day ICU mortality conditional on the values of serum potassium, arterial pH, fluid intake and fluid output computed in the whole patient population using IPW-Learner ("IPW"), DR-Learner ("DR") and R-Learner ("R"). Panel B presents a histogram of the effect on the risk difference scale of initiating RRT within $24$h from the stage $2$ AKI diagnosis in the stage $\geq2$ AKI patient population that received the RRT on $7$-day ICU mortality. Additionally to IPW-Learner, DR-Learner and R-Learner, we also present propensity-score-weighted DR-Learner ("psDR"). Dashed vertical lines indicate the average treatment effect for each of the methods.} \label{figure1}
\end{figure}
\end{landscape}

\begin{landscape}
\begin{figure}[ht]
\begin{center}
\includegraphics[scale=0.97]{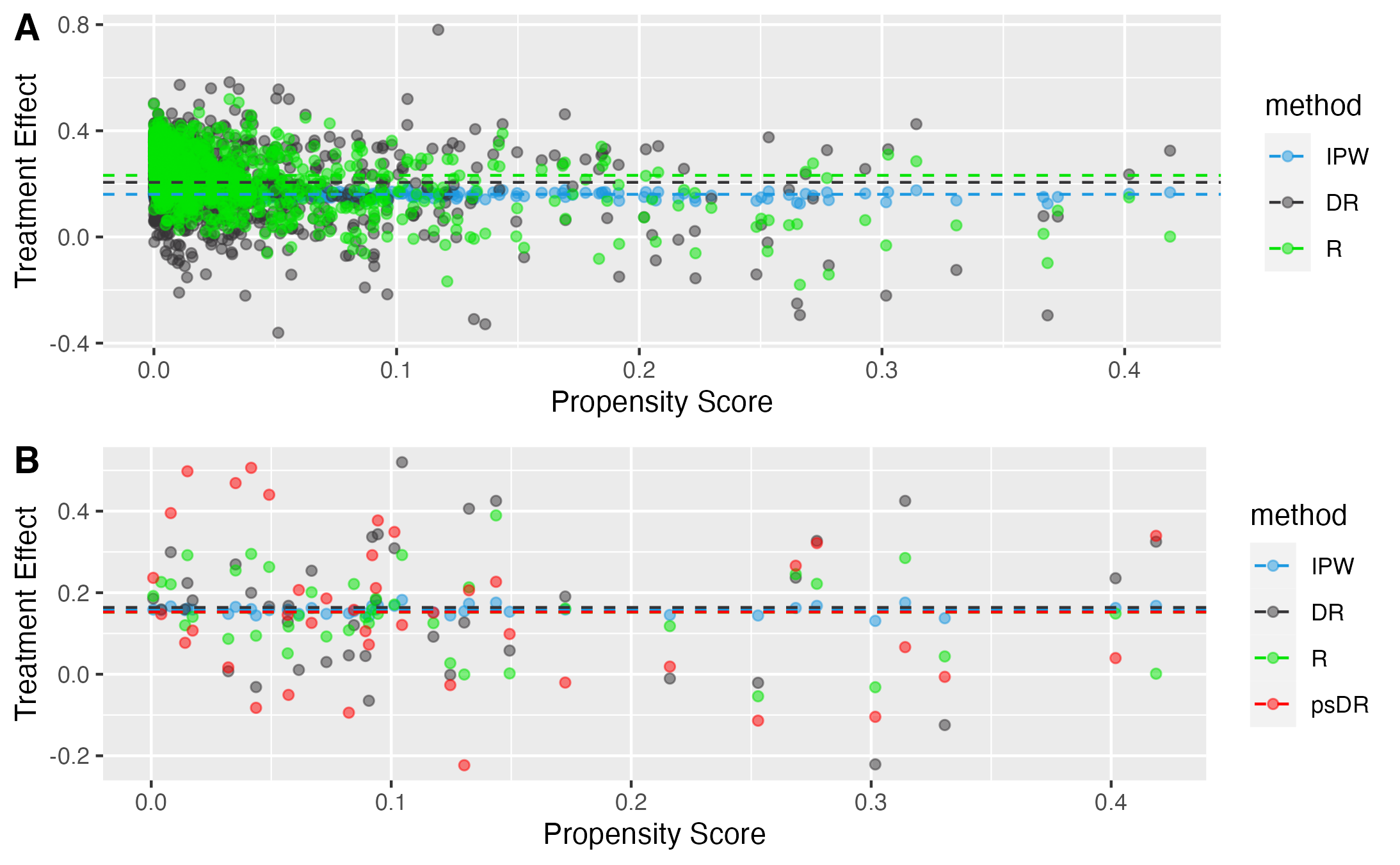}
\end{center}
\caption{Relationship between the effect on the risk difference scale of initiating RRT within 24h from the stage $2$ AKI diagnosis in the stage $\geq2$ AKI patient population on $7$-day ICU mortality and the propensity score of different patients. Similarly as in the the Figure \ref{figure1} in Panel A we present the results in the whole patient population and in Panel B we present the results in the treated patient population. Dashed horizontal lines indicate the average treatment effect for each of the methods.} \label{figure2}
\end{figure}
\end{landscape}

\section{Simulation Study} \label{Section5}

We perform a simulation study to assess the performance of the different learners described in the previous sections. We consider two simulation set-ups inspired by the simulations in \cite{Nie2021} and a simulation set-up from \cite{Belloni2017}. To perform our simulation study, similarly as in the case of the empirical example, we split the available data into three parts: training data set A, training data set B and a test data set, each of sample size $500$. A test data set is put aside and used for computation of the final output.  The results are computed using $1000$ simulations. For the purpose of making the most efficient use of the available data we apply cross-fitting. In what follows we consider three different performance metrics computed on the test set, i.e. mean-squared error, mean-squared error in the treated population and propensity-overlap-weighted mean-squared error, which is defined as
\begin{align*}
    MSE_{pow} = \frac{\sum_{i=1}^n \pi \left( X_i \right) \left\lbrace 1 - \pi \left( X_i \right) \right\rbrace \left\lbrace \hat{\tau} \left( X_i \right) - \tau \left( X_i \right) \right\rbrace^2}{\sum_{i=1}^n \pi \left( X_i \right) \left\lbrace 1 - \pi \left( X_i \right) \right\rbrace}.
\end{align*}
Note, that for computing the propensity-score-weighted DR-Learner we perform the above procedure restricting the data set to the population of treated patients to guarantee that the corresponding loss function is well-defined. Furthermore, note that in the simulations we compute the CATE conditional on the complete set of available covariates $X$. All the models, i.e. propensity score model, outcome prediction model, and weighted regressions for all pseudo-outcomes have been computed using \texttt{SuperLearner} with the following list of wrappers: \texttt{glm}, \texttt{glmnet}, random forest (\texttt{ranger}) and \texttt{xgboost}. 
\\
\indent
Simulation set-ups $1$ and $2$ are based on the following data generating mechanism from \cite{Nie2021}: $x_i \sim P_d$, $a_i \cond x_i \sim \text{Bern} \left\lbrace \pi_0 \left( x_i \right) \right\rbrace$, $\delta_i \cond x_i \sim \mathcal{N} \left( 0, 1 \right)$, $y_i = b \left( x_i \right) + \left( a_i - 0.5 \right) \tau \left( x_i \right) + \sigma \delta_i$, where $b \left( \cdot \right)$ is baseline main effect and $P_d$ is a distribution of the covariates $x$ indexed by the dimension $d = 20$ and the noise level $\sigma = 0.5$. 
\\
\indent
Simulation set-up $1$ is inspired by simulation set-up C in \cite{Nie2021}.  We use $b \left( x_i \right) = 2 \log \left\lbrace 1 + \exp \left( x_{i1} + x_{i2} + x_{i3} + x_{i4} + x_{i5} \right) \right\rbrace$, $x_i \sim \mathcal{N} \left( 0,  I_{d \times d} \right)$, where $I_{d \times d}$ is a $d-$dimensional identity matrix, $\pi \left( x_i \right) = 1 / \exp \left(1 + x_{i2} + X_{i3} + x_{i4} + x_{i5} \right)$ and $\tau \left( x_i \right) = 1$. In this set-up  the CATE is a simple function. The results are presented in Figure~\ref{figure3}. Panels A and B show the mean-squared error and the propensity-overlap-weighted mean-squared error computed on the test set using T-Learner, IPW-Learner, DR-Learner and R-Learner. Panel C presents the mean-squared error in the treated computed using the propensity-score-weighted DR-Learner, additionally to previously mentioned methods. Panels D, E and F present the same results as the corresponding plots above, however with restricted y-axis for a better interpretability of the results. The IPW-Learner performs much worse than other methods, whereas R-Learner is the best performing among considered methods. Furthermore, DR-Learner and R-Learner outperform T-Learner. This could be expected since the CATE is a simple function in this set-up, therefore the methods that directly aim at CATE estimation have advantage over T-Learner that models two outcomes separately (as discussed in Section \ref{subsect4}). Furthermore, the outcome generating process is fairly complex in this set-up, which may negatively impact performance of T-Learner. The propensity-score-weighted DR-Learner performs worse than the T-Learner, DR-Learner and R-Learner, which may be related to the fact that it is trained on a much smaller population of treated patients. 
\\
\indent
Simulation set-up $2$ is inspired by simulation set-up A in \cite{Nie2021}.  We use $b \left( x_i \right) = \sin \left( \pi x_{i1} x_{i2} \right) + 2 \left( x_{i3} - 0.5 \right)^2 + x_{i4} + 0.5 x_{i5} + x_{i6}$, $x_i \sim \mathcal{U} \left( 0,1 \right)^d$, $\pi \left( x_i \right) = \max \left[ \alpha,  \min \left\lbrace \sin \left( \pi x_{i1} x_{i2} x_{i3} x_{i4} \right), 1 - \alpha \right\rbrace \right]$ with $\alpha = 0.1$ and $\tau \left( x_i \right) = \left( x_{i1} + x_{i2} + x_{i3} \right)/2$. The results are presented in Figure~\ref{figure4}, whose  description is analogous to that of Figure~\ref{figure3}. Also in this set-up, the IPW-Learner performs much worse than other methods. The remaining methods perform similarly with R-Learner performing slightly better than other methods for each of the considered metrics. Interestingly, DR-Learner performs worse than T-Learner in this set-up and has very variable performance. This shows that DR-Learner suffers a lot, when the propensity score is difficult to estimate as in this simulation set-up.  
\\
\indent
Simulation set-up $3$ uses the data-generating mechanism from \cite{Belloni2017}. It is defined as 
\begin{align*}
    a_i &= \mathbbm{1} \left\lbrace \frac{\exp \left( c_d x_i^\prime \beta\right)}{1 + \exp \left( c_d x_i^\prime \beta\right)} > v_i \right\rbrace, \text{ } v_i \sim \mathcal{U} \left( 0,1 \right) \\
    y_i &= \theta a_i + c_y x_i^\prime \beta a_i + \zeta_i, \text{ } \zeta_i \sim \mathcal{N} \left( 0,1 \right),
\end{align*}
such that $x_i \sim \mathcal{N} \left( 0, \Sigma \right)$ where $\Sigma$ is a matrix with entries $\Sigma_{kj} = 0.5^{\lvert k-j \rvert}$. Furthermore, $\beta$ is a $d$-dimensional vector with entries $\beta_j = \frac{1}{j^2}$ and constants $c_y$ and $c_d$ are given by $c_y =\sqrt{\frac{R_y^2}{\left( 1-R_y^2 \right) \beta^\prime \Sigma \beta}}$ and $c_d =\sqrt{\frac{\left( \pi^2/3 \right) R_d^2}{\left( 1-R_d^2 \right) \beta^\prime \Sigma \beta}}$ with $R_y^2 = 0.5$ and $R_d^2 = 0.5$. The results are presented in Figure~\ref{figure4b}, whose  description is analogous to that of Figure~\ref{figure3}. In this set-up, T-Learner performs the worst and R-Learner performs the best among the considered methods. Interestingly, DR-Learner has even slightly worse performance than IPW-Learner in this setting. 
\\
\indent
Figure \ref{figure4a} presents the inverse-probability weights from three considered simulation set-ups.

\begin{landscape}
\begin{figure}
\begin{center}
\includegraphics[scale=0.57]{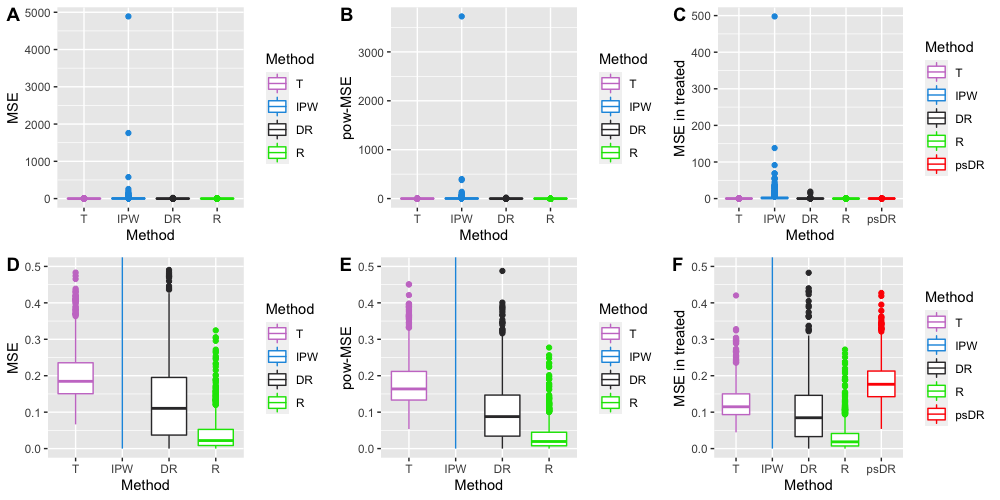}
\end{center}
\caption{Results of the simulation set-up $1$. Panels A and B show the mean-squared error and the propensity-overlap-weighted mean-squared error computed on the test set in the second simulation setting using T-Learner ("T"), IPW-Learner ("IPW"), DR-Learner ("DR") and R-Learner ("R"). Panel C presents the mean-squared error in the treated computed using the propensity-score-weighted DR-Learner ("psDR"), additionally to previously mentioned methods. Panels D, E and F present the same results as the corresponding plots above, however with a restricted y-axis for a better interpretability of the results. \label{figure3}}
\end{figure}
\end{landscape}

\begin{landscape}
\begin{figure}
\begin{center}
\includegraphics[scale=0.57]{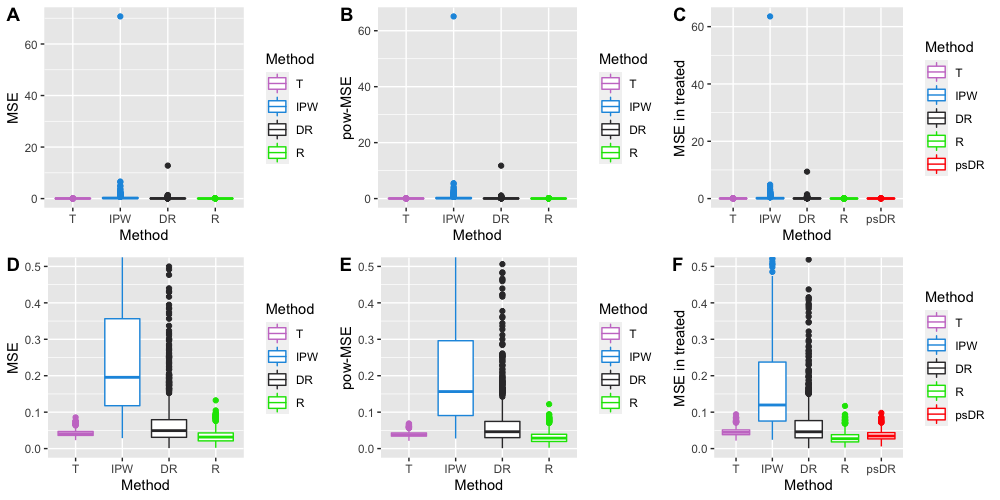}
\end{center}
\caption{Results of the simulation set-up $2$. Description of the figure is analogous to that of Figure \ref{figure3}.}  \label{figure4}
\end{figure}
\end{landscape}

\begin{landscape}
\begin{figure}
\begin{center}
\includegraphics[scale=0.57]{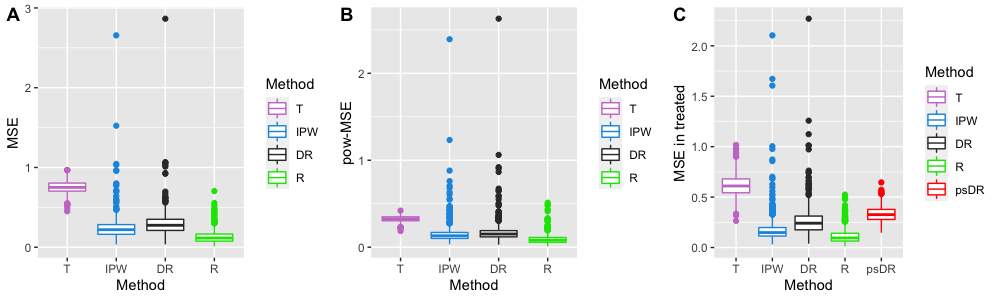}
\end{center}
\caption{Results of the simulation set-up $3$. Description of the figure is analogous to that of Figure \ref{figure3}.}  \label{figure4b}
\end{figure}
\end{landscape}

\begin{landscape}
\begin{figure}
\begin{center}
\includegraphics[scale=0.57]{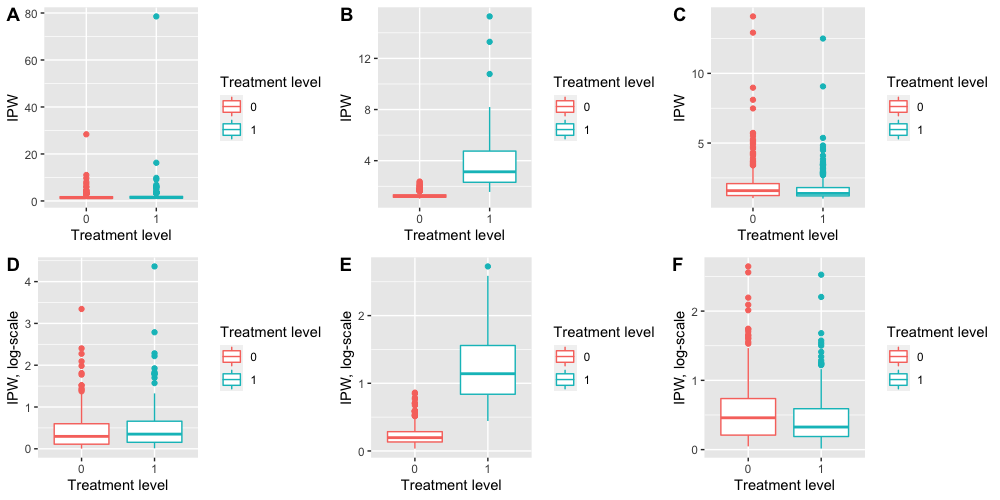}
\end{center}
\caption{Panels A, B and C show inverse-probability weights in the simulation set-ups $1$, $2$ and $3$, respectively. Panels C, D and E show corresponding inverse-probability weights on the log scale.} 
\label{figure4a}
\end{figure}
\end{landscape}

\section{Discussion} \label{Section6}

In this paper we have presented a class of weighted Neyman-orthogonal learners for heterogeneous treatment effects. Some of the recently developed approaches for CATE estimation, e.g. DR-Learner and R-Learner, can be viewed as special cases resulting from the presented framework for particular choices of the weights. Viewing different learners in the light of a general class of weighted learners allows to gain insights into differences and similarities between them. Furthermore, relying on the Neyman-orthogonality of the derived loss functions and leveraging the recently developed literature on the orthogonal statistical learning \citep{Foster2023} we provided error bounds for CATE estimates resulting from minimization of
the considered weighted loss functions.  
\\
\indent
We consider the problem of CATE estimation in two distinct cases, namely when the conditioning set of covariates $V$ is sufficient for confounding adjustment, i.e. $V = X$ and when $V$ is a strict subset of $X$, i.e. when the conditioning set $V$ is not sufficient for confounding adjustment. In the first case all weighted loss functions target the same estimand, i.e. CATE, as long as CATE belongs to the class of functions over which we are minimizing the loss function. The situation is different when the assumed simple form of the CATE is not met or when $V$ is a strict subset of $X$. In this case the choice of weights defines the estimand that is being targeted. Therefore, different weighted loss functions aim to estimate different quantities. 
\\
\indent
Additionally to allowing comparisons of existing methods we also proposed a new learner, which is based on the weight function being the propensity-score and which corresponds to DR-Learner restricted to the population of treated patients. Such a learner can be of particular interest in retrospective studies, where the focus is on the treatment effects in treated. Furthermore, one could consider other choices of the weight functions leading to new learners, e.g. $\omega \left( x \right) = 1 - \pi \left( x \right)$, which gives rise to a DR-Learner restricted to the control group \citep{Li2018}. \cite{Crump2006} consider weight functions given by an indicator function $\omega \left( x \right) = \mathbbm{1} \left\lbrace x \in \mathcal{A} \right\rbrace$, where $\mathcal{A}$ is a closed subset of the covariate space $\mathcal{X}$. In particular, one could set $\mathcal{A} = \left\lbrace x \in \mathcal{X} \cond \alpha \leq \pi \left( x \right) \leq 1 - \alpha \right\rbrace$, which leads to restricting the analysis to patients having the propensity score within the range $\left[ \alpha, 1 - \alpha \right]$ for a chosen level of $\alpha$. To avoid problems with differentiability in  the construction of an orthogonal learner, one could instead approximate the indicator function by well-chosen sigmoid functions, e.g.
\begin{align*}
    \frac{1}{1+e^{-\theta \left( x - \alpha \right)}} - \frac{1}{1+e^{-\theta \left( x - (1-\alpha) \right)}},
\end{align*}
where the approximation improves as the parameter $\theta$ increases. 
\\
\indent
One of the limitations of our approach focusing on estimating heterogeneous treatment effects is the fact that it may still provide too limited information to a decision-maker, not solely because of the difficulty in obtaining confidence intervals on the proposed estimates of the CATE. Imagine a situation, where a clinician needs to make a decision whether to initiate an invasive treatment for a particular patient (e.g. whether to initiate RRT for an AKI patient). It could potentially happen that for a given patient there is a slight benefit from applying the treatment, hence we would obtain a positive CATE estimate. However, when the patient's outcome is already satisfactory without treatment initiation, then the risks and costs of applying the treatment could potentially outweigh the benefit of the treatment. Therefore, ideally one would wish to provide to the clinician a prediction for patient's outcome under each hypothetical intervention, which is also referred to in the literature as counterfactual or causal prediction \citep{Hernan2019, VanGeloven2020, Coston2020}. Our proposed framework likewise suggests a class of weighted learners for predictions under hypothetical interventions. We make an outline of the proposal for orthogonal loss functions for weighted mean of the potential outcome in the Appendix \ref{AppA0}. 
\\
\indent
Moreover, we discuss in the Appendix \ref{AppA0} two additional examples, namely orthogonal loss function with weights depending on nuisance parameters other than the propensity score, which allows additional interesting choices of the weight functions, e.g. transportability weights. The second extension concerns the application of orthogonal loss functions for estimation of  conditional estimands for time-to-event outcomes. Detailed study of these is beyond the scope of this work.

\section*{Acknowledgments}

The authors would like to thank Vasilis Syrgkanis and Wim Van Biesen for helpful discussions and Bram Gadeyne, Veerle Brams, Christian Danneels and Johan Steen for technical support. The authors were supported by the Flemish Research Council (FWO Research Project 3G068619 — Grant FWO.OPR.2019.0045.01).

\section*{Data Availability Statement}
According to GDPR rules, and in line with the informed consent and the approval of the ethics committee of the Ghent University Hospital, full, non-aggregated data can only be made available by the authors after permission of the ethics committee. Interested parties can address the authors to discuss an eventual project proposal to be submitted to the ethics committee.

\section*{Supporting information}

Additional supporting information, including proofs of the main results, may be found in the Appendix. The \texttt{R} code to reproduce our simulations is available at \\
\texttt{https://github.com/pmorzywolek/CATEsimulations}.

\bibliography{CATE}
\bibliographystyle{apalike}

\newpage

\appendix

\section{Additional examples} \label{AppA0}

\subsection{Orthogonal loss functions for a weighted mean of the potential outcome} \label{AppA0subsect1}

In the main text we have considered a class of weighted loss functions for CATE estimation. One can analogously construct a class of function for mean of the potential outcomes under different interventions. We consider $n$ independent and identically distributed observations $Z_i \coloneqq \left( Y_i, A_i, X_i \right)$, for $i = 1, \dots,n$, where $X_i \in \mathcal{X} \subseteq \mathbb{R}^{d_X}$ is a minimal set of covariates sufficient for confounding adjustment of a binary treatment $A_i \in \left\lbrace 0, 1 \right\rbrace$ on outcome $Y_i \in \mathbb{R}$. Let $V \in \mathcal{V} \subseteq \mathbb{R}^{d_V}$ be a subset of the features $X$, i.e. $d_V \leq d_X$. We wish to estimate the conditional mean of the potential outcome $Y^1$ given by 
\begin{equation*} 
\mathbb{E} \left( Y^1 \cond V = v \right),
\end{equation*}
which is the minimizer of the following population risk function 
\begin{align*}
\mathcal{L} \left( g \right) = \mathbb{E} \left[ \left\lbrace Y^1 - g \left( V \right) \right\rbrace^2 \right].
\end{align*}
Similarly as in the case of CATE, as argued in Section \ref{subsect8}, we may rather be interested in minimizing a retargeted version of the loss function given by
\begin{align} \label{eqn28}
\mathcal{L} \left( g; \omega \right) = \mathbb{E} \left[ \omega \left( X \right) \left\lbrace Y^1 - g \left( V \right) \right\rbrace^2 \right]
\end{align}
over the function space $\mathcal{G}$, for each choice of the weight function $\omega \left( \cdot \right)$. 
\\
\indent
To construct an observed data loss function that approximates  (\ref{eqn28}) and is Neyman-orthogonal, we follow the procedure outlined in Section \ref{subsect1}. As a first step we choose a finite-dimensional estimand - weighted mean potential outcome of $Y^1$, i.e. $\frac{ \mathbb{E} \left\lbrace \omega \left( X \right) Y^1 \right\rbrace }{ \mathbb{E} \left\lbrace \omega \left( X \right)  \right\rbrace }$ where $\omega \left( X \right) \coloneqq \lambda \left\lbrace \pi \left( X \right) \right\rbrace$. Its efficient influence function (EIF) is given by:
\begin{align} \label{eqn27}
&\phi \left( Z; \eta, \lambda \left\lbrace \pi \left( X \right) \right\rbrace \right) \\
\nonumber
&= \frac{ \lambda \left\lbrace \pi \left( X \right) \right\rbrace }{ \mathbb{E} \left[ \lambda \left\lbrace \pi \left( X \right) \right\rbrace \right] } \left( \frac{A}{\pi \left( X \right)} \left\lbrace Y - Q^{\left( 1 \right)} \left( X \right) \right\rbrace + \left[ \left\lbrace A - \pi \left( X \right) \right\rbrace \frac{ \lambda^\prime \left\lbrace \pi \left( X \right) \right\rbrace }{\lambda \left\lbrace \pi \left( X \right) \right\rbrace} + 1 \right] \left\lbrace Q^{\left( 1 \right)} \left( X \right) - g \right\rbrace \right)
\end{align}
where $\eta = \left( \pi \left( X \right), Q^{\left( 1 \right)} \left( X \right) \right)$ is the nuisance parameter. If $\rho \left( A, \pi \left( X \right) \right) \coloneqq \left\lbrace A - \pi \left( X \right) \right\rbrace \lambda^\prime \left\lbrace \pi \left( X \right) \right\rbrace + \lambda \left\lbrace \pi \left( X \right) \right\rbrace \neq 0$, then we can rewrite (\ref{eqn27}) as: 
\begin{align*}
&\phi \left( Z; \eta, \lambda \left\lbrace \pi \left( X \right) \right\rbrace \right) = \frac{ \rho \left( A, \pi \left( X \right) \right)}{ \mathbb{E} \left[ \lambda \left\lbrace \pi \left( X \right) \right\rbrace \right] } \left\lbrace \varphi \left( Z; \eta,  \lambda \left\lbrace \pi \left( X \right) \right\rbrace \right) - g \right\rbrace
\end{align*}
where
\begin{align*} 
\nonumber
&\varphi \left( Z; \eta, \lambda \left\lbrace \pi \left( X \right) \right\rbrace \right) \coloneqq \frac{\lambda \left\lbrace \pi \left( X \right) \right\rbrace}{ \rho \left( A, \pi \left( X \right) \right)} \left[ \frac{A}{\pi \left( X \right)} \left\lbrace Y - Q^{\left( 1 \right)} \left( X \right) \right\rbrace \right] + Q^{\left( 1 \right)} \left( X \right).
\end{align*} 
Analogously as in Section 3.1 in the case of the CATE we obtain the following loss function for estimation of the conditional mean of the potential outcome $Y^1$:
\begin{align*} 
&\mathcal{L} \left( g, \eta, \lambda \left\lbrace \pi \left( X \right) \right\rbrace \right) = \frac{1}{\mathbb{E} \left[ \lambda \left\lbrace \pi \left( X \right) \right\rbrace \right]} \mathbb{E} \left[ \rho \left( A, \pi \left( X \right) \right) \left\lbrace \varphi \left( Z; \eta, \lambda \left( \pi \right) \right) - g \left( V \right) \right\rbrace^2 \right].
\end{align*}
Note that analogously we can represent the estimation of the conditional mean of the potential outcome $Y^0$.

\subsection{Transportability}

In the main text we have focused on known weight functions depending on the (possibly unknown) propensity score, i.e. weight functions of the form $\omega \left( \cdot \right) = \lambda \left\lbrace \pi \left( \cdot \right) \right\rbrace$. 
Alternatively, one could consider $\omega$ to be directly a function of covariates $X$ and nuisance parameters other than the propensity score. This allows consideration of additional interesting choices of weight function, e.g. transportability weights. 
\\
\indent
Following \cite{Dahabreh2020a} we consider extending inferences from a randomized trial to a new target population in a nested trial design. In particular, let $X \in \mathcal{X}$ denote baseline covariates, $A \in \left\lbrace 0,1 \right\rbrace$ a binary assigned treatment, $Y \in \mathbb{R}$ the outcome and $S \in \left\lbrace 0,1 \right\rbrace$ the trial participation indicator with $1$ denoting trial participants and $0$ denoting nonparticipants. The estimand of interest is $\mathbb{E} \left( Y^a \cond V, S=0 \right)$ for $V \subseteq X$, i.e. the mean potential outcome under the intervention $A=a$ in the target population conditional on the subset of covariates $V$. The data consists of independent and identically distributed tuples $Z_i \coloneqq \left( X_i, S_i, S_i \times A_i, S_i \times Y_i \right)$, $i=1,\dots,n$, such that $n_{RCT} = \sum_{i=1}^n \mathbbm{1} \left( S_i = 1 \right)$ and $n_{obs} = \sum_{i=1}^n \mathbbm{1} \left( S_i = 0 \right)$, where $\mathbbm{1} \left( \cdot \right)$ is an indicator function. In particular, the treatment and outcome information is available only for trial participants. 
\\
\indent
Similarly as in the main text, to contruct an appropriate loss function for $\mathbb{E} \left( Y^a \cond V, S=0 \right)$ we first consider the finite-dimensional estimand $g_a \coloneqq \mathbb{E} \left( Y^a \cond S=0 \right)$, i.e. the mean potential outcome under the intervention $A=a$ in the target population. Under standard assumptions \citep{Dahabreh2020a} it can be shown that it can be identified as a weighted estimand $\frac{\mathbb{E} \left\lbrace \omega \left( S, X \right) Y^a \right\rbrace}{\mathbb{E} \left\lbrace \omega \left( S, X \right) \right\rbrace}$, for $a \in \left\lbrace 0,1 \right\rbrace$, with a weight function given by 
\begin{align*}
    \omega \left( S, X \right) \coloneqq S \frac{\mathbb{P} \left( S=0 \cond X \right)}{\mathbb{P} \left( S=1 \cond X \right)}.
\end{align*}
Indeed,
\begin{align*}
    \mathbb{E} \left( Y^a \cond S=0 \right) &= \mathbb{E} \left\lbrace \mathbb{E} \left( Y^a \cond X, S=0 \right) \cond S=0 \right\rbrace \\
    &= \mathbb{E} \left\lbrace \mathbb{E} \left( Y^a \cond X, S=1 \right) \cond S=0 \right\rbrace \\
    &= \frac{1}{\mathbb{P} \left( S=0 \right)}\mathbb{E} \left\lbrace \left( 1-S \right) \mathbb{E} \left( Y^a \cond X, S=1 \right) \right\rbrace \\
    &= \frac{1}{\mathbb{P} \left( S=0 \right)}\mathbb{E} \left\lbrace \mathbb{P} \left( S=0 \cond X \right) \mathbb{E} \left( Y^a \cond X, S=1 \right) \right\rbrace \\
    &= \frac{1}{\mathbb{P} \left( S=0 \right)}\mathbb{E} \left\lbrace \mathbb{P} \left( S=0 \cond X \right) \frac{\mathbb{E} \left( S Y^a \cond X \right)}{\mathbb{P} \left( S=1 \cond X \right)} \right\rbrace \\
    &= \frac{1}{\mathbb{P} \left( S=0 \right)} \mathbb{E} \left\lbrace S \frac{\mathbb{P} \left( S=0 \cond X \right)}{\mathbb{P} \left( S=1 \cond X \right)} \mathbb{E} \left( Y^a \cond X \right) \right\rbrace \\
    &= \frac{1}{\mathbb{P} \left( S=0 \right)} \mathbb{E} \left[ \mathbb{E} \left\lbrace S \frac{\mathbb{P} \left( S=0 \cond X \right)}{\mathbb{P} \left( S=1 \cond X \right)} Y^a \cond X \right\rbrace \right] \\
    &= \frac{1}{\mathbb{P} \left( S=0 \right)} \mathbb{E} \left\lbrace S \frac{\mathbb{P} \left( S=0 \cond X \right)}{\mathbb{P} \left( S=1 \cond X \right)} Y^a \right\rbrace
\end{align*}
Furthermore, we have
\begin{align*}
    \mathbb{P} \left( S=0 \right) &= \mathbb{E} \left\lbrace \mathbb{P} \left( S=0 \cond X \right) \right\rbrace \\
    &= \mathbb{E} \left\lbrace \frac{\mathbb{P} \left( S=0 \cond X \right)}{\mathbb{P} \left( S=1 \cond X \right)} \mathbb{E} \left( S \cond X \right) \right\rbrace \\
    &= \mathbb{E} \left[ \mathbb{E} \left\lbrace S \frac{\mathbb{P} \left( S=0 \cond X \right)}{\mathbb{P} \left( S=1 \cond X \right)} \cond X \right\rbrace \right] \\
    &= \mathbb{E} \left\lbrace S \frac{\mathbb{P} \left( S=0 \cond X \right)}{\mathbb{P} \left( S=1 \cond X \right)} \right\rbrace.
\end{align*}
As a result we obtain
\begin{align*}
    \mathbb{E} \left( Y^a \cond S=0 \right) &= \frac{\mathbb{E} \left\lbrace S \frac{\mathbb{P} \left( S=0 \cond X \right)}{\mathbb{P} \left( S=1 \cond X \right)} Y^a \right\rbrace}{\mathbb{E} \left\lbrace S \frac{\mathbb{P} \left( S=0 \cond X \right)}{\mathbb{P} \left( S=1 \cond X \right)} \right\rbrace}.
\end{align*}
The estimand has the following efficient influence function \citep{Dahabreh2020a}:
\begin{align*}
    \phi \left( Z; \eta, \omega \left( S, X \right) \right) = \frac{1}{\mathbb{P} \left( S=0 \right)} &\left[ S \frac{\mathbb{P} \left( S=0 \cond X \right)}{\mathbb{P} \left( S=1 \cond X \right)}\frac{\mathbbm{1} \left( A=a \right)}{\mathbb{P} \left( A=a \cond S=1, X \right)} \left\lbrace Y - \mathbb{E} \left( Y \cond X, S=1, A=a \right) \right\rbrace \right. \\
    &\left.+ \left( 1-S \right) \left\lbrace \mathbb{E} \left( Y \cond X, S=1, A=a \right) - g_a \right\rbrace \right],
\end{align*}
where $\eta \coloneqq \left( \mathbb{P} \left( S=0 \cond X \right), \mathbb{P} \left( S=1, A=a \cond X \right), \mathbb{E} \left( Y \cond X, S=1, A=a \right) \right)$ is the nuisance parameter. Allowing $g_a$ to be function of covariates $V$ we have the following
\begin{align*}
    \mathbb{E} &\left[ S \frac{\mathbb{P} \left( S=0 \cond X \right)}{\mathbb{P} \left( S=1 \cond X \right)}\frac{\mathbbm{1} \left( A=a \right)}{\mathbb{P} \left( A=a \cond S=1, X \right)} \left\lbrace Y - \mathbb{E} \left( Y \cond X, S=1, A=a \right) \right\rbrace \right. \\
    &\left.+ \left( 1-S \right) \left\lbrace \mathbb{E} \left( Y \cond X, S=1, A=a \right) - g_a \left( V \right) \right\rbrace \right] \\
    &=\mathbb{E}\left[ S \frac{\mathbb{P} \left( S=0 \cond X \right)}{\mathbb{P} \left( S=1 \cond X \right)}\frac{\mathbbm{1} \left( A=a \right)}{\mathbb{P} \left( A=a \cond S=1, X \right)} \left\lbrace Y - \mathbb{E} \left( Y \cond X, S=1, A=a \right) \right\rbrace \right. \\
    &\left.+ \mathbb{P} \left( S=0 \cond X \right) \left\lbrace \mathbb{E} \left( Y \cond X, S=1, A=a \right) - g_a \left( V \right) \right\rbrace\right. \\
    &\left.
    - \left\{S-\mathbb{P} \left( S=1 \cond X \right)\right\} \left\lbrace \mathbb{E} \left( Y \cond X, S=1, A=a \right) - g_a \left( V \right) \right\rbrace\right]
\end{align*}
This suggests a loss function of the form 
\begin{align*}
&\mathcal{L} \left( g, \eta, \omega \left( S, X \right) \right) = \frac{1}{\mathbb{E} \left[ \omega \left( S, X \right) \right]} \mathbb{E} \left[ \omega \left( S, X \right) \left\lbrace \varphi_a \left( Z; \eta \right) - g_a \left( V \right) \right\rbrace^2 \right],
\end{align*}
where
\begin{align*}
    \omega \left( S, X \right) = \mathbb{P} \left( S=0 \cond X \right)
\end{align*}
and
\begin{align*}
    \varphi_a \left( Z; \eta\right) \coloneqq & \frac{S}{\mathbb{P} \left( S=1 \cond X \right)}\frac{\mathbbm{1} \left( A=a \right)}{\mathbb{P} \left( A=a \cond S=1, X \right)} \left\lbrace Y - \mathbb{E} \left( Y \cond X, S=1, A=a \right) \right\rbrace  \\
    &+ \mathbb{E} \left( Y \cond X, S=1, A=a \right),
\end{align*}
but requires infinite-dimensional targeting of the probabilities  $\mathbb{P} \left( S=1 \cond X \right)$ to ensure that 
\[
E\left[\left\{S-\mathbb{P} \left( S=1 \cond X \right)\right\} \left\lbrace \mathbb{E} \left( Y \cond X, S=1, A=a \right) - g_a \left( V \right) \right\rbrace\right] \]
is close to zero for all functions $g_a \left( V \right) $ in the considered function class (see \cite{vansteelandt2023orthogonal} for an illustration of such targeting).

\subsection{Orthogonal loss functions for time-to-event outcomes}

We consider $n$ independent and identically distributed observations $Z_i \coloneqq \left( X_i, A_i, Y_i, \Delta_i \right)$, for $i = 1, \dots,n$, where $X_i \in \mathcal{X} \subseteq \mathbb{R}^{d_X}$ is set of baseline covariates, $A_i \in \left\lbrace 0, 1 \right\rbrace$ is a binary treatment, $
Y_i \coloneqq \min \left( T_i, C_i \right)$ with $T_i \in \left( 0, \infty \right]$ being the observed event time and $C_i \in \left[ 0, \infty \right]$ being the right-censoring time. Furthermore, let $\Delta_i \coloneqq \mathbbm{1} \left( T_i \leq C_i \right)$ and $V \in \mathcal{V} \subseteq \mathbb{R}^{d_V}$ be a subset of the features $X$, i.e. $d_V \leq d_X$. We wish to estimate the population probability that the individual would experience the event time later than time $t$ conditional on the subset of covariates $V$, if the individual would receive (possibly contrary to the fact) the treatment, i.e. 
\begin{equation*} 
\mathbb{P} \left( T^1 > t \cond V = v \right),
\end{equation*}
where $T^1$ is the event time of interest under assignment to exposure $A=1$, which is the minimizer of the following population risk function 
\begin{align} \label{eqn29}
\mathcal{L} \left( g \right) = \mathbb{E} \left[ \left\lbrace \mathbb{P} \left( T^1 > t \cond X \right) - \theta \left( V \right) \right\rbrace^2 \right].
\end{align}
To construct an observed data loss function that approximates  (\ref{eqn29}) and is Neyman-orthogonal, we follow the procedure outlined in the Section \ref{subsect1}. As a first step we choose a finite-dimensional estimand - the population probability that the individual would experience the event time later than time $t$, if the individual would receive (possibly contrary to the fact) the treatment $A=1$, i.e. $\mathbb{P} \left( T^1 > t \right)$. As described in \cite{Westling2023}, under suitable causal assumptions $\mathbb{P} \left( T^1 > t \right)$ can be identified as $\mathbb{E} \left\lbrace S_0 \left( t \cond 1, X \right) \right\rbrace$, where 
\begin{align*}
    S_0 \left( t \cond a, x \right) \coloneqq \Prodi_{\left( 0,t \right]} \left\lbrace 1 - \Lambda_0 \left( du \cond a,x \right) \right\rbrace,
\end{align*}
with $\Prodi$ being the Riemann-Stieltjes product integral \citep{Gill1990} and $\Lambda_0 \left( t \cond a, x \right) \coloneqq \int_0^t \frac{F_{0,1} \left( du \cond a, x \right)}{R_0 \left( t \cond a, x \right)}$ for $F_{0,\delta} \left( du \cond a, x \right)\coloneqq \mathbb{P} \left( Y \leq t, \Delta = \delta \cond A=a, X=x \right)$ and \\ $R_0 \left( t \cond a, x \right) \coloneqq \mathbb{P} \left( Y \geq t \cond A=a, X=x \right)$. The efficient influence function of $\mathbb{E} \left\lbrace S_0 \left( t \cond 1, X \right) \right\rbrace$ is given by:
\begin{align*}
\phi_{t} \left( z \right) &= S_0 \left( t \cond 1, x \right) \times \\
\nonumber
&\times \left[ 1 - \frac{\mathbbm{1} \left( a=1 \right)}{\pi_0 \left( a \cond x \right)} \left\lbrace \frac{\mathbbm{1} \left( y \leq t, \delta = 1 \right)}{S_0 \left( y \cond a,x \right) G_0 \left( y \cond a,x \right)} \right\rbrace  - \int_0^{t \wedge y} \frac{\Lambda_0 \left( du \cond a,x \right)}{S_0 \left( y \cond a,x \right) G_0 \left( y \cond a,x \right)} \right],
\end{align*}
where $\pi_0 \left( a \cond x \right) \coloneqq \mathbb{P} \left( A=a \cond X=x \right)$ \citep{Westling2023}. Analogously as described in Section \ref{subsect1}, given the EIF we can construct the loss function for estimation of $\mathbb{P} \left( T^1 > t \cond V \right)$, i.e.
\begin{align*} 
&\mathcal{L} \left( g, \eta \right) = \mathbb{E} \left[ \left\lbrace \phi_{t} \left( Z \right) - g \left( V \right) \right\rbrace^2 \right].
\end{align*}

\newpage 

\section{Proofs} \label{AppA}

\begin{proof}[Proof of Lemma \ref{lemma1}]
We have the following:
\begin{align*}
&\mathbb{E} \left[ \left\lbrace \frac{AY}{\pi_0 \left( X \right)} - \frac{\left( 1-A \right) Y}{1- \pi_0 \left( X \right)} - g \left( V \right) \right\rbrace^2 \right] \\
=&\mathbb{E} \left[ \left\lbrace \frac{AY}{\pi_0 \left( X \right)} - \frac{\left( 1-A \right) Y}{1- \pi_0 \left( X \right)} - \left( Y^1 - Y^0 \right) + \left( Y^1 - Y^0 \right) - g \left( V \right) \right\rbrace^2 \right] \\
=&\mathbb{E} \left[ \left\lbrace \frac{AY}{\pi_0 \left( X \right)} - \frac{\left( 1-A \right) Y}{1- \pi_0 \left( X \right)} - \left( Y^1 - Y^0 \right) \right\rbrace^2 \right] + \mathbb{E} \left[ \left\lbrace \left( Y^1 - Y^0 \right) - g \left( V \right) \right\rbrace^2 \right] \\
+& 2\mathbb{E} \left[ \left\lbrace \frac{AY}{\pi_0 \left( X \right)} - \frac{\left( 1-A \right) Y}{1- \pi_0 \left( X \right)} - \left( Y^1 - Y^0 \right) \right\rbrace \left\lbrace \left( Y^1 - Y^0 \right) - g \left( V \right) \right\rbrace \right].
\end{align*}
Furthermore, we can show that:
\begin{align*}
&\mathbb{E} \left[ \left\lbrace \frac{AY}{\pi_0 \left( X \right)} - \frac{\left( 1-A \right) Y}{1- \pi_0 \left( X \right)} - \left( Y^1 - Y^0 \right) \right\rbrace  g \left( V \right) \right] \\
&= \mathbb{E} \left( \mathbb{E} \left[ \left\lbrace \frac{AY}{\pi_0 \left( X \right)} - \frac{\left( 1-A \right) Y}{1- \pi_0 \left( X \right)} - \left( Y^1 - Y^0 \right) \right\rbrace  g \left( V \right) \cond V \right] \right) \\
&= \mathbb{E} \left\lbrace \mathbb{E} \left( \mathbb{E} \left[ \left\lbrace \frac{AY}{\pi_0 \left( X \right)} - \frac{\left( 1-A \right) Y}{1- \pi_0 \left( X \right)} - \left( Y^1 - Y^0 \right) \right\rbrace  g \left( V \right) \cond X \right] \cond V \right) \right\rbrace \\
&= \mathbb{E} \left[ g \left( V \right) \left\lbrace \mathbb{E} \left( \mathbb{E} \left[ \left\lbrace \frac{AY}{\pi_0 \left( X \right)} - \frac{\left( 1-A \right) Y}{1- \pi_0 \left( X \right)} \right\rbrace \cond X \right] \cond V \right) - \mathbb{E} \left( Y^1 - Y^0 \cond V \right) \right\rbrace \right] \\
&= \mathbb{E} \left( g \left( V \right) \left[ \mathbb{E} \left\lbrace \mathbb{E} \left( Y^1 - Y^0 \cond X \right) \cond V \right\rbrace - \mathbb{E} \left( Y^1 - Y^0 \cond V \right) \right] \right) \\
&= \mathbb{E} \left[ g \left( V \right) \left\lbrace \mathbb{E} \left( Y^1 - Y^0 \cond V \right) - \mathbb{E} \left( Y^1 - Y^0 \cond V \right) \right\rbrace \right] = 0.
\end{align*}
Therefore, we have shown that 
\begin{align*}
\mathbb{E} \left[ \left\lbrace \frac{AY}{\pi_0 \left( X \right)} - \frac{\left( 1-A \right) Y}{1- \pi_0 \left( X \right)} - g \left( V \right) \right\rbrace^2 \right] &= \mathbb{E} \left[ \left\lbrace \left( Y^1 - Y^0 \right) - g \left( V \right) \right\rbrace^2 \right] \\ 
&+ \mathbb{E} \left[ \left\lbrace \frac{AY}{\pi_0 \left( X \right)} - \frac{\left( 1-A \right) Y}{1- \pi_0 \left( X \right)} - \left( Y^1 - Y^0 \right) \right\rbrace^2 \right] \\
&+ 2\mathbb{E} \left[ \left\lbrace \frac{AY}{\pi_0 \left( X \right)} - \frac{\left( 1-A \right) Y}{1- \pi_0 \left( X \right)} - \left( Y^1 - Y^0 \right) \right\rbrace \left( Y^1 - Y^0 \right) \right].
\end{align*}
Note that the second and third term do not depend on $g$, hence do not have impact on the minimization problem.  Therefore, it follows that the loss functions (\ref{eqn1}) and (\ref{eqn2}) have the same minimizer. 
\end{proof}
 
\begin{proof}[Proof of Lemma \ref{lemma2}]
Since $V = X$, we have the following:
\begin{align*}
&\mathbb{E} \left[ \omega \left( X \right) \left\lbrace \left( Y^1 - Y^0 \right) - g \left( X \right) \right\rbrace^2 \right] \\
&= \mathbb{E} \left( \mathbb{E} \left[ \omega \left( X \right) \left\lbrace \left( Y^1 - Y^0 \right) - g \left( X \right) \right\rbrace^2 \cond X \right] \right) \\
&= \mathbb{E} \left( \omega \left( X \right) \mathbb{E} \left[ \left\lbrace \left( Y^1 - Y^0 \right) - \tau \left( X \right) + \tau \left( X \right) - g \left( X \right) \right\rbrace^2 \cond X \right] \right) \\
&= \mathbb{E} \left\lbrace \omega \left( X \right) \left( \mathbb{E} \left[ \left\lbrace \left( Y^1 - Y^0 \right) - \tau \left( X \right) \right\rbrace^2 \cond X \right] + \mathbb{E} \left[ \left\lbrace \tau \left( X \right) - g \left( X \right) \right\rbrace^2 \cond X \right] \right. \right. \\
&\left. \left. + 2 \mathbb{E} \left[ \left\lbrace \left( Y^1 - Y^0 \right) - \tau \left( X \right) \right\rbrace \left\lbrace \tau \left( X \right) - g \left( X \right) \right\rbrace \cond X \right] \right) \right\rbrace \\
&= \mathbb{E} \left( \omega \left( X \right) \mathbb{E} \left[ \left\lbrace \left( Y^1 - Y^0 \right) - \tau \left( X \right) \right\rbrace^2 \cond X \right] \right) + \mathbb{E} \left[ \omega \left( X \right) \left\lbrace \tau \left( X \right) - g \left( X \right) \right\rbrace^2 \right],
\end{align*}
where the last equality follows using the law of iterated expectations. The first term does not contain $g$, therefore is not relevant for the minimization. The second term, and therefore the whole expression, is minimized for any choice of weight function $\omega \left( \cdot \right)$ at $g \left( X \right) = \tau \left( X \right)$.
\end{proof}

\begin{proof}[Proof of Lemma \ref{lemma4}]
Calculate the pathwise derivative with respect to the first argument (target parameter):
\begin{align*} 
&D_g \mathcal{L} \left( g_0, \eta_0, \lambda \left\lbrace \pi_0 \left( X \right) \right\rbrace \right) \left[ g - g_0 \right] = \\
&\frac{d}{dt} \left[ \mathbb{E} \left\lbrace \left[ \left\lbrace A - \pi_0 \left( X \right) \right\rbrace \lambda^\prime \left\lbrace \pi_0 \left( X \right) \right\rbrace + \lambda \left\lbrace \pi_0 \left( X \right) \right\rbrace \right] \left( \varphi_0 \left( Z; \eta, \lambda \left( \pi \right) \right) - \left[ g_0 \left( V \right) + t \left\lbrace g \left( V \right) - g_0 \left( V \right) \right\rbrace \right] \right)^2 \right\rbrace \right] \Bigr|_{t=0} \\
&= -2 \mathbb{E} \left\lbrace \left[ \left\lbrace A - \pi_0 \left( X \right) \right\rbrace \lambda^\prime \left\lbrace \pi_0 \left( X \right) \right\rbrace + \lambda \left\lbrace \pi_0 \left( X \right) \right\rbrace \right] \left( \varphi_0 \left( Z; \eta, \lambda \left( \pi \right) \right) - \left[ g_0 \left( V \right) + t \left\lbrace g \left( V \right) - g_0 \left( V \right) \right\rbrace \right] \right) \right.\\
&\left. \times \left\lbrace g \left( V \right) - g_0 \left( V \right) \right\rbrace \right\rbrace \Bigr|_{t=0} \\
&= -2 \mathbb{E} \left( \left[ \left\lbrace A - \pi_0 \left( X \right) \right\rbrace \lambda^\prime \left\lbrace \pi_0 \left( X \right) \right\rbrace + \lambda \left\lbrace \pi_0 \left( X \right) \right\rbrace \right] \left\lbrace \varphi_0 \left( Z; \eta, \lambda \left( \pi \right) \right) - g_0 \left( V \right) \right\rbrace \left\lbrace g \left( V \right) - g_0 \left( V \right) \right\rbrace \right) \\
&= -2 \mathbb{E} \left\lbrace \left[ \left\lbrace A - \pi_0 \left( X \right) \right\rbrace \lambda^\prime \left\lbrace \pi_0 \left( X \right) \right\rbrace + \lambda \left\lbrace \pi_0 \left( X \right) \right\rbrace \right] \left( \frac{\lambda \left\lbrace \pi_0 \left( X \right) \right\rbrace}{ \left\lbrace A - \pi_0 \left( X \right) \right\rbrace \lambda^\prime \left\lbrace \pi_0 \left( X \right) \right\rbrace + \lambda \left\lbrace \pi_0 \left( X \right) \right\rbrace } \right.\right.\\
&\left.  \times \left[ \frac{A}{\pi_0 \left( X \right)} \left\lbrace Y - Q^{\left( 1 \right)}_0 \left( X \right) \right\rbrace - \frac{1-A}{1-\pi_0 \left( X \right)} \left\lbrace Y - Q^{\left( 0 \right)}_0 \left( X \right) \right\rbrace \right] + Q^{\left( 1 \right)}_0 \left( X \right) - Q^{\left( 0 \right)}_0 \left( X \right) - g_0 \left( V \right) \right) \\
&\left. \times \left\lbrace g \left( V \right) - g_0 \left( V \right) \right\rbrace \right\rbrace \\
&= -2 \mathbb{E} \left\lbrace \left( \lambda \left\lbrace \pi_0 \left( X \right) \right\rbrace \left[ \frac{A}{\pi_0 \left( X \right)} \left\lbrace Y - Q^{\left( 1 \right)}_0 \left( X \right) \right\rbrace - \frac{1-A}{1-\pi_0 \left( X \right)} \left\lbrace Y - Q^{\left( 0 \right)}_0 \left( X \right) \right\rbrace \right] \right. \right. \\
&\left. \left. +  \left[ \left\lbrace A - \pi_0 \left( X \right) \right\rbrace \lambda^\prime \left\lbrace \pi_0 \left( X \right) \right\rbrace + \lambda \left\lbrace \pi_0 \left( X \right) \right\rbrace \right] \left\lbrace Q^{\left( 1 \right)}_0 \left( X \right) - Q^{\left( 0 \right)}_0 \left( X \right) - g_0 \left( V \right) \right\rbrace \right) \left\lbrace g \left( V \right) - g_0 \left( V \right) \right\rbrace \right\rbrace.
\end{align*}
Calculate pathwise derivative with respect to $Q^{\left( 1 \right)}$:
\begin{align*}
&D_{Q^{\left( 1 \right)}} D_g \mathcal{L} \left( g_0, \eta_0, \lambda \left( \pi_0 \left( X \right) \right) \right) \left[ g - g_0,  Q^{\left( 1 \right)} - Q^{\left( 1 \right)}_0 \right] \\
&= -2 \frac{d}{dt} \mathbb{E} \left( \left[ \lambda \left\lbrace \pi_0 \left( X \right) \right\rbrace \left\lbrace \frac{A}{\pi_0 \left( X \right)} \left( Y - \left[ Q^{\left( 1 \right)}_0 \left( X \right) + t \left\lbrace Q^{\left( 1 \right)} \left( X \right) - Q^{\left( 1 \right)}_0 \left( X \right) \right\rbrace \right] \right) \right. \right. \right. \\
&\left. - \frac{1-A}{1-\pi_0 \left( X \right)} \left\lbrace Y - Q^{\left( 0 \right)}_0 \left( X \right) \right\rbrace \right\rbrace + \left[ \left\lbrace A - \pi_0 \left( X \right) \right\rbrace \lambda^\prime \left\lbrace \pi_0 \left( X \right) \right\rbrace + \lambda \left\lbrace \pi_0 \left( X \right) \right\rbrace \right] \\
&\left.\left. \times \left( \left[ Q^{\left( 1 \right)}_0 \left( X \right) + t \left\lbrace Q^{\left( 1 \right)} \left( X \right) - Q^{\left( 1 \right)}_0 \left( X \right) \right\rbrace \right] - Q^{\left( 0 \right)}_0 \left( X \right) - g_0 \left( V \right) \right) \right] \left\lbrace g \left( V \right) - g_0 \left( V \right) \right\rbrace \right) \Bigr|_{t=0} \\
&= 2 \mathbb{E} \left\lbrace \left( \frac{A \lambda \left\lbrace \pi_0 \left( X \right) \right\rbrace }{ \pi_0 \left( X \right) } -  \left[ \left\lbrace A - \pi_0 \left( X \right) \right\rbrace \lambda^\prime \left\lbrace \pi_0 \left( X \right) \right\rbrace + \lambda \left\lbrace \pi_0 \left( X \right) \right\rbrace \right] \right) \left\lbrace Q^{\left( 1 \right)} \left( X \right) - Q^{\left( 1 \right)}_0 \left( X \right) \right\rbrace \right. \\
&\left. \times \left\lbrace g \left( V \right) - g_0 \left( V \right) \right\rbrace \right\rbrace \\
&= 0.
\end{align*}
Calculate pathwise derivative with respect to $Q^{\left( 0 \right)}$:
\begin{align*}
&D_{Q^{\left( 0 \right)}} D_g \mathcal{L} \left( g_0, \eta_0, \lambda \left( \pi_0 \right) \right) \left[ g - g_0,  Q^{\left( 0 \right)} - Q^{\left( 0 \right)}_0 \right] = -2 \frac{d}{dt} \mathbb{E} \left( \left[ \lambda \left\lbrace \pi_0 \left( X \right) \right\rbrace \left\lbrace \frac{A}{\pi_0} \left\lbrace Y - Q^{\left( 1 \right)}_0 \left( X \right) \right\rbrace \right.\right.\right. \\
&\left.- \frac{1-A}{1-\pi_0 \left( X \right)} \left( Y - \left[ Q^{\left( 0 \right)}_0 \left( X \right) + t \left\lbrace Q^{\left( 0 \right)} \left( X \right) - Q^{\left( 0 \right)}_0 \left( X \right) \right\rbrace \right] \right) \right\rbrace \\
&\left. \left. + \left[ \left\lbrace A - \pi_0 \left( X \right) \right\rbrace \lambda^\prime \left\lbrace \pi_0 \left( X \right) \right\rbrace + \lambda \left\lbrace \pi_0 \left( X \right) \right\rbrace \right] \left( Q^{\left( 1 \right)}_0 \left( X \right) - \left[ Q^{\left( 0 \right)}_0 \left( X \right) + t \left\lbrace Q^{\left( 0 \right)} \left( X \right) - Q^{\left( 0 \right)}_0 \left( X \right) \right\rbrace \right]  - g_0 \left( V \right) \right) \right] \right.\\
&\left. \times \left\lbrace g \left( V \right) - g_0 \left( V \right) \right\rbrace \right) \Bigr|_{t=0} \\
&= 2 \mathbb{E} \left\lbrace \left( \frac{ \left( 1 - A \right) \lambda \left\lbrace \pi_0 \left( X \right) \right\rbrace }{ 1 - \pi_0 \left( X \right) } - \left[ \left\lbrace A - \pi_0 \left( X \right) \right\rbrace \lambda^\prime \left\lbrace \pi_0 \left( X \right) \right\rbrace + \lambda \left\lbrace \pi_0 \left( X \right) \right\rbrace \right] \right) \left\lbrace Q^{\left( 0 \right)} \left( X \right) - Q^{\left( 0 \right)}_0 \left( X \right) \right\rbrace \right.\\
&\left. \times \left\lbrace g \left( V \right) - g_0 \left( V \right) \right\rbrace \right\rbrace \\
&= 0.
\end{align*}
Calculate pathwise derivative with respect to $\pi$:
\begin{align*}
&D_{\pi} D_g \mathcal{L} \left( g_0, \eta_0,  \lambda \left( \pi_0 \right) \right) \left[ g - g_0,  \pi - \pi_0 \right] = -2 \frac{d}{dt} \mathbb{E} \left( \left[ \lambda \left[ \pi_0 \left( X \right) + t \left\lbrace \pi \left( X \right) - \pi_0 \left( X \right) \right\rbrace \right] \right.\right.\\
&\times \left[ \frac{A}{\pi_0 \left( X \right) + t \left\lbrace \pi \left( X \right) - \pi_0 \left( X \right) \right\rbrace} \left\lbrace Y - Q^{\left( 1 \right)}_0 \left( X \right) \right\rbrace - \frac{1-A}{1- \left[ \pi_0 \left( X \right) + t \left\lbrace \pi \left( X \right) - \pi_0 \left( X \right) \right\rbrace \right]} \left\lbrace Y - Q^{\left( 0 \right)}_0 \left( X \right) \right\rbrace \right] \\
&+  \left\lbrace \left( A - \left[ \pi_0 \left( X \right) + t \left\lbrace \pi \left( X \right) - \pi_0 \left( X \right) \right\rbrace \right] \right) \lambda^\prime \left[ \pi_0 \left( X \right) + t \left\lbrace \pi \left( X \right) - \pi_0 \left( X \right) \right\rbrace \right] \right. \\
&\left.\left.\left.+ \lambda \left[ \pi_0 \left( X \right) + t \left\lbrace \pi \left( X \right) - \pi_0 \left( X \right) \right\rbrace \right] \right\rbrace \left\lbrace Q^{\left( 1 \right)}_0 \left( X \right) - Q^{\left( 0 \right)}_0 \left( X \right) - g_0 \left( V \right) \right\rbrace \right] \left\lbrace g \left( V \right) - g_0 \left( V \right) \right\rbrace \right) \Bigr|_{t=0} \\
&= -2  \mathbb{E} \left[ \lambda^\prime \left\lbrace \pi_0 \left( X \right) \right\rbrace \left[ \frac{A}{\pi_0 \left( X \right)} \left\lbrace Y - Q^{\left( 1 \right)}_0 \left( X \right) \right\rbrace - \frac{1-A}{1- \pi_0 \left( X \right)} \left\lbrace Y - Q^{\left( 0 \right)}_0 \left( X \right) \right\rbrace \right] \left\lbrace \pi \left( X \right) - \pi_0 \left( X \right) \right\rbrace \right. \\
&\left. \times \left\lbrace g \left( V \right) - g_0 \left( V \right) \right\rbrace \right] \\
&+ 2 \mathbb{E} \left[ \lambda \left\lbrace \pi_0 \left( X \right)  \right\rbrace \left[ \frac{A}{\pi_0^2 \left( X \right) } \left\lbrace Y - Q^{\left( 1 \right)}_0 \left( X \right) \right\rbrace + \frac{ 1-A }{\left\lbrace 1 - \pi_0 \left( X \right)  \right\rbrace^2 } \left\lbrace Y - Q^{\left( 0 \right)}_0 \left( X \right) \right\rbrace \right] \left\lbrace \pi \left( X \right) - \pi_0 \left( X \right) \right\rbrace \right. \\
&\left. \times \left\lbrace g \left( V \right) - g_0 \left( V \right) \right\rbrace \right] \\
&- 2 \mathbb{E} \left[ \left\lbrace A - \pi_0 \left( X \right) \right\rbrace \lambda^{\prime\prime} \left\lbrace \pi_0 \left( X \right) \right\rbrace \left\lbrace Q^{\left( 1 \right)}_0 \left( X \right) - Q^{\left( 0 \right)}_0 \left( X \right) - g_0 \left( V \right) \right\rbrace \left\lbrace \pi \left( X \right) - \pi_0 \left( X \right) \right\rbrace \left\lbrace g \left( V \right) - g_0 \left( V \right) \right\rbrace \right] \\
&= 0.
\end{align*}
Therefore, the loss function (\ref{eqn5}) is an orthogonal loss function.
\end{proof}

\begin{proof}[Proof of Lemma \ref{lemma3}]
We have the following:
\begin{align*}
&\mathbb{E} \left( \left[ \left\lbrace A - \pi_0 \left( X \right) \right\rbrace \lambda^\prime \left\lbrace \pi_0 \left( X \right) \right\rbrace + \lambda \left\lbrace \pi_0 \left( X \right) \right\rbrace \right] \left\lbrace \varphi \left( Z; \eta_0, \lambda \left( \pi_0 \right) \right) - g \left( V \right) \right\rbrace^2 \right) \\
&= \mathbb{E} \left( \left[ \left\lbrace A - \pi_0 \left( X \right) \right\rbrace \lambda^\prime \left\lbrace \pi_0 \left( X \right) \right\rbrace + \lambda \left\lbrace \pi_0 \left( X \right) \right\rbrace \right] \left\lbrace \varphi \left( Z; \eta_0, \lambda \left( \pi_0 \right) \right) - \left( Y^1 - Y^0 \right) + \left( Y^1 - Y^0 \right) - g \left( V \right) \right\rbrace^2 \right) \\
&= \mathbb{E} \left( \left[ \left\lbrace A - \pi_0 \left( X \right) \right\rbrace \lambda^\prime \left\lbrace \pi_0 \left( X \right) \right\rbrace + \lambda \left\lbrace \pi_0 \left( X \right) \right\rbrace \right] \left\lbrace \varphi \left( Z; \eta_0, \lambda \left( \pi_0 \right) \right) - \left( Y^1 - Y^0 \right) \right\rbrace^2 \right) \\
&+ \mathbb{E} \left( \left[ \left\lbrace A - \pi_0 \left( X \right) \right\rbrace \lambda^\prime \left\lbrace \pi_0 \left( X \right) \right\rbrace + \lambda \left\lbrace \pi_0 \left( X \right) \right\rbrace \right] \left\lbrace \left( Y^1 - Y^0 \right) - g \left( V \right) \right\rbrace^2 \right) \\
&+ 2 \mathbb{E} \left( \left[ \left\lbrace A - \pi_0 \left( X \right) \right\rbrace \lambda^\prime \left\lbrace \pi_0 \left( X \right) \right\rbrace + \lambda \left\lbrace \pi_0 \left( X \right) \right\rbrace \right] \left\lbrace \varphi \left( Z; \eta_0, \lambda \left( \pi_0 \right) \right) - \left( Y^1 - Y^0 \right) \right\rbrace \left\lbrace \left( Y^1 - Y^0 \right) - g \left( V \right) \right\rbrace \right).
\end{align*}
Furthermore, we can show that:
\begin{align*}
&\mathbb{E} \left( \left[ \left\lbrace A - \pi_0 \left( X \right) \right\rbrace \lambda^\prime \left\lbrace \pi_0 \left( X \right) \right\rbrace + \lambda \left\lbrace \pi_0 \left( X \right) \right\rbrace \right] \left\lbrace \varphi \left( Z; \eta_0, \lambda \left( \pi_0 \right) \right) - \left( Y^1 - Y^0 \right) \right\rbrace g \left( V \right) \right) \\
&= \mathbb{E} \left\lbrace \left[ \left\lbrace A - \pi_0 \left( X \right) \right\rbrace \lambda^\prime \left\lbrace \pi_0 \left( X \right) \right\rbrace + \lambda \left\lbrace \pi_0 \left( X \right) \right\rbrace \right] \left( \frac{\lambda \left\lbrace \pi_0 \left( X \right) \right\rbrace}{ \left\lbrace A - \pi_0 \left( X \right) \right\rbrace \lambda^\prime \left\lbrace \pi_0 \left( X \right) \right\rbrace + \lambda \left\lbrace \pi_0 \left( X \right) \right\rbrace} \right.\right.\\
&\left.\left. \times \left[ \frac{A}{\pi_0 \left( X \right) } \left\lbrace Y - Q_0^{\left( 1 \right)} \left( X \right) \right\rbrace - \frac{1 - A}{1 - \pi_0 \left( X \right) } \left\lbrace Y - Q_0^{\left( 0 \right)} \left( X \right) \right\rbrace \right] + Q_0^{\left( 1 \right)} \left( X \right) - Q_0^{\left( 0 \right)} \left( X \right) - \left( Y^1 - Y^0 \right) \right) g \left( V \right) \right\rbrace \\
&= \mathbb{E} \left( \lambda \left\lbrace \pi_0 \left( X \right) \right\rbrace \left[ \frac{A}{\pi_0 \left( X \right) } \left\lbrace Y - Q_0^{\left( 1 \right)} \left( X \right) \right\rbrace - \frac{1 - A}{1 - \pi_0 \left( X \right) } \left\lbrace Y - Q_0^{\left( 0 \right)} \left( X \right) \right\rbrace \right.\right. \\
&\left.\left. + Q_0^{\left( 1 \right)} \left( X \right) - Q_0^{\left( 0 \right)} \left( X \right) - \left( Y^1 - Y^0 \right) \right] g \left( V \right) \right) \\
&+ \mathbb{E} \left[ \left\lbrace A - \pi_0 \left( X \right) \right\rbrace \lambda^\prime \left\lbrace \pi_0 \left( X \right) \right\rbrace \left\lbrace Q_0^{\left( 1 \right)} \left( X \right) - Q_0^{\left( 0 \right)} \left( X \right) \right\rbrace g \left( V \right) \right] \\
&- \mathbb{E} \left[ A \lambda^\prime \left\lbrace \pi_0 \left( X \right) \right\rbrace \left( Y^1 - Y^0 \right) g \left( V \right) \right] \\
&+ \mathbb{E} \left[ \pi_0 \left( X \right) \lambda^\prime \left\lbrace \pi_0 \left( X \right) \right\rbrace \left( Y^1 - Y^0 \right) g \left( V \right) \right].
\end{align*}
Let's consider each of the terms. The first term:
\begin{align*}
&\mathbb{E} \left( \lambda \left\lbrace \pi_0 \left( X \right) \right\rbrace \left[ \frac{A}{\pi_0 \left( X \right) } \left\lbrace Y - Q_0^{\left( 1 \right)} \left( X \right) \right\rbrace - \frac{1 - A}{1 - \pi_0 \left( X \right) } \left\lbrace Y - Q_0^{\left( 0 \right)} \left( X \right) \right\rbrace \right.\right. \\
&\left.\left. + Q_0^{\left( 1 \right)} \left( X \right) - Q_0^{\left( 0 \right)} \left( X \right) - \left( Y^1 - Y^0 \right) \right] g \left( V \right) \right) \\
&= \mathbb{E} \left\lbrace \mathbb{E} \left( \lambda \left\lbrace \pi_0 \left( X \right) \right\rbrace \left[ \frac{A}{\pi_0 \left( X \right) } \left\lbrace Y - Q_0^{\left( 1 \right)} \left( X \right) \right\rbrace - \frac{1 - A}{1 - \pi_0 \left( X \right) } \left\lbrace Y - Q_0^{\left( 0 \right)} \left( X \right) \right\rbrace \right.\right. \right. \\
&\left.\left. \left. + Q_0^{\left( 1 \right)} \left( X \right) - Q_0^{\left( 0 \right)} \left( X \right) - \left( Y^1 - Y^0 \right) \right] g \left( V \right) \cond X \right) \right\rbrace \\
&= \mathbb{E} \left( \lambda \left\lbrace \pi_0 \left( X \right) \right\rbrace g \left( V \right) \mathbb{E} \left[ \frac{A}{\pi_0 \left( X \right) } \left\lbrace Y - Q_0^{\left( 1 \right)} \left( X \right) \right\rbrace - \frac{1 - A}{1 - \pi_0 \left( X \right) } \left\lbrace Y - Q_0^{\left( 0 \right)} \left( X \right) \right\rbrace \right.\right. \\
&\left.\left. + Q_0^{\left( 1 \right)} \left( X \right) - Q_0^{\left( 0 \right)} \left( X \right) - \left( Y^1 - Y^0 \right) \cond X \right] \right) \\
&= \mathbb{E} \left[ \lambda \left\lbrace \pi_0 \left( X \right) \right\rbrace g \left( V \right) \left\lbrace \mathbb{E} \left( Y^1 - Y^0 \cond X \right) - \mathbb{E} \left( Y^1 - Y^0 \cond X \right) \right\rbrace \right] = 0.
\end{align*}
The second term:
\begin{align*}
&\mathbb{E} \left[ \left\lbrace A - \pi_0 \left( X \right) \right\rbrace \lambda^\prime \left\lbrace \pi_0 \left( X \right) \right\rbrace \left\lbrace Q_0^{\left( 1 \right)} \left( X \right) - Q_0^{\left( 0 \right)} \left( X \right) \right\rbrace g \left( V \right) \right] \\
&= \mathbb{E} \left( \mathbb{E} \left[ \left\lbrace A - \pi_0 \left( X \right) \right\rbrace \lambda^\prime \left\lbrace \pi_0 \left( X \right) \right\rbrace \left\lbrace Q_0^{\left( 1 \right)} \left( X \right) - Q_0^{\left( 0 \right)} \left( X \right) \right\rbrace g \left( V \right) \cond X \right] \right) = 0.
\end{align*}
The third term:
\begin{align*}
& \mathbb{E} \left[ A \lambda^\prime \left\lbrace \pi_0 \left( X \right) \right\rbrace \left( Y^1 - Y^0 \right) g \left( V \right) \right] \\
&= \mathbb{E} \left( \mathbb{E} \left[ A \lambda^\prime \left\lbrace \pi_0 \left( X \right) \right\rbrace \left( Y^1 - Y^0 \right) g \left( V \right) \cond X\right] \right) \\
&= \mathbb{E} \left[ \pi_0 \left( X \right) \lambda^\prime \left\lbrace \pi_0 \left( X \right) \right\rbrace \left( Y^1 - Y^0 \right) g \left( V \right) \right],
\end{align*}
which equals the forth term, hence they cancel each other out. Therefore, we have shown that:
\begin{align*}
&\mathbb{E} \left( \left[ \left\lbrace A - \pi_0 \left( X \right) \right\rbrace \lambda^\prime \left\lbrace \pi_0 \left( X \right) \right\rbrace + \lambda \left\lbrace \pi_0 \left( X \right) \right\rbrace \right] \left\lbrace \varphi \left( Z; \eta_0, \lambda \left( \pi_0 \right) \right) - \left( Y^1 - Y^0 \right) \right\rbrace g \left( V \right) \right) =0.
\end{align*}
Furthermore, we have 
\begin{align*}
&\mathbb{E} \left( \left[ \left\lbrace A - \pi_0 \left( X \right) \right\rbrace \lambda^\prime \left\lbrace \pi_0 \left( X \right) \right\rbrace + \lambda \left\lbrace \pi_0 \left( X \right) \right\rbrace \right] \left\lbrace \left( Y^1 - Y^0 \right) - g \left( V \right) \right\rbrace^2 \right) \\
&= \mathbb{E} \left\lbrace \mathbb{E} \left( \left[ \left\lbrace A - \pi_0 \left( X \right) \right\rbrace \lambda^\prime \left\lbrace \pi_0 \left( X \right) \right\rbrace + \lambda \left\lbrace \pi_0 \left( X \right) \right\rbrace \right] \left\lbrace \left( Y^1 - Y^0 \right) - g \left( V \right) \right\rbrace^2 \cond X \right) \right\rbrace \\
&= \mathbb{E} \left\lbrace \mathbb{E} \left( \left[ \left\lbrace \pi_0 \left( X \right) - \pi_0 \left( X \right) \right\rbrace \lambda^\prime \left\lbrace \pi_0 \left( X \right) \right\rbrace + \lambda \left\lbrace \pi_0 \left( X \right) \right\rbrace \right] \left\lbrace \left( Y^1 - Y^0 \right) - g \left( V \right) \right\rbrace^2 \cond X \right) \right\rbrace \\
&= \mathbb{E} \left[ \lambda \left\lbrace \pi_0 \left( X \right) \right\rbrace \left\lbrace \left( Y^1 - Y^0 \right) - g \left( V \right) \right\rbrace^2 \right],
\end{align*}
where the second equality follows by the conditional exchangeability (i.e. Assumption \ref{Assump1}). \\
Therefore, we have shown that 
\begin{align*}
&\mathbb{E} \left( \left[ \left\lbrace A - \pi_0 \left( X \right) \right\rbrace \lambda^\prime \left\lbrace \pi_0 \left( X \right) \right\rbrace + \lambda \left\lbrace \pi_0 \left( X \right) \right\rbrace \right] \left\lbrace \varphi \left( Z; \eta_0, \lambda \left( \pi_0 \right) \right) - g \left( V \right) \right\rbrace^2 \right) \\
&= \mathbb{E} \left[ \lambda \left\lbrace \pi_0 \left( X \right) \right\rbrace \left\lbrace \left( Y^1 - Y^0 \right) - g \left( V \right) \right\rbrace^2 \right] \\
&+ \mathbb{E} \left( \left[ \left\lbrace A - \pi_0 \left( X \right) \right\rbrace \lambda^\prime \left\lbrace \pi_0 \left( X \right) \right\rbrace + \lambda \left\lbrace \pi_0 \left( X \right) \right\rbrace \right] \left\lbrace \varphi \left( Z; \eta_0, \lambda \left( \pi_0 \right) \right) - \left( Y^1 - Y^0 \right) \right\rbrace^2 \right) \\
&+ 2 \mathbb{E} \left( \left[ \left\lbrace A - \pi_0 \left( X \right) \right\rbrace \lambda^\prime \left\lbrace \pi_0 \left( X \right) \right\rbrace + \lambda \left\lbrace \pi_0 \left( X \right) \right\rbrace \right] \left\lbrace \varphi \left( Z; \eta_0, \lambda \left( \pi_0 \right) \right) - \left( Y^1 - Y^0 \right) \right\rbrace \left( Y^1 - Y^0 \right) \right).
\end{align*}
Note that the second and third term do not depend on $g$, hence do not have impact on the minimization problem.  Therefore, it follows that the loss functions (\ref{eqn5}) and (\ref{eqn12}) have the same minimizer. 
\end{proof}

\subsection{Proof of Theorem \ref{theorem1}}
\begin{proof}
Using a second-order Taylor expansion on the risk at $\hat{\eta}$, there exists $\overline{g} \in \text{star} \left( \mathcal{G}, g_0 \right)$ such that
\begin{align*} 
\mathcal{L} \left( \hat{g}, \hat{\eta}, \lambda \left\lbrace \hat{\pi} \left( X \right) \right\rbrace \right) &= \mathcal{L} \left( g_0, \hat{\eta}, \lambda \left\lbrace \hat{\pi} \left( X \right) \right\rbrace \right) + D_g \mathcal{L} \left( g_0, \hat{\eta}, \lambda \left\lbrace \hat{\pi} \left( X \right) \right\rbrace \right) \left[ \hat{g} - g_0 \right] \\
&+ \frac{1}{2} D^2_g \mathcal{L} \left( \bar{g}, \hat{\eta}, \lambda \left\lbrace \hat{\pi} \left( X \right) \right\rbrace \right) \left[ \hat{g} - g_0, \hat{g} - g_0  \right].
\end{align*} 
Therefore, we obtain
\begin{align} \label{eqn16}
\frac{1}{2} D^2_g \mathcal{L} \left( \bar{g}, \hat{\eta}, \lambda \left\lbrace \hat{\pi} \left( X \right) \right\rbrace \right) \left[ \hat{g} - g_0, \hat{g} - g_0  \right] &= \mathcal{L} \left( \hat{g}, \hat{\eta}, \lambda \left\lbrace \hat{\pi} \left( X \right) \right\rbrace \right) - \mathcal{L} \left( g_0, \hat{\eta}, \lambda \left\lbrace \hat{\pi} \left( X \right) \right\rbrace \right) \\
\nonumber
&- D_g \mathcal{L} \left( g_0, \hat{\eta}, \lambda \left\lbrace \hat{\pi} \left( X \right) \right\rbrace \right) \left[ \hat{g} - g_0 \right].
\end{align} 
Furthermore, we have:
\begin{align*} 
&D_g^2 \mathcal{L} \left( \overline{g}, \eta, \lambda \left\lbrace \pi \left( X \right) \right\rbrace \right) \left[ g - g_0, g - g_0 \right] \\
&=  \frac{\partial^2}{\partial t_1 \partial t_2} \mathcal{L} \left[ \overline{g} \left( V \right)  + t_1 \left\lbrace g \left( V \right) - g_0 \left( V \right) \right\rbrace + t_2 \left\lbrace g \left( V \right) - g_0 \left( V \right) \right\rbrace , \eta, \lambda \left\lbrace \pi \left( X \right) \right\rbrace \right] \Bigr|_{t_1=t_2=0} \\
&= \frac{\partial^2}{\partial t_1 \partial t_2} \left[ \mathbb{E} \left\lbrace \left[ \left\lbrace A - \pi \left( X \right) \right\rbrace \lambda^\prime \left\lbrace \pi \left( X \right) \right\rbrace + \lambda \left\lbrace \pi \left( X \right) \right\rbrace \right] \right.\right.\\
&\left.\left. \times \left( \varphi \left( Z; \eta, \lambda \left( \pi \right) \right) - \left[ \overline{g} \left( V \right) + t_1 \left\lbrace g \left( V \right) - g_0 \left( V \right) \right\rbrace + t_2 \left\lbrace g \left( V \right) - g_0 \left( V \right) \right\rbrace \right] \right)^2 \right\rbrace \right] \Bigr|_{t_1=t_2=0} \\
&= 2 \mathbb{E} \left( \left[ \left\lbrace A - \pi \left( X \right) \right\rbrace \lambda^\prime \left\lbrace \pi \left( X \right) \right\rbrace + \lambda \left\lbrace \pi \left( X \right) \right\rbrace \right] \left\lbrace g \left( V \right) - g_0 \left( V \right) \right\rbrace^2 \right).
\end{align*}
Therefore, using condition (\ref{eqn8}) we have
\begin{align*} 
D^2_g \mathcal{L} \left( \bar{g}, \hat{\eta}, \lambda \left\lbrace \hat{\pi} \left( X \right) \right\rbrace \right) \left[ \hat{g} - g_0, \hat{g} - g_0  \right] \geq \alpha \norm{\hat{g} - g_0}_\mathcal{G}^2.
\end{align*} 
Therefore, we obtain from (\ref{eqn16}) the following
\begin{align} \label{eqn17}
\nonumber
\frac{\alpha}{2} \norm{\hat{g} - g_0}_\mathcal{G}^2 &\leq \mathcal{L} \left( \hat{g}, \hat{\eta}, \lambda \left\lbrace \hat{\pi} \left( X \right) \right\rbrace \right) - \mathcal{L} \left( g_0, \hat{\eta}, \lambda \left\lbrace \hat{\pi} \left( X \right) \right\rbrace \right) - D_g \mathcal{L} \left( g_0, \hat{\eta}, \lambda \left\lbrace \hat{\pi} \left( X \right) \right\rbrace \right) \left[ \hat{g} - g_0 \right] \\
&= R_g - D_g \mathcal{L} \left( g_0, \hat{\eta}, \lambda \left\lbrace \hat{\pi} \left( X \right) \right\rbrace \right) \left[ \hat{g} - g_0 \right].
\end{align} 
Furthermore, we apply a second-order Taylor-expansion, which implies there exists $\overline{\eta} \in \text{star} \left(\mathcal{H}, \eta_0 \right)$ such that
\begin{align*} 
- D_g \mathcal{L} \left( g_0, \hat{\eta}, \lambda \left\lbrace \hat{\pi} \left( X \right) \right\rbrace \right) \left[ \hat{g} - g_0 \right] &= - D_g \mathcal{L} \left( g_0, \eta_0, \lambda \left\lbrace \pi_0 \left( X \right) \right\rbrace \right) \left[ \hat{g} - g_0 \right] \\
&- D_\eta D_g \mathcal{L} \left( g_0, \eta_0, \lambda \left\lbrace \pi_0 \left( X \right) \right\rbrace \right) \left[ \hat{g} - g_0, \hat{\eta} - \eta_0 \right] \\
&- \frac{1}{2} D^2_\eta D_g \mathcal{L} \left( g_0, \overline{\eta}, \lambda \left\lbrace \overline{\pi} \left( X \right) \right\rbrace \right) \left[ \hat{g} - g_0, \hat{\eta} - \eta_0, \hat{\eta} - \eta_0 \right].
\end{align*} 
Using Neyman-orthogonality of the loss function (\ref{eqn5}) (see Lemma \ref{lemma4}) we have
\begin{align*}
D_\eta D_g \mathcal{L} \left( g_0, \eta_0, \lambda \left\lbrace \pi_0 \left( X \right) \right\rbrace \right) \left[ \hat{g} - g_0, \hat{\eta} - \eta_0 \right] = 0.
\end{align*} 
Therefore, we obtain 
\begin{align*} 
- D_g \mathcal{L} \left( g_0, \hat{\eta}, \lambda \left\lbrace \hat{\pi} \left( X \right) \right\rbrace \right) \left[ \hat{g} - g_0 \right] &= - D_g \mathcal{L} \left( g_0, \eta_0, \lambda \left\lbrace \pi_0 \left( X \right) \right\rbrace \right) \left[ \hat{g} - g_0 \right] \\
&- \frac{1}{2} D^2_\eta D_g \mathcal{L} \left( g_0, \overline{\eta}, \lambda \left\lbrace \overline{\pi} \left( X \right) \right\rbrace \right) \left[ \hat{g} - g_0, \hat{\eta} - \eta_0, \hat{\eta} - \eta_0 \right].
\end{align*} 
Combining the above results we obtain from (\ref{eqn17}) the following:
\begin{align} \label{eqn18}
\frac{\alpha}{2} \norm{\hat{g} - g_0}_\mathcal{G}^2 & \leq R_g  - D_g \mathcal{L} \left( g_0, \eta_0, \lambda \left\lbrace \pi_0 \left( X \right) \right\rbrace \right) \left[ \hat{g} - g_0 \right] \\
\nonumber
&- \frac{1}{2} D^2_\eta D_g \mathcal{L} \left( g_0, \overline{\eta}, \lambda \left\lbrace \overline{\pi} \left( X \right) \right\rbrace \right) \left[ \hat{g} - g_0, \hat{\eta} - \eta_0, \hat{\eta} - \eta_0 \right]. 
\end{align} 
Now, we need to calculate the second-order derivative with respect to the nuisance parameters. The gradient calculations can be found in Appendix \ref{AppB}. Here, $\nabla_g$ denotes the derivative with respect to the target parameter $g$ and $\nabla_\eta$ denotes the derivative with respect to the nuisance parameter $\eta$.
\begin{align} \label{eqn19}
\nonumber
& D_{\eta}^2 D_g \mathcal{L} \left( g_0, \overline{\eta}, \lambda \left\lbrace \overline{\pi} \left( X \right) \right\rbrace \right) \left[ g - g_0,  \eta - \eta_0, \eta - \eta_0 \right]  \\
\nonumber
&= \mathbb{E} \left[ \left( \eta \left( X \right) - \eta_0 \left( X \right) \right)^\intercal \nabla_{\eta \eta}^2 \nabla_g l \left( g_0, \overline{\eta}, \lambda \left\lbrace \overline{\pi} \left( X \right) \right\rbrace \right) \left( \eta \left( X \right) - \eta_0 \left( X \right) \right) \left( g \left( V \right) - g_0 \left( V \right) \right) \right] \\
\nonumber
&= \mathbb{E} \left( \begin{bmatrix}
\pi \left( X \right) - \pi_0 \left( X \right) & Q^{\left( 1 \right)} \left( X \right) - Q^{\left( 1 \right)}_0 \left( X \right) & Q^{\left( 0 \right)} \left( X \right) - Q^{\left( 0 \right)}_0 \left( X \right)
\end{bmatrix} 
\nabla_{\eta \eta}^2 \nabla_g l \left( g_0, \overline{\eta}, \lambda \left\lbrace \overline{\pi} \left( X \right) \right\rbrace \right) \right. \\
\nonumber
&\left. \times
\begin{bmatrix}
\pi \left( X \right) - \pi_0 \left( X \right) \\
Q^{\left( 1 \right)} \left( X \right) - Q^{\left( 1 \right)}_0 \left( X \right) \\
Q^{\left( 0 \right)} \left( X \right) - Q^{\left( 0 \right)}_0 \left( X \right)
\end{bmatrix} 
\left\lbrace g \left( V \right) - g_0 \left( V \right) \right\rbrace \right) \\
\nonumber
&= - 4 \mathbb{E} \left\lbrace \left( \lambda \left\lbrace \overline{\pi} \left( X \right) \right\rbrace \left[ \frac{A}{\overline{\pi}^3 \left( X \right)} \left\lbrace Y - \overline{Q}^{\left( 1 \right)} \left( X \right) \right\rbrace - \frac{ 1-A }{\left\lbrace 1-\overline{\pi} \left( X \right) \right\rbrace^3 } \left\lbrace Y - \overline{Q}^{\left( 0 \right)} \left( X \right) \right\rbrace \right]  \right. \right.  \\
\nonumber
&- \lambda^\prime \left\lbrace \overline{\pi} \left( X \right) \right\rbrace \left[ \frac{A}{\overline{\pi}^2 \left( X \right)} \left\lbrace Y - \overline{Q}^{\left( 1 \right)} \left( X \right) \right\rbrace + \frac{1-A}{\left\lbrace 1-\overline{\pi} \left( X \right) \right\rbrace^2} \left\lbrace Y - \overline{Q}^{\left( 0 \right)} \left( X \right) \right\rbrace \right] \\
\nonumber
&+ \frac{1}{2} \lambda^{\prime\prime} \left\lbrace \overline{\pi} \left( X \right) \right\rbrace \left[ \frac{A}{\overline{\pi} \left( X \right)} \left\lbrace Y - \overline{Q}^{\left( 1 \right)} \left( X \right) \right\rbrace - \overline{Q}^{\left( 1 \right)} \left( X \right) - \frac{1-A}{1-\overline{\pi} \left( X \right)} \left\lbrace Y - \overline{Q}^{\left( 0 \right)} \left( X \right) \right\rbrace + \overline{Q}^{\left( 0 \right)} \left( X \right) + g_0 \left( V \right) \right] \\
\nonumber
&\left.\left. + \frac{1}{2} \lambda^{\prime\prime\prime} \left\lbrace \overline{\pi} \left( X \right) \right\rbrace  \left\lbrace A - \overline{\pi} \left( X \right) \right\rbrace \left\lbrace \overline{Q}^{\left( 1 \right)} \left( X \right) - \overline{Q}^{\left( 0 \right)} \left( X \right) - g_0 \left( V \right) \right\rbrace \right) \left\lbrace \pi \left( X \right) - \pi_0 \left( X \right) \right\rbrace^2 \left\lbrace g \left( V \right) - g_0 \left( V \right) \right\rbrace \right\rbrace \\
\nonumber
&- 4 \mathbb{E} \left( \left[ \lambda \left\lbrace \overline{\pi} \left( X \right) \right\rbrace \frac{A}{\overline{\pi}^2 \left( X \right)} - \lambda^\prime \left\lbrace \overline{\pi} \left( X \right) \right\rbrace \frac{A}{\overline{\pi} \left( X \right)} + \left\lbrace A - \overline{\pi} \left( X \right) \right\rbrace \lambda^{\prime\prime} \left\lbrace \overline{\pi} \left( X \right) \right\rbrace \right] \left\lbrace \pi \left( X \right) - \pi_0 \left( X \right) \right\rbrace \right. \\
\nonumber
&\left. \times \left\lbrace Q^{\left( 1 \right)} \left( X \right) - Q^{\left( 1 \right)}_0 \left( X \right) \right\rbrace \left\lbrace g \left( V \right) - g_0 \left( V \right) \right\rbrace \right) \\
\nonumber
&- 4 \mathbb{E} \left( \left[ \lambda \left\lbrace \overline{\pi} \left( X \right) \right\rbrace  \frac{ 1-A }{ \left\lbrace 1-\overline{\pi} \left( X \right) \right\rbrace^2} + \lambda^\prime \left\lbrace \overline{\pi} \left( X \right) \right\rbrace \frac{1-A}{1-\overline{\pi} \left( X \right)} - \left\lbrace A - \overline{\pi} \left( X \right) \right\rbrace \lambda^{\prime\prime} \left\lbrace \overline{\pi} \left( X \right) \right\rbrace \right] \right.\\
\nonumber
&\left. \times \left\lbrace \pi \left( X \right) - \pi_0 \left( X \right) \right\rbrace \left\lbrace Q^{\left( 0 \right)} \left( X \right) - Q^{\left( 0 \right)}_0 \left( X \right) \right\rbrace \left\lbrace g \left( V \right) - g_0 \left( V \right) \right\rbrace \right) \\
\nonumber
&= - 4 \mathbb{E} \left\lbrace \left( \lambda \left\lbrace \overline{\pi} \left( X \right) \right\rbrace \left[ \frac{\pi_0 \left( X \right)}{\overline{\pi}^3 \left( X \right)} \left\lbrace Q_0^{\left( 1 \right)} \left( X \right) - \overline{Q}^{\left( 1 \right)} \left( X \right) \right\rbrace - \frac{ 1-\pi_0 \left( X \right) }{\left\lbrace 1-\overline{\pi} \left( X \right) \right\rbrace^3 } \left\lbrace Q_0^{\left( 0 \right)} \left( X \right) - \overline{Q}^{\left( 0 \right)} \left( X \right) \right\rbrace \right]  \right. \right.  \\
\nonumber
&- \lambda^\prime \left\lbrace \overline{\pi} \left( X \right) \right\rbrace \left[ \frac{\pi_0 \left( X \right)}{\overline{\pi}^2 \left( X \right)} \left\lbrace Q_0^{\left( 1 \right)} \left( X \right) - \overline{Q}^{\left( 1 \right)} \left( X \right) \right\rbrace + \frac{1-\pi_0 \left( X \right)}{\left\lbrace 1-\overline{\pi} \left( X \right) \right\rbrace^2} \left\lbrace Q_0^{\left( 0 \right)} \left( X \right) - \overline{Q}^{\left( 0 \right)} \left( X \right) \right\rbrace \right] \\
\nonumber
&+ \frac{1}{2} \lambda^{\prime\prime} \left\lbrace \overline{\pi} \left( X \right) \right\rbrace \left[ \frac{\pi_0 \left( X \right)}{\overline{\pi} \left( X \right)} \left\lbrace Q_0^{\left( 1 \right)} \left( X \right) - \overline{Q}^{\left( 1 \right)} \left( X \right) \right\rbrace - \overline{Q}^{\left( 1 \right)} \left( X \right) \right.\\
\nonumber
&\left.- \frac{1-\pi_0 \left( X \right)}{1-\overline{\pi} \left( X \right)} \left\lbrace Q_0^{\left( 0 \right)} \left( X \right) - \overline{Q}^{\left( 0 \right)} \left( X \right) \right\rbrace + \overline{Q}^{\left( 0 \right)} \left( X \right) + g_0 \left( V \right) \right] \\
\nonumber
&\left.\left. + \frac{1}{2} \lambda^{\prime\prime\prime} \left\lbrace \overline{\pi} \left( X \right) \right\rbrace  \left\lbrace \pi_0 \left( X \right) - \overline{\pi} \left( X \right) \right\rbrace \left\lbrace \overline{Q}^{\left( 1 \right)} \left( X \right) - \overline{Q}^{\left( 0 \right)} \left( X \right) - g_0 \left( V \right) \right\rbrace \right) \left\lbrace \pi \left( X \right) - \pi_0 \left( X \right) \right\rbrace^2 \left\lbrace g \left( V \right) - g_0 \left( V \right) \right\rbrace \right\rbrace \\
\nonumber
&- 4 \mathbb{E} \left( \left[ \lambda \left\lbrace \overline{\pi} \left( X \right) \right\rbrace \frac{\pi_0 \left( X \right)}{\overline{\pi}^2 \left( X \right)} - \lambda^\prime \left\lbrace \overline{\pi} \left( X \right) \right\rbrace \frac{\pi_0 \left( X \right)}{\overline{\pi} \left( X \right)} + \left\lbrace \pi_0 \left( X \right) - \overline{\pi} \left( X \right) \right\rbrace \lambda^{\prime\prime} \left\lbrace \overline{\pi} \left( X \right) \right\rbrace \right] \left\lbrace \pi \left( X \right) - \pi_0 \left( X \right) \right\rbrace \right. \\
\nonumber
&\left. \times \left\lbrace Q^{\left( 1 \right)} \left( X \right) - Q^{\left( 1 \right)}_0 \left( X \right) \right\rbrace \left\lbrace g \left( V \right) - g_0 \left( V \right) \right\rbrace \right) \\
\nonumber
&- 4 \mathbb{E} \left( \left[ \lambda \left\lbrace \overline{\pi} \left( X \right) \right\rbrace  \frac{ 1-\pi_0 \left( X \right) }{ \left\lbrace 1-\overline{\pi} \left( X \right) \right\rbrace^2} + \lambda^\prime \left\lbrace \overline{\pi} \left( X \right) \right\rbrace \frac{1-\pi_0 \left( X \right)}{1-\overline{\pi} \left( X \right)} - \left\lbrace \pi_0 \left( X \right) - \overline{\pi} \left( X \right) \right\rbrace \lambda^{\prime\prime} \left\lbrace \overline{\pi} \left( X \right) \right\rbrace \right] \right.\\
\nonumber
&\left. \times \left\lbrace \pi \left( X \right) - \pi_0 \left( X \right) \right\rbrace \left\lbrace Q^{\left( 0 \right)} \left( X \right) - Q^{\left( 0 \right)}_0 \left( X \right) \right\rbrace \left\lbrace g \left( V \right) - g_0 \left( V \right) \right\rbrace \right) \\
&= - 4 \mathbb{E} \left[ C_1 \left( X \right) \left\lbrace \pi \left( X \right) - \pi_0 \left( X \right) \right\rbrace^2 \left\lbrace g \left( V \right) - g_0 \left( V \right) \right\rbrace \right] \\
\nonumber
&- 4 \mathbb{E} \left[ C_2 \left( X \right) \left\lbrace \pi \left( X \right) - \pi_0 \left( X \right) \right\rbrace \left\lbrace Q^{\left( 1 \right)} \left( X \right) - Q^{\left( 1 \right)}_0 \left( X \right) \right\rbrace \left\lbrace g \left( V \right) - g_0 \left( V \right) \right\rbrace \right] \\
\nonumber
&- 4 \mathbb{E} \left[ C_3 \left( X \right) \left\lbrace \pi \left( X \right) - \pi_0 \left( X \right) \right\rbrace \left\lbrace Q^{\left( 0 \right)} \left( X \right) - Q^{\left( 0 \right)}_0 \left( X \right) \right\rbrace \left\lbrace g \left( V \right) - g_0 \left( V \right) \right\rbrace \right]
\end{align}
where
\begin{align} \label{eqn20}
C_1 \left( X \right) &\coloneqq  \lambda \left\lbrace \overline{\pi} \left( X \right) \right\rbrace \left[ \frac{\pi_0 \left( X \right)}{\overline{\pi}^3 \left( X \right)} \left\lbrace Q_0^{\left( 1 \right)} \left( X \right) - \overline{Q}^{\left( 1 \right)} \left( X \right) \right\rbrace - \frac{ 1-\pi_0 \left( X \right) }{\left\lbrace 1-\overline{\pi} \left( X \right) \right\rbrace^3} \left\lbrace Q_0^{\left( 0 \right)} \left( X \right) - \overline{Q}^{\left( 0 \right)} \left( X \right) \right\rbrace \right]  \\
\nonumber
&- \lambda^\prime \left\lbrace \overline{\pi} \left( X \right) \right\rbrace \left[ \frac{\pi_0 \left( X \right)}{\overline{\pi}^2 \left( X \right)} \left\lbrace Q_0^{\left( 1 \right)} \left( X \right) - \overline{Q}^{\left( 1 \right)} \left( X \right) \right\rbrace + \frac{1-\pi_0 \left( X \right)}{\left\lbrace 1-\overline{\pi} \left( X \right) \right\rbrace^2} \left\lbrace Q_0^{\left( 0 \right)} \left( X \right) - \overline{Q}^{\left( 0 \right)} \left( X \right) \right\rbrace \right] \\
\nonumber
&+ \frac{1}{2} \lambda^{\prime\prime} \left\lbrace \overline{\pi} \left( X \right) \right\rbrace \left[ \frac{\pi_0 \left( X \right)}{\overline{\pi} \left( X \right)} \left\lbrace Q_0^{\left( 1 \right)} \left( X \right) - \overline{Q}^{\left( 1 \right)} \left( X \right) \right\rbrace - \overline{Q}^{\left( 1 \right)} \left( X \right) \right. \\
\nonumber
&\left.- \frac{1-\pi_0 \left( X \right)}{1-\overline{\pi} \left( X \right)} \left\lbrace Q_0^{\left( 0 \right)} \left( X \right) - \overline{Q}^{\left( 0 \right)} \left( X \right) \right\rbrace + \overline{Q}^{\left( 0 \right)} \left( X \right) + g_0 \left( V \right) \right] \\
\nonumber
&+ \frac{1}{2} \lambda^{\prime\prime\prime} \left\lbrace \overline{\pi} \left( X \right) \right\rbrace  \left\lbrace \pi_0 \left( X \right) - \overline{\pi} \left( X \right) \right\rbrace \left\lbrace \overline{Q}^{\left( 1 \right)} \left( X \right) - \overline{Q}^{\left( 0 \right)} \left( X \right) - g_0 \left( V \right) \right\rbrace
\end{align}
\begin{align} \label{eqn21}
C_2 \left( X \right) \coloneqq  \lambda \left\lbrace \overline{\pi} \left( X \right) \right\rbrace \frac{\pi_0 \left( X \right)}{\overline{\pi}^2 \left( X \right)} - \lambda^\prime \left\lbrace \overline{\pi} \left( X \right) \right\rbrace \frac{\pi_0 \left( X \right)}{\overline{\pi} \left( X \right)} + \left\lbrace \pi_0 \left( X \right) - \overline{\pi} \left( X \right) \right\rbrace \lambda^{\prime\prime} \left\lbrace \overline{\pi} \left( X \right) \right\rbrace
\end{align}
\begin{align} \label{eqn22}
C_3 \left( X \right) \coloneqq \lambda \left\lbrace \overline{\pi} \left( X \right) \right\rbrace  \frac{ 1-\pi_0 \left( X \right) }{ \left\lbrace 1-\overline{\pi} \left( X \right) \right\rbrace^2} + \lambda^\prime \left\lbrace \overline{\pi} \left( X \right) \right\rbrace \frac{1-\pi_0 \left( X \right)}{1-\overline{\pi} \left( X \right)} - \left\lbrace \pi_0 \left( X \right) - \overline{\pi} \left( X \right) \right\rbrace \lambda^{\prime\prime} \left\lbrace \overline{\pi} \left( X \right) \right\rbrace.
\end{align}
Based on the above calculation (\ref{eqn19}) we have the following:
\begin{align*} 
&- \frac{1}{2} D^2_\eta D_g \mathcal{L} \left( g_0, \overline{\eta}, \lambda \left\lbrace \overline{\pi} \left( X \right) \right\rbrace \right) \left[ \hat{g} - g_0, \hat{\eta} - \eta_0, \hat{\eta} - \eta_0 \right] \\
&= 2 \mathbb{E} \left[ C_1 \left( X \right) \left\lbrace \hat{\pi} \left( X \right) - \pi_0 \left( X \right) \right\rbrace^2 \left\lbrace \hat{g} \left( V \right) - g_0 \left( V \right) \right\rbrace \right] \\
&+ 2 \mathbb{E} \left[ C_2 \left( X \right) \left\lbrace \hat{\pi} \left( X \right) - \pi_0 \left( X \right) \right\rbrace \left\lbrace \hat{Q}^{\left( 1 \right)} \left( X \right) - Q^{\left( 1 \right)}_0 \left( X \right) \right\rbrace \left\lbrace \hat{g} \left( V \right) - g_0 \left( V \right) \right\rbrace \right] \\
&+ 2 \mathbb{E} \left[ C_3 \left( X \right) \left\lbrace \hat{\pi} \left( X \right) - \pi_0 \left( X \right) \right\rbrace \left\lbrace \hat{Q}^{\left( 0 \right)} \left( X \right) - Q^{\left( 0 \right)}_0 \left( X \right) \right\rbrace \left\lbrace \hat{g} \left( V \right) - g_0 \left( V \right) \right\rbrace \right].
\end{align*} 
Combining the above expressions we obtain from (\ref{eqn18}) the following:
\begin{align*} 
\frac{\alpha}{2} \norm{\hat{g} - g_0}_2^2 &\leq R_g - D_g \mathcal{L} \left( g_0, \eta_0, \lambda \left\lbrace \pi_0 \left( X \right) \right\rbrace \right) \left[ \hat{g} - g_0 \right] \\
&+ 2 \mathbb{E} \left[ C_1 \left( X \right) \left\lbrace \hat{\pi} \left( X \right) - \pi_0 \left( X \right) \right\rbrace^2 \left\lbrace \hat{g} \left( V \right) - g_0 \left( V \right) \right\rbrace \right] \\
&+ 2 \mathbb{E} \left[ C_2 \left( X \right) \left\lbrace \hat{\pi} \left( X \right) - \pi_0 \left( X \right) \right\rbrace \left\lbrace \hat{Q}^{\left( 1 \right)} \left( X \right) - Q^{\left( 1 \right)}_0 \left( X \right) \right\rbrace \left\lbrace \hat{g} \left( V \right) - g_0 \left( V \right) \right\rbrace \right] \\
&+ 2 \mathbb{E} \left[ C_3 \left( X \right) \left\lbrace \hat{\pi} \left( X \right) - \pi_0 \left( X \right) \right\rbrace \left\lbrace \hat{Q}^{\left( 0 \right)} \left( X \right) - Q^{\left( 0 \right)}_0 \left( X \right) \right\rbrace \left\lbrace \hat{g} \left( V \right) - g_0 \left( V \right) \right\rbrace \right].
\end{align*} 
Applying Cauchy-Schwarz inequality to last three terms we get:
\begin{align*} 
\frac{\alpha}{2} \norm{\hat{g} - g_0}_2^2 &\leq R_g - D_g \mathcal{L} \left( g_0, \eta_0, \lambda \left\lbrace \pi_0 \left( X \right) \right\rbrace \right) \left[ \hat{g} - g_0 \right] \\
&+ 2 \sqrt{ \mathbb{E} \left[ C_1^2 \left( X \right) \left\lbrace \hat{\pi} \left( X \right) - \pi_0 \left( X \right) \right\rbrace^4 \right] } \norm{\hat{g} - g_0}_2 \\
&+ 2 \sqrt{ \mathbb{E} \left[ C_2^2 \left( X \right) \left\lbrace \hat{\pi} \left( X \right) - \pi_0 \left( X \right) \right\rbrace^2 \left\lbrace \hat{Q}^{\left( 1 \right)} \left( X \right) - Q^{\left( 1 \right)}_0 \left( X \right) \right\rbrace^2 \right] } \norm{\hat{g} - g_0}_2 \\
&+ 2 \sqrt{ \mathbb{E} \left[ C_3^2 \left( X \right) \left\lbrace \hat{\pi} \left( X \right) - \pi_0 \left( X \right) \right\rbrace^2 \left\lbrace \hat{Q}^{\left( 0 \right)} \left( X \right) - Q^{\left( 0 \right)}_0 \left( X \right) \right\rbrace^2 \right] } \norm{\hat{g} - g_0}_2.
\end{align*} 
Furthermore, using the AM-GM inequality for the last two terms, for any constants $\delta_1>0$, $\delta_2>0$ and $\delta_3>0$ we have:
\begin{align*} 
\frac{\alpha}{2} \norm{\hat{g} - g_0}_2^2 &\leq R_g - D_g \mathcal{L} \left( g_0, \eta_0, \lambda \left\lbrace \pi_0 \left( X \right) \right\rbrace \right) \left[ \hat{g} - g_0 \right] \\
&+ \frac{1}{\delta_1} \mathbb{E} \left[ C_1^2 \left( X \right) \left\lbrace \hat{\pi} \left( X \right) - \pi_0 \left( X \right) \right\rbrace^4 \right] \\
&+ \frac{1}{\delta_2}  \mathbb{E} \left[ C_2^2 \left( X \right) \left\lbrace \hat{\pi} \left( X \right) - \pi_0 \left( X \right) \right\rbrace^2 \left\lbrace \hat{Q}^{\left( 1 \right)} \left( X \right) - Q^{\left( 1 \right)}_0 \left( X \right) \right\rbrace^2 \right] \\
&+ \frac{1}{\delta_3} \mathbb{E} \left[ C_3^2 \left( X \right) \left\lbrace \hat{\pi} \left( X \right) - \pi_0 \left( X \right) \right\rbrace^2 \left\lbrace \hat{Q}^{\left( 0 \right)} \left( X \right) - Q^{\left( 0 \right)}_0 \left( X \right) \right\rbrace^2 \right]  \\
&+ \left( \delta_1 + \delta_2 + \delta_3 \right) \norm{\hat{g} - g_0}_2^2.
\end{align*} 
Therefore, we obtain (for $\delta_1 + \delta_2 + \delta_3 < \frac{\alpha}{2}$):
\begin{align} \label{eqn23}
\norm{\hat{g} - g_0}_2^2 \leq \frac{1}{\alpha/2 - \delta_1 - \delta_2 - \delta_3} &\left( R_g - D_g \mathcal{L} \left( g_0, \eta_0, \lambda \left\lbrace \pi_0 \left( X \right) \right\rbrace \right) \left[ \hat{g} - g_0 \right] \right. \\
\nonumber
&+ \frac{1}{\delta_1} \mathbb{E} \left[ C_1^2 \left( X \right) \left\lbrace \hat{\pi} \left( X \right) - \pi_0 \left( X \right) \right\rbrace^4 \right] \\
\nonumber
&+ \frac{1}{\delta_2}  \mathbb{E} \left[ C_2^2 \left( X \right) \left\lbrace \hat{\pi} \left( X \right) - \pi_0 \left( X \right) \right\rbrace^2 \left\lbrace \hat{Q}^{\left( 1 \right)} \left( X \right) - Q^{\left( 1 \right)}_0 \left( X \right) \right\rbrace^2 \right] \\
\nonumber
&\left. + \frac{1}{\delta_3} \mathbb{E} \left[ C_3^2 \left( X \right) \left\lbrace \hat{\pi} \left( X \right) - \pi_0 \left( X \right) \right\rbrace^2 \left\lbrace \hat{Q}^{\left( 0 \right)} \left( X \right) - Q^{\left( 0 \right)}_0 \left( X \right) \right\rbrace^2 \right] \right).
\end{align} 
Furthermore, since we have $D_g \mathcal{L} \left( g_0, \eta_0, \lambda \left\lbrace \pi_0 \left( X \right) \right\rbrace \right) \left[ \hat{g} - g_0 \right] \geq 0$, hence we obtain from (\ref{eqn23}) the following:
\begin{align*}
\norm{\hat{g} - g_0}_\mathcal{G}^2 \leq \frac{1}{\alpha/2 - \delta_1 - \delta_2 - \delta_3} &\left( R_g + \frac{1}{\delta_1} \norm{ C_1 \left( X \right) \left\lbrace \hat{\pi} \left( X \right) - \pi_0 \left( X \right) \right\rbrace^2 }_2^2  \right. \\
&+ \frac{1}{\delta_2} \norm{ C_2 \left( X \right) \left\lbrace \hat{\pi} \left( X \right) - \pi_0 \left( X \right) \right\rbrace \left\lbrace \hat{Q}^{\left( 1 \right)} \left( X \right) - Q^{\left( 1 \right)}_0 \left( X \right) \right\rbrace }_2^2 \\
&\left. + \frac{1}{\delta_3} \norm{ C_3 \left( X \right) \left\lbrace \hat{\pi} \left( X \right) - \pi_0 \left( X \right) \right\rbrace \left\lbrace \hat{Q}^{\left( 0 \right)} \left( X \right) - Q^{\left( 0 \right)}_0 \left( X \right) \right\rbrace }_2^2 \right).
\end{align*}
\end{proof}

\newpage

\section{Gradients calculations for the weighted DR-Learner loss function} \label{AppB}

\begin{align*} 
&l \left( g, \eta, \lambda \left\lbrace \pi \left( x \right) \right\rbrace \right) =  \left[ \left\lbrace a - \pi \left( x \right) \right\rbrace \lambda^\prime \left\lbrace \pi \left( x \right) \right\rbrace + \lambda \left\lbrace \pi \left( x \right) \right\rbrace \right] \left( \frac{\lambda \left\lbrace \pi \left( x \right) \right\rbrace }{ \left\lbrace a - \pi \left( x \right) \right\rbrace \lambda^\prime \left\lbrace \pi \left( x \right) \right\rbrace + \lambda \left\lbrace \pi \left( x \right) \right\rbrace } \right. \\
&\left. \times \left[ \frac{a}{\pi \left( x \right)} \left\lbrace y - Q^{\left( 1 \right)} \left( x \right) \right\rbrace - \frac{1-a}{1-\pi \left( x \right)} \left\lbrace y - Q^{\left( 0 \right)} \left( x \right) \right\rbrace \right] + Q^{\left( 1 \right)} \left( x \right) - Q^{\left( 0 \right)} \left( x \right) - g \left( v \right) \right)^2 
\end{align*}
\begin{align*} 
\nabla_g  l \left( g, \eta, \lambda \left\lbrace \pi \left( x \right) \right\rbrace \right) &= -2 \left( \lambda \left\lbrace \pi \left( x \right) \right\rbrace \left[ \frac{a}{\pi \left( x \right)} \left\lbrace y - Q^{\left( 1 \right)} \left( x \right) \right\rbrace - \frac{1-a}{1-\pi \left( x \right)} \left\lbrace y - Q^{\left( 0 \right)} \left( x \right) \right\rbrace \right] \right. \\
&\left. + \left[ \left\lbrace a - \pi \left( x \right) \right\rbrace \lambda^\prime \left\lbrace \pi \left( x \right) \right\rbrace + \lambda \left\lbrace \pi \left( x \right) \right\rbrace \right] \left\lbrace Q^{\left( 1 \right)} \left( x \right) - Q^{\left( 0 \right)} \left( x \right) - g \left( v \right) \right\rbrace \right)
\end{align*}
\begin{align*} 
\nabla_{Q^{\left( 1 \right)}} \nabla_g  l \left( g, \eta, \lambda \left\lbrace \pi \left( x \right) \right\rbrace \right) &= -2 \left[ \lambda \left\lbrace \pi \left( x \right) \right\rbrace \left\lbrace 1- \frac{a}{\pi \left( x \right)} \right\rbrace +  \left\lbrace a - \pi  \left( x \right) \right\rbrace \lambda^\prime \left\lbrace \pi \left( x \right) \right\rbrace \right]
\end{align*}
\begin{align*} 
\nabla_{Q^{\left( 0 \right)}} \nabla_g  l \left( g, \eta, \lambda \left\lbrace \pi \left( x \right) \right\rbrace \right) &= 2 \left[ \lambda \left\lbrace \pi \left( x \right) \right\rbrace \left\lbrace 1 - \frac{1-a}{1-\pi \left( x \right)} \right\rbrace + \left\lbrace a - \pi \left( x \right) \right\rbrace \lambda^\prime \left\lbrace \pi \left( x \right) \right\rbrace \right]
\end{align*}
\begin{align*} 
\nabla_\pi \nabla_g  l \left( g, \eta, \lambda \left\lbrace \pi \left( x \right) \right\rbrace \right) &= 2 \lambda \left\lbrace \pi \left( x \right) \right\rbrace \left[ \frac{a}{\pi^2 \left( x \right)} \left\lbrace y - Q^{\left( 1 \right)} \left( x \right) \right\rbrace + \frac{1-a}{\left\lbrace 1-\pi \left( x \right) \right\rbrace^2} \left\lbrace y - Q^{\left( 0 \right)} \left( x \right) \right\rbrace \right] \\
&-2 \lambda^\prime \left\lbrace \pi \left( x \right) \right\rbrace \left[ \frac{a}{\pi \left( x \right)} \left\lbrace y - Q^{\left( 1 \right)} \left( x \right) \right\rbrace - \frac{1-a}{1-\pi \left( x \right)} \left\lbrace y - Q^{\left( 0 \right)} \left( x \right) \right\rbrace \right] \\
&-2 \lambda^{\prime\prime} \left\lbrace \pi \left( x \right) \right\rbrace \left\lbrace a - \pi \left( x \right) \right\rbrace \left\lbrace Q^{\left( 1 \right)} \left( x \right) - Q^{\left( 0 \right)} \left( x \right) - g \left( v \right) \right\rbrace 
\end{align*}
\begin{align*} 
\nabla_{Q^{\left( 1 \right)}Q^{\left( 1 \right)}}^2 \nabla_g  l \left( g, \eta, \lambda \left\lbrace \pi \left( x \right) \right\rbrace \right) = 0
\end{align*}
\begin{align*} 
\nabla_{Q^{\left( 0 \right)}} \nabla_{Q^{\left( 1 \right)}} \nabla_g  l \left( g, \eta,  \lambda \left\lbrace \pi \left( x \right) \right\rbrace \right) = 0
\end{align*}
\begin{align*} 
\nabla_{\pi} \nabla_{Q^{\left( 1 \right)}} \nabla_g  l \left( g, \eta,  \lambda \left\lbrace \pi \left( x \right) \right\rbrace \right) = - 2 \left[ \lambda \left\lbrace \pi \left( x \right) \right\rbrace \frac{a}{\pi^2 \left( x \right)} - \lambda^\prime \left\lbrace \pi \left( x \right) \right\rbrace \frac{a}{\pi \left( x \right)} + \left\lbrace a - \pi \left( x \right) \right\rbrace \lambda^{\prime\prime} \left\lbrace \pi \left( x \right) \right\rbrace \right]
\end{align*}
\begin{align*} 
\nabla_{Q^{\left( 1 \right)}} \nabla_{Q^{\left( 0 \right)}} \nabla_g  l \left( g, \eta,  \lambda \left\lbrace \pi \left( x \right) \right\rbrace \right) = 0
\end{align*}
\begin{align*} 
\nabla_{Q^{\left( 0 \right)}Q^{\left( 0 \right)}}^2 \nabla_g  l \left( g, \eta,  \lambda \left\lbrace \pi \left( x \right) \right\rbrace \right) = 0
\end{align*}
\begin{align*} 
\nabla_{\pi} \nabla_{Q^{\left( 0 \right)}} \nabla_g  l \left( g, \eta,  \lambda \left\lbrace \pi \left( x \right) \right\rbrace \right) = -2  \left[ \lambda \left\lbrace \pi \left( x \right) \right\rbrace \frac{ 1-a }{\left\lbrace 1-\pi \left( x \right) \right\rbrace^2}  + \lambda^\prime \left\lbrace \pi \left( x \right) \right\rbrace \frac{1-a}{1-\pi \left( x \right)} - \left\lbrace a - \pi \left( x \right) \right\rbrace \lambda^{\prime\prime} \left\lbrace \pi \left( x \right) \right\rbrace \right]
\end{align*}
\begin{align*} 
\nabla_{Q^{\left( 1 \right)}} \nabla_{\pi} \nabla_g  l \left( g, \eta,  \lambda \left\lbrace \pi \left( x \right) \right\rbrace \right) = - 2 \lambda \left\lbrace \pi \left( x \right) \right\rbrace \frac{a}{\pi^2 \left( x \right)} + 2  \lambda^\prime \left\lbrace \pi \left( x \right) \right\rbrace \frac{a}{\pi \left( x \right)} - 2 \left\lbrace a - \pi \left( x \right) \right\rbrace \lambda^{\prime\prime} \left\lbrace \pi \left( x \right) \right\rbrace
\end{align*}
\begin{align*} 
\nabla_{Q^{\left( 0 \right)}} \nabla_{\pi} \nabla_g  l \left( g, \eta,  \lambda \left\lbrace \pi \left( x \right) \right\rbrace \right) &= - 2 \lambda \left\lbrace \pi \left( x \right) \right\rbrace  \frac{ 1-a }{ \left\lbrace 1-\pi \left( x \right) \right\rbrace^2} - 2 \lambda^\prime \left\lbrace \pi \left( x \right) \right\rbrace \frac{1-a}{1-\pi \left( x \right)} \\
&+ 2 \left\lbrace a - \pi \left( x \right) \right\rbrace \lambda^{\prime\prime} \left\lbrace \pi \left( x \right) \right\rbrace
\end{align*}
\begin{align*} 
&\nabla_{\pi \pi}^2 \nabla_g  l \left( g, \eta,  \lambda \left\lbrace \pi \left( x \right) \right\rbrace \right) = -4 \lambda \left\lbrace \pi \left( x \right) \right\rbrace \left[ \frac{a}{\pi^3 \left( x \right)} \left\lbrace y - Q^{\left( 1 \right)} \left( x \right) \right\rbrace - \frac{ 1-a }{\left\lbrace 1-\pi \left( x \right)\right\rbrace^3} \left\lbrace y - Q^{\left( 0 \right)} \left( x \right) \right\rbrace \right] \\ 
&+ 4 \lambda^\prime \left\lbrace \pi \left( x \right) \right\rbrace \left[ \frac{a}{\pi^2 \left( x \right)} \left\lbrace y - Q^{\left( 1 \right)} \left( x \right) \right\rbrace + \frac{1-a}{\left\lbrace 1-\pi \left( x \right) \right\rbrace^2} \left\lbrace y - Q^{\left( 0 \right)} \left( x \right) \right\rbrace \right] \\
&-2 \lambda^{\prime\prime} \left\lbrace \pi \left( x \right) \right\rbrace \left[ \frac{a}{\pi \left( x \right)} \left\lbrace y - Q^{\left( 1 \right)} \left( x \right) \right\rbrace - Q^{\left( 1 \right)} \left( x \right) - \frac{1-a}{1-\pi \left( x \right)} \left\lbrace y - Q^{\left( 0 \right)} \left( x \right) \right\rbrace + Q^{\left( 0 \right)} \left( x \right) + g \left( v \right) \right] \\
&-2 \lambda^{\prime\prime\prime} \left\lbrace \pi \left( x \right) \right\rbrace \left\lbrace a - \pi \left( x \right) \right\rbrace \left\lbrace Q^{\left( 1 \right)} \left( x \right) - Q^{\left( 0 \right)} \left( x \right) - g \left( v \right) \right\rbrace 
\end{align*}
\begin{align*} 
&\nabla_{\eta \eta}^2 \nabla_g  l \left( g, \eta,  \lambda \left\lbrace \pi \left( x \right) \right\rbrace \right) \\
& =
\begin{bmatrix}
\nabla_{\pi \pi}^2 \nabla_g  l \left( g, \eta, \lambda \left\lbrace \pi \left( x \right) \right\rbrace \right)  & \nabla_{\pi} \nabla_{Q^{\left( 1 \right)}} \nabla_g  l \left( g, \eta, \lambda \left\lbrace \pi \left( x \right) \right\rbrace \right)  & \nabla_{\pi} \nabla_{Q^{\left( 0 \right)}} \nabla_g  l \left( g, \eta, \lambda \left\lbrace \pi \left( x \right) \right\rbrace \right) \\
\nabla_{Q^{\left( 1 \right)}} \nabla_{\pi} \nabla_g  l \left( g, \eta, \lambda \left\lbrace \pi \left( x \right) \right\rbrace \right)   & \nabla_{Q^{\left( 1 \right)} Q^{\left( 1 \right)}}^2 \nabla_g  l \left( g, \eta, \lambda \left\lbrace \pi \left( x \right) \right\rbrace \right) & \nabla_{Q^{\left( 1 \right)}} \nabla_{Q^{\left( 0 \right)}} \nabla_g  l \left( g, \eta, \lambda \left\lbrace \pi \left( x \right) \right\rbrace \right) \\
\nabla_{Q^{\left( 0 \right)}} \nabla_{\pi} \nabla_g  l \left( g, \eta, \lambda \left\lbrace \pi \left( x \right) \right\rbrace \right)   & \nabla_{Q^{\left( 0 \right)}} \nabla_{Q^{\left( 1 \right)}} \nabla_g  l \left( g, \eta, \lambda \left\lbrace \pi \left( x \right) \right\rbrace \right) & \nabla_{Q^{\left( 0 \right)} Q^{\left( 0 \right)}}^2 \nabla_g  l \left( g, \eta, \lambda \left\lbrace \pi \left( x \right) \right\rbrace \right) 
\end{bmatrix}
\end{align*}

\newpage

\section{Empirical Example - Additional information}

\subsection{Overview of the weights} \label{subsect5}

\begin{figure}[h]
\begin{center}
\includegraphics[scale=0.8]{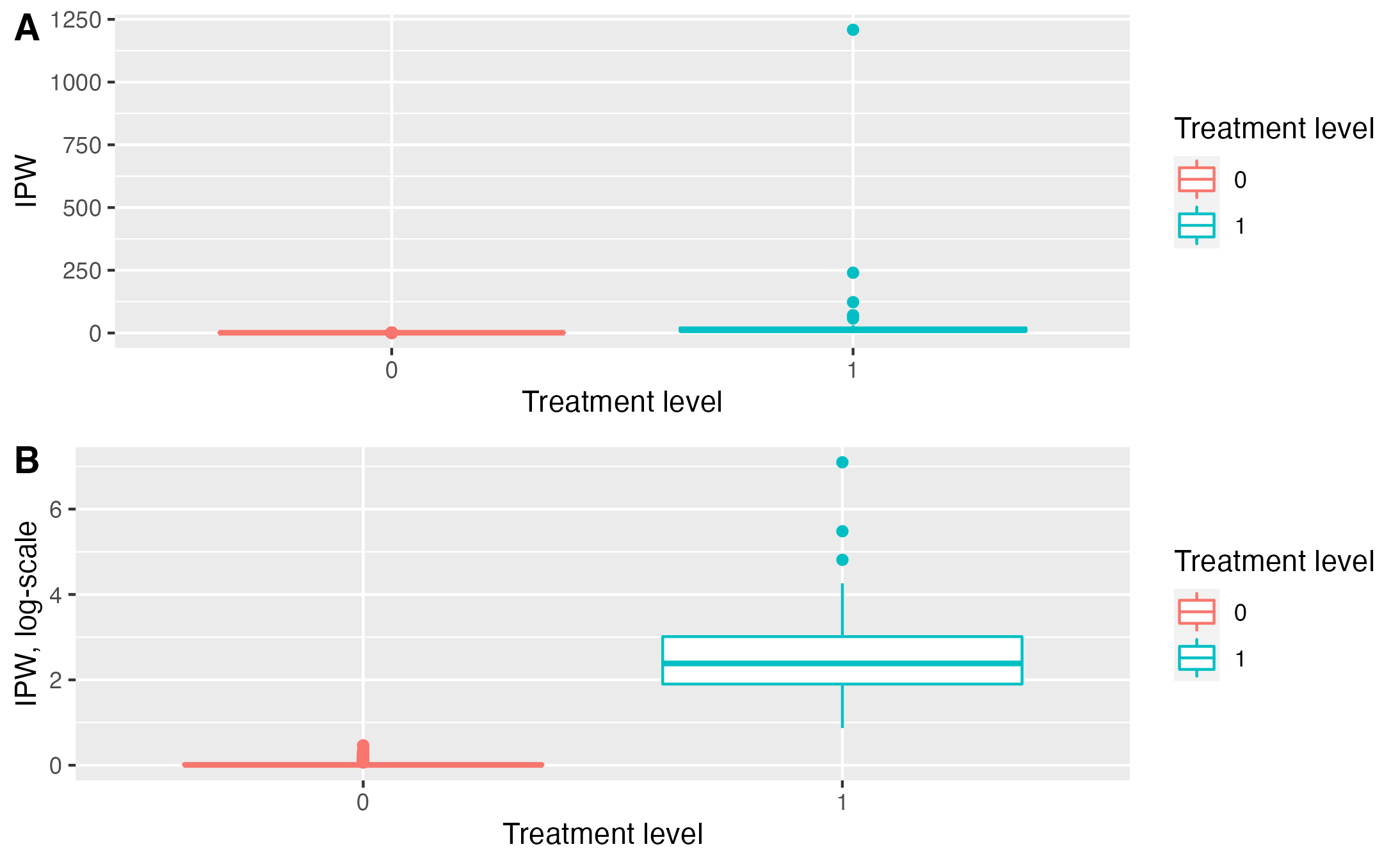}
\end{center}
\caption{Inverse-probability weights in both treatment groups. Treatment level "$1$" denotes "RRT initiation" and treatment level "$0$" denotes "no RRT initiation".} \label{figure5}
\end{figure}

\begin{figure}
\begin{center}
\includegraphics[scale=0.8]{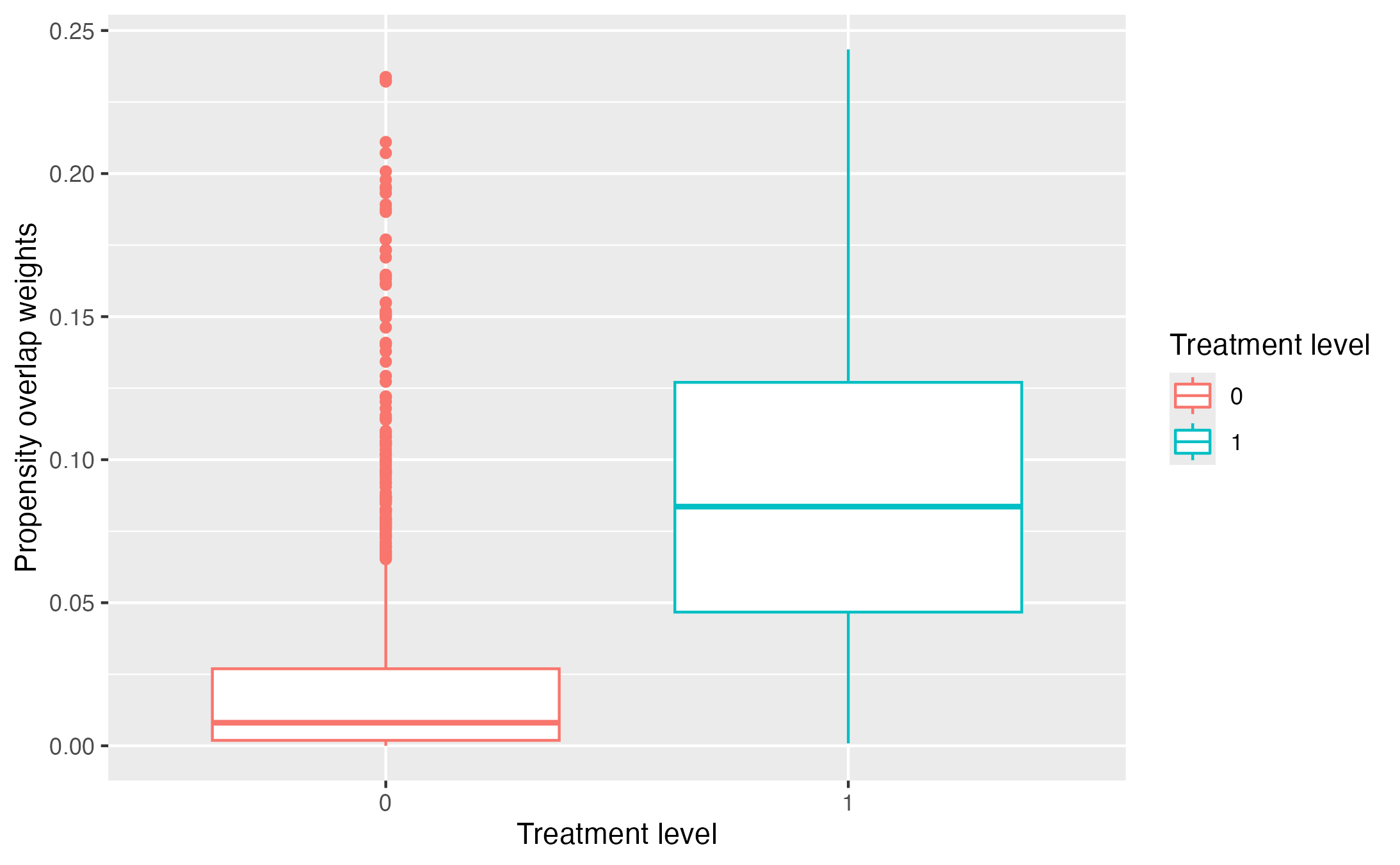}
\end{center}
\caption{Propensity-overlap weights in both treatment groups. Treatment level "$1$" denotes "RRT initiation" and treatment level "$0$" denotes "no RRT initiation".} \label{figure6}
\end{figure}

\end{document}